%% file: main.tex
\documentclass[runningheads]{llncs}

\input{packages}
\input{macro}
\input{pre_main}

\begin{document}
\title{Incorrectness Separation Logic with Arrays and Pointer Arithmetic
}
%
%

\author{Yeonseok Lee\thanks{\email{lee.yeonseok.x2@s.mail.nagoya-u.ac.jp}}  \and 
Koji Nakazawa\thanks{\email{knak@i.nagoya-u.ac.jp}}}

\authorrunning{Y. Lee and K. Nakazawa}

\institute{Nagoya University, Nagoya, Japan 
}
\maketitle              

\begin{abstract}
Incorrectness Separation Logic (ISL) is a proof system designed to automate verification and detect bugs in programs manipulating heap memories. 
In this study, we extend ISL to support variable-length array predicates and pointer arithmetic. 
Additionally, we prove the relative completeness of this extended ISL by
 constructing the weakest postconditions. 
Relative completeness means that all valid ISL triples are provable, assuming an oracle capable of checking entailment between formulas; this property ensures the reliability of the proof system.
\keywords{Program Verification \and Separation Logic \and Incorrectness Logic \and Array  \and Completeness.}
\end{abstract}

\input{1_intro}

\input{2_ISL}

\input{3_wpo}

\input{4_completeness}

\input{5_related_work}

\input{6_conclusion}

%
%
%
\bibliographystyle{splncs04}

\bibliography{bib}

\appendix
\input{7_appendix_expressiveness_arxiv}

\end{document}

%% file: packages.tex
\usepackage[dvips]{graphics}
\usepackage[utf8]{inputenc} 
\usepackage{amsmath}
\usepackage[T1]{fontenc}    
\usepackage{url}            
\usepackage{booktabs}       
\usepackage{amsfonts}       
\usepackage{nicefrac}       
\usepackage{lipsum}
\usepackage{graphicx}
\usepackage{textcomp} 
\usepackage{listings}
\usepackage{xcolor}
\usepackage{comment}
\usepackage{mathpartir} 
\usepackage[normalem]{ulem}
\usepackage{tikz}
\usepackage{algpseudocode}
\usepackage{hyperref}       
\usepackage[pass]{geometry}
\usepackage{stmaryrd} 

\usepackage{setspace}
\usepackage{pdfpages}
\usepackage[whole]{bxcjkjatype}

%% file: macro.tex
\newcommand{\bigast}{\mathop{\scalebox{1.5}{\raisebox{-0.2ex}{$*$}}}}%
\newcommand{\arr}{\normalfont\textsf{arr}}
\newcommand{\bararr}{\overline{\normalfont\textsf{arr}}}
\newcommand{\Arr}{\normalfont\textsf{Arr}}

\def\inter#1{{[\![#1]\!]}}
\def\termS{{\normalfont\textsf{termS}}}

%% file: pre_main.tex
\algnewcommand\algorithmicforeach{\textbf{for each}}
\algdef{S}[FOR]{ForEach}[1]{\algorithmicforeach\ #1\ \algorithmicdo}

%% file: 1_intro.tex
\section{Introduction}
\subsection{Background}

The automation of software verification is a critical aspect of software engineering, as it helps ensure the reliability of software systems at an economical cost.
To achieve this, researchers and practitioners often rely on {formal methods} to verify software correctness and identify incorrectness.

A classical yet influential formal method is Hoare Logic (HL) \cite{hoare1969axiomatic}, which focuses on verifying \textit{correctness}. A Hoare triple is written as
$\{P\} \ \mathbb{C} \ \{Q\}$,
where \(P\) is a precondition and \(Q\) is a postcondition. This triple states that for every state \(s\) satisfying \(P\), if the command \(\mathbb{C}\) terminates in state \(s'\), then \(s'\) satisfies \(Q\). Equivalently, if we let \(\textsf{post}(\mathbb{C}, P)\) denote the set of possible states after executing \(\mathbb{C}\) from any initial state in \(P\), then \(Q\) \textit{over-approximates} \(\textsf{post}(\mathbb{C}, P)\), i.e., \(\textsf{post}(\mathbb{C}, P) \subseteq Q\).

Recently, research has also focused on formal methods that verify \textit{incorrectness}, aiming to identify the presence of bugs. 
Reverse Hoare Logic (RHL) \cite{de2011reverse} is one such method.
The RHL triple is written as
$[P] \ \mathbb{C} \ [Q]$\footnote{In some literature, the notation  
\([P] \ \mathbb{C} \ [Q]\) is used to denote a Hoare triple representing  
total correctness, which states that if \(\mathbb{C}\) is executed in a state satisfying \(P\),  
it always terminates, and the final state satisfies \(Q\).
However, in this paper, we adopt square brackets following the conventions established in prior works \cite{o2019incorrectness,raad2020local}.},
which indicates that for every state \(s'\) in \(Q\), \(s'\) can be reached by executing \(\mathbb{C}\) from some state \(s\) in \(P\). Here, \(Q\) \textit{under-approximates} \(\textsf{post}(\mathbb{C}, P)\), i.e., \(\textsf{post}(\mathbb{C}, P) \supseteq Q\). Note that the direction of inclusion is reversed from Hoare Logic (HL) to RHL.
If $Q$ represents an unexpected state and \([P] \ \mathbb{C} \ [Q]\) is valid, we can conclude that this triple implies $\mathbb{C}$ contains a bug.

To illustrate the difference between Hoare Logic (HL) and Reverse Hoare Logic (RHL), we present the following examples, which were also introduced in our previous work \cite{lee2024relative}.
\begin{itemize}
    \item \(\{\top\} \ x := y \ \{\top\}\) is valid in HL but \([\top] \ x := y \ [\top]\) is not valid in RHL. In the RHL setting, the state \((x=1 \land y=2)\) (which satisfies \(\top\) as a postcondition) cannot be reached from a state satisfying \(\top\) by executing \(x := y\).
    
    \item \(\{\top\} \ x := y \ \{ 0<x<10 \land x=y\}\) is not valid in HL, whereas \([\top] \ x := y \ [ 0<x<10 \land x=y]\) is valid in RHL. In HL, the final state \((x=10 \land y=10)\) can be reached from \(\top\) and does not satisfy \((0<x<10 \land x=y)\), hence the triple is not valid.
\end{itemize}

Another logic focusing on under-approximation is Incorrectness Logic (IL) \cite{o2019incorrectness}, which has a similar form to RHL but embeds an additional exit condition into the postcondition:
$[P] \ \mathbb{C} \ [\epsilon: Q]$,
where \(\epsilon\) can be \textit{ok} for normal termination or \textit{er} for erroneous termination.

Separation Logic (SL), introduced by Reynolds \cite{reynolds2002separation}, builds upon HL to enable modular reasoning about pointer-manipulating programs. Its key operator, the \textit{separating conjunction} (\(*\)), partitions the heap into disjoint parts, each of which can be reasoned about independently. A prominent feature of SL is the $\textsc{Frame}$ rule:
\[
\inferrule[]
{\{P\} \ \mathbb{C} \ \{Q\}}
{\{P * R\} \ \mathbb{C} \ \{Q * R\}},
\]
which states that if \(\mathbb{C}\) does not modify the portion of memory described by \(R\), then \(R\) remains unaffected and can be ``framed'' into the pre- and postconditions. 
This locality principle significantly improves the scalability of verifications for large  programs. 
Numerous verification tools (e.g., \textit{Infer} by Meta \cite{calcagno2011infer,calcagno2015open}) have adopted SL as their theoretical basis.
Consider the following C code to demonstrate the $\textsc{Frame}$ rule.

\lstset{
  language=C,           
  basicstyle=\ttfamily, 
  numbers=left,         
  stepnumber=1,         
  frame=single,         
  tabsize=4,            
  breaklines=true       
}
\begin{lstlisting}[language=C]
int *x = malloc(sizeof(int));  *x = 1; 
int *y = malloc(sizeof(int)); *y = 2;
free(x);
\end{lstlisting}
The $\textsc{Frame}$ rule can express the situation.
\[
\inferrule[]
{\{x \mapsto 1 \} \ \texttt{free($x$)} \ \{\textsf{emp}\}}
{\{x \mapsto 1 * y \mapsto 2 \} \ \texttt{free($x$)} \ \{\textsf{emp} * y \mapsto 2\}}
\]
When we call \texttt{free($x$)}, the block pointed to by $x$ is deallocated, so the postcondition becomes $\{\textsf{emp} * y \mapsto 2\}$. 
Crucially, since the operation \texttt{free(x)} does not affect the
block described by $y$, the $\textsc{Frame}$ rule ensures that $y$
remains unchanged.
This allows us to ``frame'' the $y$-part of the heap into the pre- and postconditions without requiring additional reasoning.

By drawing upon these ideas, Incorrectness Separation Logic (ISL) was introduced by Raad et al. \cite{raad2020local} to extend Incorrectness Logic (IL) with heap-manipulation features similar to those in Separation Logic (SL). 
They demonstrated that the naive application of the $\textsc{Frame}$ rule is unsound in the IL setting.
The following example from \cite{raad2020local} illustrates this problem.
\[
\inferrule*
{[x \mapsto - ] \ \texttt{free($x$)} \ [\textit{ok}:\textsf{emp}] }
{[x \mapsto - * x \mapsto - ] \ \texttt{free($x$)} \ [\textit{ok}:\textsf{emp} * x \mapsto -]}
\]
The premise is valid, but the conclusion is not,
since the conclusion means that every state satisfying the postcondition $\textsf{emp} * x \mapsto -$ can be reached from some state satisfying the precondition \(x \mapsto - * x \mapsto -\), which is unsatisfiable.
As a result, the inference is unsound.
However, this issue can be resolved by introducing the negative heap assertion \(x \not\mapsto\), as demonstrated in \cite{raad2020local}.
This assertion indicates that \(x\) was once part of the heap but has since been deallocated.
\[
\inferrule*
{[x \mapsto - ] \ \texttt{free($x$)} \ [ \textit{ok}: x \not\mapsto] }
{[x \mapsto - * x \mapsto - ] \ \texttt{free($x$)} \ [ \textit{ok}: x \not\mapsto * x \mapsto -]}
\]
In this corrected version, both triples are valid. The precondition and postcondition of the conclusion are false, making the triple trivially valid.

Both soundness and completeness are essential for logical proof systems to establish reliability.  
While \textit{soundness} guarantees that every provable statement is valid, \textit{completeness} ensures that all valid statements can be proven. 
However, Hoare Logic is famously known to be incomplete. 
To address this, Cook introduced the notion of \textit{relative completeness} \cite{cook1978soundness}.  
Relative completeness states that the logic is complete if an oracle exists to decide the validity of entailments.
Thus, if a proof system is relatively complete, it can have strong reliability.

\subsection{Our Contribution: ISL with Arrays and Pointer Arithmetic}
In this paper, we introduce an extension of ISL by adding support for arrays and pointer arithmetic. 
Consider the following C code, which contains an error related to array handling.
\begin{lstlisting}
int *array = (int *)malloc(10 * sizeof(int));
free(array + 1);    // Error: Free the memory at a non-head position
\end{lstlisting}
This code has an error because it attempts to free a heap cell that is not the head of the allocated block.  
Existing ISL systems \cite{raad2020local,lee2024relative} cannot describe this scenario due to their lack of support for array predicates.
To address this limitation, in this paper, we introduce an extension of ISL to handle such cases.  
Our contributions are as follows.
\begin{itemize}
    \item[$\bullet$] \textbf{Introducing ISL with arrays and pointer arithmetic (Section \ref{sec: body-ISL})}  
    
We extend the existing ISL framework \cite{raad2020local,lee2024relative} by incorporating array predicates and pointer arithmetic.  
Our framework builds upon two existing Separation Logics: traditional array-based SL \cite{brotherston2017biabduction,kimura2021decidability} and SL with arrays \cite{holik2022low}, where models incorporate memory blocks. 
Specifically, our approach utilizes array predicates with variable-size lengths \cite{brotherston2017biabduction,kimura2021decidability}, allowing for the description of more flexible arrays compared to predicates that specify only fixed lengths and starting points.

    \item[$\bullet$] \textbf{Defining weakest postcondition generating functions (Section \ref{sec: body-wpo})}  
    
Following the methods of \cite{de2011reverse,lee2024relative}, 
we define a function constructing the weakest postconditions
and prove the \textit{expressiveness}, a property that ensures these functions accurately describe the sets expressing the weakest postconditions.  
This procedure is essential for proving relative completeness.  
It is worth noting that while Cook used weakest \textit{preconditions} to prove the relative completeness of Hoare Logic (HL) \cite{cook1978soundness}, our approach focuses on weakest \textit{postconditions}. 
This is because HL relies on over-approximation, whereas Reverse HL and ISL use under-approximation.

Our work extends this previous approach by proving expressiveness for heap manipulation rules  
involving arrays and pointer arithmetic.
For other rules that do not involve heap manipulation, we rely on the results of \cite{lee2024relative}, as these rules yield the same outcomes in both settings.

    \item[$\bullet$] \textbf{Proving relative completeness (Section \ref{sec: body-completeness})}  
    
We prove the relative completeness of this extended ISL by following a method similar to that in \cite{lee2024relative}.

\end{itemize}

%% file: 2_ISL.tex
\section{Incorrectness Separation Logic}\label{sec: body-ISL}
We define a modified version of ISL, extending ISL \cite{raad2020local,lee2024relative}, to handle arrays and pointer arithmetic.
We assume 
a countably infinite set of variables, \(\textsc{Var}\). 
We use the metavariables \(x, y, z, \ldots\) to denote variables
and \(n\) to denote elements of natural numbers.

\subsection{Assertions}

First, we introduce assertions in ISL, which are potentially infinite disjunctions of existentially quantified symbolic heaps.  
Quantifier-free symbolic heaps consist of equalities (\(x \approx y\)) and inequalities (\(x \not\approx y\)) between terms,  
as well as weak and strict inequalities (\(x \leq y\) and \(x < y\), respectively).
They also include heap predicates such as singleton heaps (\(t \mapsto u\)), variable-length array predicates (\(\arr(t, t')\)), and negative array predicates (\(\bararr(t, t')\)).

The negative array predicates (\(\bararr(t, t')\)) represent a set of consecutive negative heaps with sequential addresses. 
This extends the concept of negative heaps (\(x \not\mapsto\)) from the original ISL \cite{raad2020local}, 
adapting it to handle arrays with consecutive memory locations.  
\(x \not\mapsto\) is an abbreviation of \(\bararr(x, x+1)\) in our setting.

Symbolic heaps are connected using the separating conjunction (\(*\)). The logical operators \(\lor\), \(\land\), and \(\lnot\) are not permitted within symbolic heaps.

\begin{definition}[Assertions of ISL]\normalfont 

An assertion $P$ is defined as follows, where $I$ is finite or
 countably infinite index set.

\begin{align*}
    P ::= \ & \bigvee_{i \in I} \exists \overrightarrow{x_i} . \psi_i
    & {\text{Top-level}}\\
    \psi ::= \ &  \psi *  \psi   & {\text{Symbolic Heaps}}\\
    \ &  | \ {\normalfont\textsf{emp}}  
    \ | \ t \mapsto u
    \ | \ \arr (t,t') 
    \ | \ \bararr (t,t') 
    & {\text{Spatial Formulas}}\\
    &| \ t \approx t' \ | \ t \not\approx t' \ | \ t \leq t' \ | \ t < t' & {\text{Pure Formulas}}\\
    t ::= \ & n \ | \ x  \ | \ t + t' \ | \ b (t) \ | \ e (t) \ | \ {\normalfont\texttt{null}}  & {\text{Terms}} 
\end{align*}

\end{definition}

Here, we employ vector notation to represent a sequence. For instance, \(\overrightarrow{x_i}\) denotes the sequence \(x_{i1}, \ldots, x_{in}\).  
When \(I = \emptyset\),  
\(\bigvee_{i \in I} \exists\overrightarrow{x_i}.\psi_i\) is denoted by \(\normalfont\textsf{false}\).

We ignore the order of disjuncts in assertions and atomic formulas in symbolic heaps.  
Specifically, we identify the formulas \(\psi_1 * \psi_2\) and \(\psi_2 * \psi_1\).  
Additionally, we refer to a formula connected by \( * \) and containing only atomic pure formulas 
as a pure formula.
We use the metavariable $\pi$ to denote pure formulas.
\( \psi_{\normalfont\text{pure}} \) denotes the symbolic heap obtained by extracting only the pure formulas from \( \psi \).

\( b(t) \) represents the head address of the memory block containing \( t \), and \( e(t) \) denotes the tail address of that block.  
We impose a constraint that \( t \) in \( b(t) \) or \( e(t) \) cannot  contain \( b \) or \( e \). For example, terms like \( b(b(t)) \) are not allowed in this setting.

For brevity, we represent the full expression  
\(t_0 \mathrel{R_0} t_1 * t_1 \mathrel{R_1} t_2 * \ldots * t_{n-1} \mathrel{R_{n-1}} t_n\),  
where \(R_i \in \{<, \approx\}\),  
by the shorter expression  
\(t_0 \mathrel{R_0} t_1 \mathrel{R_1} t_2 \ldots t_{n-1} \mathrel{R_{n-1}} t_n\).  
For example, \(x < y * y < z\) is written as \(x < y < z\).

\begin{definition}[$\normalfont\textsf{term}(t)$]
 \label{def: formal definition of term(t)}
 \normalfont
 We define the set ${\textsf{term}}(t)$ of terms for a term $t$ as follows.
\begin{align*}
  \textsf{term}(t + t') &= \{ t + t' \} \cup \textsf{term}(t) \cup \textsf{term}(t') \\ 
  \textsf{term}(b(t)) &= \{ b(t) \} \cup \textsf{term}(t) \\
  \textsf{term}(e(t)) &= \{ e(t) \} \cup \textsf{term}(t) \\
  \textsf{term}(n) &= \{ n \}, \quad
  \textsf{term}(x) = \{ x \}, \quad
  \textsf{term}(\texttt{null}) = \emptyset
\end{align*}

\end{definition}

\begin{definition}[$\normalfont\textsf{termS}(\psi)$ and $\normalfont\textsf{termS}^-(\psi)$]
\label{def: formal definition of terms(psi)}
 \normalfont
 We define the set $\textsf{termS}(\psi)$ of terms for a symbolic
 heap $\psi$ as follows.
  \begin{align*}
    \textsf{termS}(\psi_1 * \psi_2) &= 
    \textsf{termS}(\psi_1) \cup \textsf{termS}(\psi_2) \\
    \textsf{termS}(t \mapsto u ) &= \textsf{term}(t) \cup \textsf{term}(t+1) \cup \textsf{term}(u) \\
    \textsf{termS}(\arr(t, t')) &= \textsf{termS}(\bararr(t, t')) 
    = \textsf{termS}(t \approx t')   \\ 
    &=\textsf{termS}(t \not\approx t') 
    = \textsf{termS}(t \leq t') \\
    &= \textsf{termS}(t < t') 
    = \textsf{term}(t) \cup \textsf{term}(t') \\
        \textsf{termS}(\textsf{emp}) &= \emptyset
  \end{align*}
Furthermore, we define
\[
 \termS^-(\psi)
=\{r\in\termS(\psi)\mid \text{$r$ is contains neither $b$ nor $e$}\}.
\]
\end{definition}
 
Note that $\termS(\psi)$ contains not only terms that appear freely in $\psi$ but also related terms, such as $t+1$ in $t \mapsto u$.

For the semantics, we assume the set \(\textsc{Val}\) of values and the set \(\textsc{Loc}\) of locations, such that  
\(\textsc{Loc} \subseteq \textsc{Val}\) and \(\textsc{Val} = \mathbb{N}\) (non-negative integers).  
We assume \(0\) denotes the null pointer and \(0 \in \textsc{Val} \setminus \textsc{Loc}\).
Additionally, \(\bot \not\in \textsc{Val}\)\footnote{The concept of using \(\bot\) to signify deallocated locations was introduced by Raad et al. \cite{raad2020local}, specifically for negative heaps, expressed as \(x \not\mapsto\).  
It is important to note that \(\bot\) does not denote ``undefined'', as in other literature.} is used to track deallocated locations.

A \emph{state} $\normalfont \sigma \in \textsc{State}$ is defined as a triple  
$(s, h, B)$ (often called a heap model).  
A state consists of a store $\normalfont s \in \textsc{Store}$, a heap $\normalfont h \in \textsc{Heap}$,  
and a block set $\normalfont B \in \textsc{Blocks}$.

A \emph{store} is a total function from \(\normalfont \textsc{Var}\) to \(\normalfont \textsc{Val} \cup \{0\}\).  
For a store \(s\), the notation \(s[x \mapsto v]\) represents the store defined as  
$s[x \mapsto v](x) = v$  and  $s[x \mapsto v](y) = s(y)$ for $y \neq x$.

A \emph{heap} is a finite partial function from $\normalfont \textsc{Loc}$ to $\normalfont \textsc{Val} \cup \{\bot\}$.  
The update $h[l \mapsto v]$ for a heap $h$ is defined analogously to $s[x \mapsto v]$.  
Additionally, \(h[l \mapsto \uparrow]\) represents the heap defined by  
\(\textsf{dom}(h[l \mapsto \uparrow]) = \textsf{dom}(h) \setminus \{ l \}\),  
and \(h[l \mapsto \uparrow](l') = h(l')\) for \(l \neq l'\).
Two heaps \(h_1\) and \(h_2\) are \emph{disjoint} if  
${\normalfont{\textsf{dom}}}(h_1) \cap {\normalfont{\textsf{dom}}}(h_2) = \emptyset$.
The operation \(h_1 \circ h_2\) combines two disjoint heaps \(h_1\) and \(h_2\) and is defined as follows:
\[
(h_1 \circ h_2)(l) = 
\begin{cases} 
h_1(l) & \text{if } l \in {\normalfont{\textsf{dom}}}(h_1), \\ 
h_2(l) & \text{if } l \in {\normalfont{\textsf{dom}}}(h_2), \\ 
\text{undefined} & \text{if } l \notin {\normalfont{\textsf{dom}}}(h_1) \cup {\normalfont{\textsf{dom}}}(h_2).
\end{cases}
\]  
The operation \(h_1 \circ h_2\) is undefined if \(h_1\) and \(h_2\) are not disjoint.

For convenience, we sometimes denote \( \textsf{dom}_+ (h) \) as the set of locations \( l \) satisfying \( h(l) \in \textsc{Val} \), 
and  \( \textsf{dom}_- (h) \) as the set of locations \( l \) satisfying \( h(l) = \bot \). 
By definition, \( \textsf{dom}(h) = \textsf{dom}_+ (h) \cup \textsf{dom}_- (h) \) and 
\( \textsf{dom}_+ (h) \cap \textsf{dom}_- (h) = \emptyset \) hold.

We adopt the concept of a block from \cite{holik2022low}.  
An interval $[l, l')$ for $l,l' \in \textsc{Loc}$ representing consecutive memory locations, 
where the lower bound is included, and the upper bound is excluded.  
$[l, l')$ is defined as 
$[l, l') = \{ l'' \in \textsc{Loc} \mid l \leq l'' < l' \}$.
The set of intervals is defined as 
$\textsc{Int} = \{[l, l')   \mid
l, l' \in \textsc{Loc} \land l < l'   \}$.

A block set $B$ is a
finite collection of non-overlapping blocks, and 
the set of block sets is defined as follows.
\begin{align*}
\textsc{Blocks} = \{ B \subseteq_\text{fin} \textsc{Int} \mid \,
& \forall [l_1, {l'}_1), [l_2, {l'}_2) \in B .  \\ 
& [l_1, {l'}_1) \neq [l_2, {l'}_2) \rightarrow [l_1, {l'}_1) \cap [l_2, {l'}_2) = \emptyset \}
\end{align*}
For $B \in \textsc{Blocks}$, $\bigcup B$ means
$\bigcup_{[l, l') \in B} [l, l')$.

We define the functions \( b_B \) and \( e_B : \textsc{Loc} \to \textsc{Val} \), parameterized by  
\( B \in \textsc{Blocks} \), which return the base and end addresses, respectively, of the block to which a given location belongs.
Specifically, for  \( l \in \textsc{Loc} \), we define \( b_B(l) = u \) and \( e_B(l) = v \)
if \( l \in [u, v) \in B \). 
If no such interval exists, then we define \( b_B(l) = 0 \) and \( e_B(l) = 0 \).

\begin{definition}[Conditions on the models and exact models]
 \label{def: Conditions on the models of our ISL}
 \normalfont
We assume the following conditions for the states $( s,h,B )$. 
\begin{enumerate}
    \item 
    $\normalfont\textsf{dom}_+(h) \subseteq  \bigcup B$
    \item 
    $\normalfont\textsf{dom}_- (h) \cap  \left( \bigcup B \right) =
	  \emptyset$
\end{enumerate}

 Furthermore, if $(s,h,B)$ satisfies
 $\normalfont\textsf{dom}_+(h) = \bigcup B$, we call $(s,h,B)$
 \emph{exact}.
 We use $\textsc{State}^+$ to denote the set of the exact states.
\end{definition}
The first condition is slightly modified from the corresponding one in \cite{holik2022low}.  
In \cite{holik2022low}, it is sufficient to state that ``$h(l)$ is defined'',
but in ISL, due to the presence of negative heaps,  
we must distinguish between the cases where \(h(l) \in \textsc{Val}\) or \(h(l) = \bot\).

The second condition, similar to the first one, considers the case where
$h(l) = \bot$. This condition indicates that deallocated locations are not within a block; that is, locations within blocks must not be deallocated.

The denotational semantics of programs is defined for
exact states.

We  define the term interpretation $[\![ t ]\!]_{s,B}$ 
in the following.

\begin{definition}[Term interpretation]
\label{def: formal definition of term interpretation}
 {   \normalfont
\begin{align*}
  [\![ t + t' ]\!]_{s,B}   &=  [\![ t  ]\!]_{s,B} +  [\![  t' ]\!]_{s,B} \\ 
  [\![ b(t) ]\!]_{s,B}   &= b_B ([\![ t  ]\!]_{s,B}) \\
  [\![ e(t) ]\!]_{s,B}   &= e_B ([\![ t  ]\!]_{s,B}) \\
  [\![ n ]\!]_{s,B}   &=  n , \quad
  [\![ x ]\!]_{s,B} =  s(x) , \quad
  [\![ \texttt{null} ]\!]_{s,B}= 0
\end{align*}}
Note that if $t$ does not contain $b$ or $e$, 
 $[\![ t ]\!]_{s, B}$ does not depend on $B$.
\end{definition}


The semantics of the assertions of ISL is defined as follows.

\begin{definition}[Assertion semantics of ISL]
 \normalfont
 We define the relation $(s,h,B)\models P$ as 
{ 
\begin{align*}
    (s,h,B) \models \bigvee_{i \in I} \exists \overrightarrow{x_i} . \psi_i &\Leftrightarrow
   \exists i\in I.(s,h,B) \models 
     \exists \overrightarrow{x_i} . \psi_i, \\
    (s,h,B) \models \exists \overrightarrow{x_i}  . \psi_i  &\Leftrightarrow
    \exists \overrightarrow{v_i}  \in \textsc{Val} .  
    (s [\overrightarrow{x_i} \mapsto \overrightarrow{v_i}] ,h,B) \models \psi_i, \\
    (s,h,B) \models \psi_1 * \psi_2 &\Leftrightarrow 
    \exists h_1,h_2.\; h = h_1 \circ h_2 \text{\ and\ } \\
    & \ \ \ \ (s,h_1,B) \models \psi_1 \text{\ and\ } (s,h_2,B) \models \psi_2, \\ 
    (s,h,B) \models \textsf{emp} &\Leftrightarrow 
    \textsf{dom}(h) = \emptyset, \\
    (s,h,B) \models t \mapsto u  &\Leftrightarrow 
    \textsf{dom}(h) = 
     \{ [\![ t ]\!]_{s,B} \} \;\&\; h([\![ t ]\!]_{s,B}) = [\![ u ]\!]_{s,B},  \\
    (s,h,B) \models \textsf{arr}(t,t') &\Leftrightarrow 
    [\![ t ]\!]_{s,B}  < [\![ t' ]\!]_{s,B} \;\&\; \\
   & \ \ \ \ \ \ \  \textsf{dom}(h) = \textsf{dom}_+ (h) = 
     \{ [ [\![ t ]\!]_{s,B}  [\![ t' ]\!]_{s,B} ) \} ,
     \\
    (s,h,B) \models \bararr (t,t') &\Leftrightarrow 
    [\![ t ]\!]_{s,B}  < [\![ t' ]\!]_{s,B} \;\&\; \\
    & \ \ \ \ \ \ \  \textsf{dom}(h) = \textsf{dom}_- (h) = 
     \{ [ [\![ t ]\!]_{s,B}  [\![ t' ]\!]_{s,B} ) \} 
    , \\
    (s,h,B) \models t \approx t' &\Leftrightarrow 
    [\![ t ]\!]_{s,B}  = [\![ t' ]\!]_{s,B} \;\&\; \textsf{dom}(h) = \emptyset,\\
    (s,h,B) \models t \not\approx t' &\Leftrightarrow 
    [\![ t ]\!]_{s,B}  \neq [\![ t' ]\!]_{s,B} \;\&\; \textsf{dom}(h) = \emptyset, \\ 
    (s,h,B) \models t \leq t' &\Leftrightarrow 
    [\![ t ]\!]_{s,B}  \leq [\![ t' ]\!]_{s,B} \;\&\; \textsf{dom}(h) = \emptyset, \\ 
    (s,h,B) \models t < t' &\Leftrightarrow 
    [\![ t ]\!]_{s,B}  < [\![ t' ]\!]_{s,B} \;\&\; \textsf{dom}(h) = \emptyset ,
 \end{align*}}
where $s[\overrightarrow{x_i} \mapsto \overrightarrow{v_i}]$ means
$s[x_{i1} \mapsto v_{i1}] \ldots [x_{in} \mapsto v_{in}]$.
For a pure formula \(\psi\), \((s,B) \models \psi\) means that \((s, \emptyset , B) \models \psi \) for the empty heap \(\emptyset\).
Additionally, we define  entailments between formulas.  
The notation $\psi \models \varphi$ means that if $(s, h, B) \models \psi$, then $(s, h, B) \models \varphi$.  
We write $\psi \equiv \varphi$ to indicate that both $\psi \models \varphi$ and $\varphi \models \psi$ hold.

\end{definition}

\subsection{Program Language}

We recall the program language for ISL in \cite{raad2020local}.

\begin{definition}[Program Language $\normalfont\mathbb{C}\in \textsc{Comm}$]
{\normalfont
\begin{align*}
    \mathbb{C} ::= \ & \texttt{skip} \ | \ x:=t \ | \     x:= \texttt{*} \\
    &| \ \texttt{assume($\pi$)} \ | \ \texttt{local $x := t$ in $\mathbb{C}$} \\
    &| \
    \mathbb{C}_1 ; \mathbb{C}_2 \ | \ \mathbb{C}_1 + \mathbb{C}_2 \ | \ \mathbb{C}^\star⋆ 
    \ | \   \texttt{error()} \\
    &| \  x:= \texttt{alloc($t$)} \ | \ \texttt{free($t$)} \ | \ 
    x := [t] \ | \ [t]:= t'  
\end{align*}
In this definition, \(\pi\) is a pure formula, and \(t\) does not contain \(b\) or \(e\).
}
\end{definition}

We largely follow the command language traditions of Separation Logic (SL) \cite{reynolds2002separation}.  
The programming language includes standard constructs such as \texttt{skip},  
assignment \(x := t\), and nondeterministic assignment \(x := \texttt{*}\) (where \texttt{*} represents a nondeterministically selected value).  
It also features \texttt{assume($\pi$)},
local variable declarations \texttt{local $x$ in $\mathbb{C}$}, sequential composition \(\mathbb{C}_1 ; \mathbb{C}_2\),  
nondeterministic choice \(\mathbb{C}_1 + \mathbb{C}_2\), and loops \(\mathbb{C}^\star\).  
Additionally, it supports \texttt{error()} statements and various instructions for heap manipulation.

Heap manipulation instructions include operations such as allocation, deallocation, lookup, and mutation, described as follows:
\begin{itemize}
    \item \(x := \texttt{alloc}(t)\): Allocates \(t\) memory cells (from \(x\) to \(x + t - 1\)) on the heap, assigning the starting location to \(x\). It may reuse previously deallocated locations.
    \item \texttt{free($t$)}: Deallocates the memory location referred to by \(t\).
    \item \(x := [t]\): Reads the value stored at the heap location \(t\) and assigns it to \(x\).
    \item \([t] := t'\): Updates the heap location \(t\) with the value of \(t'\).
\end{itemize}

\begin{definition}[$\normalfont\textsf{termC}(\mathbb{C})$]
\label{def: formal definition of termC(command)}
 {\normalfont
 We define the set ${\textsf{termC}(\mathbb{C})}$ of terms for a program
 $\mathbb{C}$.
\begin{align*}    
    \textsf{termC}(\texttt{skip}) & = \textsf{termC}(\texttt{error()}) = \emptyset \\
    \textsf{termC}(x := t) & = \{ x \} \cup \textsf{term}(t) \\
    \textsf{termC}(x := \texttt{*}) & = \{ x \} \\
    \textsf{termC}(\texttt{assume($\psi$)}) & = \textsf{termS}(\psi) \\
    \textsf{termC}(\texttt{local } x := t \texttt{ in } \mathbb{C}) & = 
    (\textsf{termC}(\mathbb{C})  \setminus \{x\} ) \cup \textsf{term}(t) \\
    \textsf{termC}(x := \texttt{alloc($t$)}) & = \{x\}\cup\textsf{term}(t) \\
    \textsf{termC}(\texttt{free($t$)}) & = 
    \textsf{term}(t) \cup  \textsf{term}(b(t)) \cup \textsf{term}(e(t)) \\
    \textsf{termC}(x := [t]) & = 
    \{ x \} \cup \textsf{term}(b(t)) \\
    \textsf{termC}([t] := t') & = 
    \textsf{term}(t + 1) \cup \textsf{term}(b(t)) \cup \textsf{term}(t') \\
    \textsf{termC}(\mathbb{C}_1 ; \mathbb{C}_2) & = 
    \textsf{termC}(\mathbb{C}_1 + \mathbb{C}_2) = \\
    & \ \ \ \ \ \textsf{termC}(\mathbb{C}_1) \cup \textsf{termC}(\mathbb{C}_2) \\
    \textsf{termC}(\mathbb{C}^\star) & = \textsf{termC}(\mathbb{C})
\end{align*}}
\end{definition}
This function encompasses both the terms appearing in \(\mathbb{C}\) and the variables or heap locations referenced by \(\mathbb{C}\).
It is used to generate the weakest postconditions.

\begin{definition}[The set of terms modified by $\mathbb{C}$, $\normalfont\textsf{mod}(\mathbb{C})$]
\label{def: formal definition of mod(C)}
{    \normalfont
\begin{align*}
    \textsf{mod}(\texttt{skip}) &= \textsf{mod}(\texttt{assume($\psi$)}) = \textsf{mod}(\texttt{error()})  \\
    &= \textsf{mod}(\texttt{free($t$)}) = \textsf{mod}([t] := t') =  \emptyset \\
    \textsf{mod}(x := t) &= \textsf{mod}(x := \texttt{*}) = \textsf{mod}(x := \texttt{alloc($t$)})  \\ 
    & = \textsf{mod}(x := [t]) = \{ x \} \\
    \textsf{mod}(\texttt{local $x$ in $\mathbb{C}$}) &= \textsf{mod}(\mathbb{C}) \setminus \{ x \} \\
    \textsf{mod}(\mathbb{C}_1 ; \mathbb{C}_2) &= \textsf{mod}(\mathbb{C}_1 + \mathbb{C}_2) =  \textsf{mod}(\mathbb{C}_1) \cup \textsf{mod}(\mathbb{C}_2) \\
    \textsf{mod}(\mathbb{C}^\star) &= \textsf{mod}(\mathbb{C})
\end{align*}}
\end{definition}

\begin{definition}[Denotational semantics of ISL]\label{def:
 Denotational semantics of ISL} 
 \normalfont
 We define $[\![\mathbb{C}]\!]_\epsilon$
 for a program $\mathbb{C}$ and an \textit{exit condition}
 $\epsilon\in\{\textit{ok},\textit{er}\}$ as a binary relation on
 {\normalfont$\textsc{State}^+$}. We use the composition of binary relations: for
 binary relations $R_1$ and $R_2$, $R_1;R_2$ is the relation
 $\{(a,c)\mid \exists b.(a,b)\in R_1\text{\ and\ }(b,c)\in R_2\}$.
Here, \( \normalfont\mathbb{C}^0 = \texttt{skip} \) and \( \mathbb{C}^{m+1} = \mathbb{C};\mathbb{C}^{m} \).

 {\scriptsize\normalfont
 \begin{align*}
  [\![ \texttt{skip} ]\!]_{\it ok}
  &= \{ (\sigma,\sigma) \ | \ \sigma \in \textsc{State}^+ \}\\
  [\![ \texttt{skip} ]\!]_{\it er}
  &=\emptyset\\
   [\![ x := t ]\!]_{\it ok} &=
   \{ ( (s,h,B),  (s[x \mapsto [\![ t ]\!]_{s,B}], h, B)   )  \} \\
   [\![ x := t ]\!]_{\it er} &=
   \emptyset \\
 [\![ x := \texttt{*} ]\!]_{\it ok} &=
 \{ ( (s,h,B),  (s[x \mapsto v], h, B)   )  \ | \ v \in \textsc{Val} \}\\
 [\![ x := \texttt{*} ]\!]_{\it er} &=
 \emptyset\\
[\![ \texttt{assume($\psi$)} ]\!]_{\it ok} &=
\{ ((s,h,B),(s,h,B)) \ | \   (s, B ) \models \psi \} \\
[\![ \texttt{assume($\psi$)} ]\!]_{\it er} &=
\emptyset\\
[\![ \texttt{error()} ]\!]_{\it ok} &=
\emptyset \\
[\![ \texttt{error()} ]\!]_{\it er} &=
\{ (\sigma,\sigma) \ | \ \sigma \in \textsc{State}^+ \} \\ 
[\![ \mathbb{C}_1; \mathbb{C}_2 ]\!]_{\it ok} &=
   [\![ \mathbb{C}_1]\!]_\textit{ok};[\![ \mathbb{C}_2]\!]_\textit{ok}\\
[\![ \mathbb{C}_1; \mathbb{C}_2 ]\!]_{\it er} &=
   [\![ \mathbb{C}_1]\!]_\textit{er} \cup
   [\![ \mathbb{C}_1]\!]_\textit{ok};[\![ \mathbb{C}_2]\!]_\textit{er} \\
[\![ \texttt{local $x $ in $\mathbb{C}$} ]\!]_\epsilon &=
\{  (  (s ,h,B),  (s' ,h',B')   ) \mid \exists v , v' \in \textsc{Val} . \\
&\ \ \ \ ( (s [x \mapsto v ] ,h,B)  ,  (s' [x \mapsto v'] ,h',B')  ) 
\in [\![\mathbb{C}   ]\!]_\epsilon    \;\&\; s(x) = s'(x) \} \\
[\![ \mathbb{C}_1 + \mathbb{C}_2 ]\!]_\epsilon &= 
[\![ \mathbb{C}_1  ]\!]_\epsilon \cup [\![ \mathbb{C}_2  ]\!]_\epsilon \\
[\![ \mathbb{C}^\star ]\!]_\epsilon &=  
\bigcup_{m\in \mathbb{N}} [\![ \mathbb{C}^m ]\!]_\epsilon \\
[\![ x := \texttt{alloc($t$)} ]\!]_\textit{ok} &=
\{ (  (s,h,B) ,  (s [x \mapsto l] , 
h[i \mapsto v_i ]_{i=l}^{l+[\![ t ]\!]_{s,B} - 1 }  ,  
B \cup \{ [l , l + [\![ t ]\!]_{s,B}  )  \} ) )  \mid   \\
&l \in \textsc{Loc} \;\&\;
(\bigcup B) \cap  [ l , l + [\![ t ]\!]_{s,B}  )   = \emptyset
\;\&\; \text{ $v_i \in \textsc{Val}$ for $i \in [0, [\![ t ]\!]_{s,B} )$} \}  \\
[\![ x := \texttt{alloc($t$)} ]\!]_\textit{er} &= \emptyset \\
[\![\texttt{free($t$)} ]\!]_\textit{ok} &= 
\{ (  (s,h,B), 
(s, h [i  \mapsto \bot]_{i= b_B( [\![ t ]\!]_{s,B}  ) }^{e_B( [\![ t ]\!]_{s,B}) -1} , B \setminus \{ [  b_B( [\![ t ]\!]_{s,B}  ) , e_B( [\![ t ]\!]_{s,B})   ) \} ) )  \mid \\
& [\![ t ]\!]_{s,B} =  b_B( [\![ t ]\!]_{s,B}  ) \;\&\;  b_B( [\![ t ]\!]_{s,B}  ) > 0  \}  \\
[\![\texttt{free($t$)} ]\!]_\textit{er} &= 
\{ (  (s,h,B), (s,h,B) ) \mid 
 [\![ t ]\!]_{s,B} \neq  b_B( [\![ t ]\!]_{s,B}  ) \;\text{or}\; b_B( [\![ t ]\!]_{s,B}  ) = 0 \}
 \\
[\![ x := [t] ]\!]_\textit{ok} &= 
\{ (  (s,h,B), (s [x \mapsto h([\![ t ]\!]_{s,B})] , h , B) ) \mid  b_B( [\![ t ]\!]_{s,B} ) > 0  \} 
 \\
 [\![ x := [t] ]\!]_\textit{er} &= 
\{ (  (s,h,B), (s,h,B) ) \mid 
 b_B( [\![ t ]\!]_{s,B} ) = 0  \} 
 \\
 [\![ [t] := t' ]\!]_\textit{ok} &= 
\{ (  (s,h,B), 
(s , h [[\![ t ]\!]_{s,B} \mapsto [\![ t' ]\!]_{s,B} ] , B) ) \mid  b_B( [\![ t ]\!]_{s,B} ) > 0   \} 
 \\
 [\![ [t] := t' ]\!]_\textit{er} &= 
\{ (  (s,h,B), (s,h,B) ) \mid 
 b_B( [\![ t ]\!]_{s,B} ) = 0   \} 
 \end{align*}
 }

\end{definition}

Based on the previous definitions, the validity of ISL triples of the form \([P]\ \mathbb{C}\ [\epsilon: Q]\) is defined as follows.

\begin{definition}[Validity of ISL Triples]
    $$\models [P] \ \mathbb{C} \ [\epsilon: Q] \overset{\text{def}}{\iff}
    \forall \sigma'\models Q  , \exists \sigma  \models P. 
    (\sigma,\sigma') \in [\![ \mathbb{C} ]\!]_\epsilon$$
\end{definition}

\subsection{Canonicalization}\label{subsec: Canonicalization}

We need to define canonical forms for the inference rules of certain heap-manipulating proof rules.
Before defining our canonical forms, we introduce two functions $\textsf{perm}(T)$ and $\textsf{cases}(T)$. 
Let \( \textsf{perm}(T) \) denote the set of all permutations of a finite set \( T = \{t_1, t_2, \ldots , t_n \} \) of terms as follows.
\begin{align*}
\textsf{perm}(T) = \big\{ \,
& (t_{\rho(1)}, t_{\rho(2)}, \dots, t_{\rho(n)}) \mid \, \\
& \rho: \{1, 2, \dots, n\} \to \{1, 2, \dots, n\} \text{ is a bijection} 
\big\}
\end{align*}

For  a finite set $T$ of terms, the case analysis function \( \textsf{cases}(T) \) generates
a pure formulas representing
every possible sequence of strict inequalities and equalities between the elements of \( T \) for each permutation.  
\begin{align*}
\textsf{cases}(T) = \bigcup_{(t_1, t_2, \dots, t_n) \in \textsf{perm}(T)} 
\big\{ \,
& \texttt{null} \mathrel{R_0} t_1 \mathrel{R_1} t_2 \dots t_{n-1} \mathrel{R_{n-1}} t_n \mid \\
& R_i \in \{<, \approx\} \text{ for } 0 \leq i \leq n-1 
\big\}
\end{align*}

The following lemma holds by the definition of $\textsf{cases}(T)$.

 \begin{lemma}
  $ \psi \equiv 
  \bigvee_{\normalfont\pi \in \textsf{cases}(T)}( \pi * \psi) $
  holds for any $\psi$ and $T$.
 \end{lemma}

The canonical forms of quantifier-free symbolic heaps concerning a
set  $T$ of terms are defined as follows.

\begin{definition}[Canonical forms]
{\normalfont
\begin{align*}
\textsf{CF}_{\text{sh}}(T) = \{ \psi \mid   \psi_{\text{pure}} \models \pi \text{ for some } \pi \in \textsf{cases}(T) \}
\end{align*}}
\normalfont
If \( \psi \in \textsf{CF}_\text{sh}(T) \) for a finite set of terms \( T \), then \( \psi \) is called \textit{canonical} for \( T \).  
\end{definition}

The canonicalization
 divides into the cases for each possible sequence of strict
 inequalities and equalities individually to the input symbolic heap.

\begin{definition}[Canonicalization $\normalfont\textsf{cano}$]
 \label{def: Case Analysis for all possible aliasing and non-aliasing}
 \normalfont
 The set $\textsf{cases}_{\rm sh}(\psi,{\mathbb C})$ of cases
 for canonicalization is defined as
 \[
  {\normalfont\textsf{cases}_{\rm sh}}(\psi, \mathbb{C}) =
  \{\pi * \psi \mid
  \pi \in {\normalfont\textsf{cases}} ({\normalfont \textsf{termS}(\psi) \cup \textsf{termC}(\mathbb{C})})
  \}.
\]
For $P = \bigvee_{i \in I} \exists \overrightarrow{x_i} . \psi_i$ and 
 $\normalfont  \overrightarrow{x_i} \notin \textsf{termC}(\mathbb{C})$.
$$\normalfont
 \textsf{cano}(P,\mathbb{C}) =  
 \bigvee_{i \in I}
 \bigvee_{\varphi \in \textsf{cases}(\psi_i,\mathbb{C})}
 \exists \overrightarrow{x_i} .\varphi$$
\end{definition}

\begin{lemma}\label{lma: cano and P is equivalent}
\(\normalfont  P \equiv \textsf{cano}(P, \mathbb{C}) \) and  
 \(  \normalfont \psi_i \in \textsf{CF}_\text{sh} (\textsf{termS}(\psi_i) \cup \textsf{termC} (\mathbb{C})) \)  
for any \(i \in I\) hold for any \(P\)  
and \(\mathbb{C}\), with \(\normalfont \textsf{cano}(P, \mathbb{C}) = \bigvee_{i \in I} \exists \overrightarrow{x_i} . \psi_i\).
\end{lemma}

\subsection{Proof Rules}

In this section, we give the inference rules of our system.

For rules regarding memory deallocation, we propose the notation \( \Arr(t, t')_\alpha \) serves as a unifying representation for \( t \mapsto \alpha \) and \( \arr(t, t') \), simplifying the rules for heap manipulation proofs. 
In this notation, \( \alpha \) represents either a term or the symbol \( \bullet \),
which is explicitly not a term. Define:  

\begin{itemize}
    \item   \( \Arr(t, t+1)_u \) denotes \( t \mapsto u \), and
    \item  \( \Arr(t, t')_\bullet \) denotes \( \arr(t, t') \).
\end{itemize}

Furthermore, for a simpler notation, we use the following
convention. 
Consider a symbolic heap \( \psi \) and a set of terms \( T = \{t_1, t_2, \dots, t_n\} \), where no two terms in \( T \) are in the subterm relation.
The replacement of all elements in \( T
\) with a term \( u \) is denoted as $\psi[T := u]$.
We further define the following sets:  
\begin{align*}
T_b(\psi, t) &= \{ b(t') \in \textsf{termS}(\psi) \mid \psi_\text{pure} \models b(t')
\approx b(t) \},\ \text{and}\\
T_e(\psi, t) &= \{ e(t') \in \textsf{termS}(\psi ) \mid \psi_\text{pure} \models e(t') \approx e(t) \}.
\end{align*}

\begin{definition}[ISL Proof Rules]
We define our ISL proof rules in Figures \ref{fig:generic proof rules of ISL}, \ref{fig: rules for Memory Allocation}, \ref{fig: rules for Memory Deallocation}, and \ref{fig: rules for Memory Manipulation}, where the
 axioms of the form $[P]\ \mathbb{C}\
 [\textit{ok}:Q_1][\textit{er}:Q_2]$ means that both $[P]\ \mathbb{C}\
 [\textit{ok}:Q_1]$ and $[P]\ \mathbb{C}\ [\textit{er}:Q_2]$.
\end{definition}

Note that we employ infinitary syntax and incorporate an infinitary version of the
\textsc{Disj} rule following a similar approach used in RHL
\cite{de2011reverse}.

The heap manipulation rules require symbolic heaps to be in canonical
form as a premise.  
We provide an intuitive explanation of the rules for memory allocation in Figure \ref{fig: rules for Memory Allocation} and deallocation rules in Figure \ref{fig: rules for Memory Deallocation}.

The rules for allocation depend on the relationship between the
locations to be allocated and the locations indicated by \(\bararr\) in
the precondition.  As the block information changes, the values of
\(b(t')\) and \(e(t')\) for \(t'\), which are in the newly allocated
block locations, also change.  Therefore, the values of \(b(t')\) and
\(e(t')\) are replaced with \(0\) (using a function $\textsf{rep}_\alpha^\beta$), since
before allocation, \(t'\) is not allocated.

The rules for deallocation perform a case analysis on the range of the
freed locations and the locations of \(\mapsto\) and \(\arr\) in the
precondition.  Since block information changes, the values of \(b(t')\)
and \(e(t')\) for \(t'\) in the deallocated block 
also change.  Therefore, such terms in the precondition are replaced with
the values before the deallocation operation (using a replacement
\(\tau\)).

Regarding the \textsc{Frame} rule, we consider only the \textit{ok} case.  
This is because, if we include the \textit{er} case in the \textsc{Frame} rule,  
we find that soundness is broken in our setting.  
The following is an example.
$$
\inferrule[]
{[x \not\approx \texttt{null} * \textsf{emp}] \ \texttt{free($x$)} \  
[\textit{er}: x \not\approx \texttt{null} * \textsf{emp}]}
{[x \not\approx \texttt{null} * \textsf{emp} * x \mapsto 1] \ \texttt{free($x$)} \  
[\textit{er}: x \not\approx \texttt{null} * \textsf{emp} * x \mapsto 1 ]  }
$$
If we use $x \mapsto 1$ as a frame, the conclusion triple becomes invalid, while the premise triple remains valid,  
indicating that soundness does not hold in this case.


\begin{figure}
{\scriptsize\begin{mathpar}
\inferrule[\textsc{Skip}]
{ \ }
{ [P] \ \texttt{skip} \ [\textit{ok}: P] \
[\textit{er}:\textsf{false}] }
\and
\inferrule[\textsc{Error}]
{ \ }
{ [ P ] \ \texttt{error()} \ [\textit{ok}:\textsf{false}] \
[\textit{er}: P]}
\and
\inferrule[\textsc{Seq1}]
{ [P] \ \mathbb{C}_1 \ 
[\textit{er}: Q] }
{ [P] \ \mathbb{C}_1;\mathbb{C}_2 \ 
[\textit{er}: Q] }
\and
\inferrule[\textsc{Seq2}]
{ [P] \ \mathbb{C}_1 \ [\textit{ok}:R] \ \ \ 
 [R] \ \mathbb{C}_2 \ [\epsilon:Q]
}
{ [P] \ \mathbb{C}_1;\mathbb{C}_2 \ 
[ \epsilon : Q] }
\and    
\inferrule[\textsc{Loop zero} \normalfont\text{(equal to \textsc{Skip})}]
{ \ }
{ [P] \ \mathbb{C}^\star \ [\textit{ok}: P] \ [\textit{er}:\textsf{false}] }
\and
\inferrule[\textsc{Loop non-zero}]
{  [P] \ \mathbb{C}^\star;\mathbb{C} \ [\epsilon: Q]  }
{  [ P ] \ \mathbb{C}^\star \ [\epsilon: Q] }
\and
\inferrule[\textsc{Cons}]
{P' \models P \ \ \ 
 [P'] \ \mathbb{C} \ [\epsilon: Q']
\ \ \ Q \models Q'}
{ [P] \ \mathbb{C} \ [\epsilon: Q] }
\and
\inferrule[\textsc{Disj}]
{ [  P_i ] \ \mathbb{C} \
    [\epsilon:  Q_i  ]  \quad \text{for all $i \in I$}}
{ [   \bigvee_{i \in I} P_i ] \ \mathbb{C} \ 
[ \epsilon :  \bigvee_{i \in I} Q_i] }
\and
\inferrule[\textsc{Choice}]
{ [P] \ \mathbb{C}_1 \ [ \epsilon : Q] \ \ \ 
 [P] \ \mathbb{C}_2 \ [\epsilon:Q]
}
{ [P] \ \mathbb{C}_1 + \mathbb{C}_2 \ 
[ \epsilon : Q] }
\and
\inferrule[\textsc{Exist}]
{ [  \psi ] \ \mathbb{C}  \ 
[ \epsilon :   \varphi  ]  \ \ \  x \notin \textsf{termS}(\mathbb{C})  }
{ [  \exists  x  . \psi  ] \ \mathbb{C}  \ 
[ \epsilon :  \exists  x  . \varphi  ] }
\and
\inferrule[\textsc{Assign}]
{\ }
{ [ \psi ] \ x:=t \ [\textit{ok}:
\exists  x' . \psi \theta * x \approx  t  \theta   ] \
[\textit{er}:\textsf{false}] }
\and
\inferrule[\textsc{Havoc}]
{ \ }
{ [ \psi ] \ x:=\texttt{*} \ [\textit{ok}: \exists x'. \psi \theta ] \
[\textit{er}:\textsf{false}] }
\and
\inferrule[\textsc{Assume} \normalfont\text{($\pi$ is a pure formula)}]
{ \ }
{ [\psi] \ \texttt{assume($\pi$)} \ [\textit{ok}: \pi * \psi] \
[\textit{er}:\textsf{false}] }
\and
\inferrule*[lab=\textsc{Local}]
{ [ \psi   ] \ \mathbb{C}  \ 
[ \epsilon :  \varphi]   \ \ \ 
{x \notin \textsf{termS}(\psi)}
}
{ [  \psi  ] 
\ \texttt{local $x $ in $\mathbb{C}$}  \ 
[ \epsilon : \exists  x. \varphi] }
\and
\inferrule[\textsc{Frame Ok}]
{ [ \psi] \ \mathbb{C} \ [ \textit{ok} :  \varphi] \ \ \ 
\textsf{mod}(\mathbb{C}) \cap
\textsf{termS}( \phi) = \emptyset }
{ [ \psi * \phi] \ \mathbb{C} \ 
[ \textit{ok}  : \varphi * \phi  ] }
\end{mathpar}}
Here, $\theta$ denotes $[x:=x']$.
\caption{Proof Rules of ISL}
    \label{fig:generic proof rules of ISL}
\end{figure}

\begin{figure}
{\scriptsize\begin{mathpar}
\inferrule[\textsc{Alloc 1}] 
{ \psi \in \textsf{CF}_\text{sh}
 (\textsf{termS}(\psi)\cup\textsf{termC}(x:=\texttt{alloc($t$)})) \\ 
0\le\alpha\le\beta\le N \\
 \psi_{\text{pure}}\models{\tt null}\ R_1\ u_1\ R_2\ 
       \cdots\ R_N u_N \text{ ($R_i\in\{<,\approx\}$)} \\
       }
{ [\psi] \ x:=\texttt{alloc($t$)} \ 
[\textit{ok}:  
\exists x'.  \textsf{rep}_\alpha^\beta \bigg(
\psi\theta * \arr(x,x+t\theta) * \varphi
 \bigg)
]}
\and
\inferrule[\textsc{Alloc 2}]  
{ \psi \in \textsf{CF}_\text{sh} (\textsf{termS}(\psi)\cup\textsf{termC}(x:=\texttt{alloc($t$)})  )   \\
\psi= \bigast_{i=1}^{m}\bararr(t_i,t'_i) * \psi'
       \text{ ($\psi'$ contains no $\bararr$)} \\
       0\le\alpha\le\beta\le N \\
\psi_{\text{pure}} \models 
 \bigast_{i=1}^{m-1} t_i < t_{i+1}  \\ 
 1 \leq j \leq k \leq m \\
 \psi_{\text{pure}}\models{\tt null}\ R_1\ u_1\ R_2\ 
       \cdots\ R_N u_N \text{ ($R_i\in\{<,\approx\}$)} \\
              }
{ [  \psi ]
\ x:=\texttt{alloc($t$)} \
[\textit{ok}:  
\exists x'.   \textsf{rep}_\alpha^\beta \bigg(
\psi'\theta \\
 * \bigast_{i=1}^{j-1} \bararr(t_i\theta,t'_i\theta)
 * \bararr(t_j\theta,x)
 * \arr(x,x+t\theta)
 * \bararr(x+t\theta,t'_k\theta)
 * \bigast_{i=k+1}^{m} \bararr(t_i\theta,t'_i\theta)\\
 * t_j\theta<x<t'_j\theta * t_k\theta<x+t\theta< t'_k\theta * \varphi \bigg)
] }
\and
\inferrule[\textsc{Alloc 3}]  
{ \psi \in \textsf{CF}_\text{sh} (\textsf{termS}(\psi)\cup\textsf{termC}(x:=\texttt{alloc($t$)})  )   \\
\psi= \bigast_{i=1}^{m}\bararr(t_i,t'_i) * \psi'
       \text{ ($\psi'$ contains no $\bararr$)} \\
       0\le\alpha\le\beta\le N \\
\psi_{\text{pure}} \models 
 \bigast_{i=1}^{m-1} t_i < t_{i+1}  \\ 
 1 \leq j \leq k \leq m \\
 \psi_{\text{pure}}\models{\tt null}\ R_1\ u_1\ R_2\ 
       \cdots\ R_N u_N \text{ ($R_i\in\{<,\approx\}$)} \\
       }
{ [  \psi ]
\ x:=\texttt{alloc($t$)} \
[\textit{ok}:  
\exists x'.   \textsf{rep}_\alpha^\beta \bigg(
\psi'\theta\\
 * \bigast_{i=1}^{j-1} \bararr(t_i\theta,t'_i\theta)
 * \arr(x,x+t\theta)
 * \bararr(x+t\theta,t'_k\theta)
 * \bigast_{i=k+1}^{m} \bararr(t_i\theta,t'_i\theta)\\
 * t'_{j-1}\theta\le x\le t_j\theta 
 * t_k\theta<x+t\theta< t'_k\theta * \varphi
 \bigg)
]  }
\and

\inferrule[\textsc{Alloc 4}]  
{ \psi \in \textsf{CF}_\text{sh} (\textsf{termS}(\psi)\cup\textsf{termC}(x:=\texttt{alloc($t$)})  )  \\
\psi= \bigast_{i=1}^{m}\bararr(t_i,t'_i) * \psi'
       \text{ ($\psi'$ contains no $\bararr$)} \\
          0\le\alpha\le\beta\le N \\
\psi_{\text{pure}} \models 
 \bigast_{i=1}^{m-1} t_i < t_{i+1}  \\ 
 1 \leq j \leq k \leq m \\
 \psi_{\text{pure}}\models{\tt null}\ R_1\ u_1\ R_2\ 
       \cdots\ R_N u_N \text{ ($R_i\in\{<,\approx\}$)} \\
       }
{ [  \psi ]
\ x:=\texttt{alloc($t$)} \
[\textit{ok}:  
\exists x'.    \textsf{rep}_\alpha^\beta \bigg(
\psi'\theta\\
 * \bigast_{i=1}^{j-1} \bararr(t_i\theta,t'_i\theta)
 * \bararr(t_j\theta,x) 
 * \arr(x,x+t\theta)
 * \bigast_{i=k+1}^{m} \bararr(t_i\theta,t'_i\theta)\\
 * t_j\theta<x<t'_j\theta 
 * t'_k\theta\le x+t\theta\le t_{k+1}\theta * \varphi
 \bigg)
] }
\and
\inferrule[\textsc{Alloc 5}]  
{ \psi \in \textsf{CF}_\text{sh} (\textsf{termS}(\psi)\cup\textsf{termC}(x:=\texttt{alloc($t$)})  )  \\
\psi= \bigast_{i=1}^{m}\bararr(t_i,t'_i) * \psi'
       \text{ ($\psi'$ contains no $\bararr$)} \\
          0\le\alpha\le\beta\le N \\
\psi_{\text{pure}} \models 
 \bigast_{i=1}^{m-1} t_i < t_{i+1}  \\ 
 1 \leq j \leq k \leq m \\
 \psi_{\text{pure}}\models{\tt null}\ R_1\ u_1\ R_2\ 
       \cdots\ R_N u_N \text{ ($R_i\in\{<,\approx\}$)} \\
       }
{ [  \psi ]
\ x:=\texttt{alloc($t$)} \
[\textit{ok}:  
\exists x'.   \textsf{rep}_\alpha^\beta \bigg(
\psi'\theta\\
 * \bigast_{i=1}^{j-1} \bararr(t_i\theta,t'_i\theta)
 * \arr(x,x+t\theta)
 * \bigast_{i=k+1}^{m} \bararr(t_i\theta,t'_i\theta)\\
 * t'_{j-1}\theta\le x\le t_{j}\theta 
 * t'_k\theta\le x+t\theta\le t_{k+1}\theta * \varphi
 \bigg)
] }
\and
\inferrule[\textsc{AllocEr}]
{ \ }
{ [\psi] \ x:=\texttt{alloc($t$)} \ [\textit{er}: \textsf{false}]}
\end{mathpar}}

Here, $\theta$ denotes $[x := x']$ and $\varphi$ denotes the formula
 $b(x) \approx x * e(x) \approx x+t\theta * u_\alpha\theta< x \le
 u_{\alpha+1}\theta * u_\beta\theta<x+t\theta\le u_{\beta+1}\theta$.
Suppose $\{u_i\mid 1\le i\le
N\}=\termS^-(\psi)$. $\textsf{rep}_\alpha^\beta(\chi)$ is the formula
obtained from $\chi$ by replacing $b(u_i\theta)$ and $e(u_i\theta)$ with
${\tt null}$ for $\alpha<i\le\beta$ where $\termS^-(\psi) = \{ u_i \mid
1 \leq i \leq N \} $.

\caption{Rules for Memory Allocation}
    \label{fig: rules for Memory Allocation}
\end{figure}

\begin{figure} 

{\scriptsize\begin{mathpar}

\inferrule[FreeArr 1]
{ \psi \in \textsf{CF}_\text{sh} (\textsf{termS}(\psi) \cup
	     \textsf{termC}(\texttt{free($t$)}))  
\\
\psi = \arr (t_1, {t'}_1 ) *  \bigast_{i=2}^{k-1} \Arr (t_i, {t'}_i)_{\alpha_i}
 *\arr (t_k, {t'}_k ) *  \psi'
   \\
   \psi_{\text{pure}} \models 
t_1 < b(t) < {t'}_1 * t_k < e(t) < {t'}_k 
*
b(t) \approx t *  
\bigast_{i=1}^{k-1} {t'}_i \approx t_{i+1}
 * \bigast_{i=1}^{k-1} {t}_i < t_{i+1}
}
{ [ \psi ] \ \texttt{free($t$)} \ 
[\textit{ok}: 
\exists y . 
(\arr (t_1 , t)   * \bararr ( t , y ) * 
\arr(y, {t'}_k )  * \psi'  )
\tau
]  
 \ [\textit{er}:\textsf{false}]
}
\and
\inferrule[FreeArr 2]
{ \psi \in \textsf{CF}_\text{sh} (\textsf{termS}(\psi) \cup \textsf{termC}(\texttt{free($t$)}))   \\
\psi = \bigast_{i=1}^{k-1} \Arr (t_i, {t'}_i)_{\alpha_i} * \arr(t_k , {t'}_k) * \psi'
\\ 
\psi_{\text{pure}} \models 
t_1 \approx b(t) * t_k < e(t) < {t'}_k 
*
b(t) \approx t * 
\bigast_{i=1}^{k-1} {t'}_i \approx t_{i+1}   
 * \bigast_{i=1}^{k-1} {t}_i < t_{i+1}
}
{ [ \psi ] \ \texttt{free($t$)} \ 
[\textit{ok}: \exists y . 
( \bararr ( t , y ) * \arr(y, {t'}_k )  * \psi' ) 
\tau
]   
 \ [\textit{er}:\textsf{false}]
}
\and
\inferrule[FreeArr 3]
{ \psi \in \textsf{CF}_\text{sh} (\textsf{termS}(\psi) \cup \textsf{termC}(\texttt{free($t$)}))   \\
\psi = \bigast_{i=1}^{k} \Arr (t_i, {t'}_i)_{\alpha_i} * \psi'
\\
\psi_{\text{pure}} \models 
t_1 \approx  b(t) *   e(t) \approx {t'}_k 
*
b(t) \approx t * 
\bigast_{i=1}^{k-1} {t'}_i \approx t_{i+1} 
 * \bigast_{i=1}^{k-1} {t}_i < t_{i+1}
}
{ [ \psi ] \ \texttt{free($t$)} \ 
[\textit{ok}: \exists y . 
( \bararr ( t , y ) *   \psi' ) 
\tau
]    
 \ [\textit{er}:\textsf{false}]
}
\and
\inferrule[FreeArr 4]
{ \psi \in \textsf{CF}_\text{sh} (\textsf{termS}(\psi) \cup \textsf{termC}(\texttt{free($t$)})) \\
\psi = \arr(t_1, {t'}_1 ) *   \bigast_{i=2}^{k} \Arr (t_i, {t'}_i)_{\alpha_i} * \psi'
\\
\psi_{\text{pure}} \models 
t_1 < b(t) < {t'}_1 *   e(t) \approx  {t'}_k 
*
b(t) \approx t * 
\bigast_{i=1}^{k-1} {t'}_i \approx t_{i+1}  
 * \bigast_{i=1}^{k-1} {t}_i < t_{i+1}
}
{ [ \psi ] \ \texttt{free($t$)} \ 
[\textit{ok}: \exists y . 
(  \arr (t_1 , t)   * \bararr (t , y )  *   \psi' ) 
\tau
]   
 \ [\textit{er}:\textsf{false}]
}
\and
\inferrule[\textsc{FreeEr}]
{ \psi \in \textsf{CF}_\text{sh} (\textsf{termS}(\psi) \cup \{x, b(t) , e(t) \}) \ \ 
\psi_{\text{pure}} \models  b(t) \not\approx t \lor b(t) \approx \texttt{null}
}
{ [\psi ] \ \texttt{free($t$)} \ 
[\textit{ok}:\textsf{false}] \ [\textit{er}:\psi]}
\end{mathpar}}
Here, \( \tau \) denotes the replacement
\( [T_b(\psi, t) := t] [T_e(\psi, t) := y] \).

\caption{Rules for Memory Deallocation}
    \label{fig: rules for Memory Deallocation}
\end{figure}

\begin{figure}

{\scriptsize\begin{mathpar}
\inferrule[\textsc{LoadPtr}]
{ \psi \in \textsf{CF}_\text{sh} (\textsf{termS}(\psi) \cup \textsf{termC}(x:=[t]))\\
\psi_{\text{pure}} \models   t_a \approx t 
\\
\psi = t_a \mapsto u * \psi' 
}
{ [\psi ] \ x := [t] \ 
[\textit{ok}: \exists x'. \psi \theta * x \approx u \theta   ]
\ 
[\textit{er}:\textsf{false}]}
\and
\inferrule[\textsc{LoadArr}]
{ \psi \in \textsf{CF}_\text{sh} (\textsf{termS}(\psi) \cup
\textsf{termC}(x:=[t])) \\
\psi_{\text{pure}} \models   t_a \leq t < {t'}_a
 \\
 \psi =\arr (t_a , {t'}_a)  * \psi' 
}
{ [\psi ] \ x := [t] \ 
[\textit{ok}: \exists x' . \psi \theta  ]
\ 
[\textit{er}:\textsf{false}]}
\and
\inferrule[\textsc{LoadEr}]
{ \psi \in \textsf{CF}_\text{sh} (\textsf{termS}(\psi) \cup \textsf{termC}(x:=[t])) \\
\psi_{\text{pure}} \models   b(t) \approx \texttt{null} 
}
{ [ \psi] \ x := [t] \  
[\textit{ok}:\textsf{false}] \
[\textit{er}: \psi ]}
\and
\inferrule[\textsc{StorePtr}]
{ \psi \in \textsf{CF}_\text{sh} (\textsf{termS}(\psi) \cup  \textsf{termC}([t]:=t'))  \\
\psi_{\text{pure}} \models   t_a \approx t  
\\
 \psi = t_a \mapsto -  * \psi' 
}
{ [ \psi ] \  [t] := t' \ 
[\textit{ok}:    t \mapsto t'  * \psi'   ] \ 
[\textit{er}:\textsf{false}]
}
\and 
\inferrule[\textsc{StoreArr 1}]
{ \psi \in \textsf{CF}_\text{sh} (\textsf{termS}(\psi) \cup  \textsf{termC}(x:=[t]))  \\
\psi_{\text{pure}} \models   t_a \approx t 
\\
 \psi = \arr (t_a , {t'}_a)  * \psi' 
}
{ [ \psi ] \  [t] := t' \ 
[\textit{ok}:  
 t \mapsto t' * \arr (t+1, {t'}_a)   * \psi'   ] \ 
[\textit{er}:\textsf{false}]
}
\and
\inferrule[\textsc{StoreArr 2}]
{ \psi \in \textsf{CF}_\text{sh} (\textsf{termS}(\psi) \cup  \textsf{termC}(x:=[t]))  \\
\psi_{\text{pure}} \models   t_a < t * t  + 1 < {t'}_a 
\\
 \psi = \arr (t_a , {t'}_a)  * \psi' 
}
{ [ \psi ] \  [t] := t' \ 
[\textit{ok}: 
\arr (t_a, t) * t \mapsto t' * \arr (t+1, {t'}_a)  * \psi'    ] \ 
[\textit{er}:\textsf{false}]
}
\and
\inferrule[\textsc{StoreArr 3}]
{ \psi \in \textsf{CF}_\text{sh} (\textsf{termS}(\psi) \cup  \textsf{termC}(x:=[t]))   \\
\psi_{\text{pure}} \models    t + 1 \approx {t'}_a 
\\
 \psi = \arr (t_a , {t'}_a)  * \psi' 
}
{ [ \psi ] \  [t] := t' \ 
[\textit{ok}:  
\arr (t_a, t) * t \mapsto t'    * \psi'   ] \ 
[\textit{er}:\textsf{false}]
}
\and
\inferrule[\textsc{StoreEr}]
{ \psi \in \textsf{CF}_\text{sh} (\textsf{termS}(\psi) \cup  \textsf{termC}(x:=[t]))  \\
\psi_{\text{pure}} \models  b(t) \approx \texttt{null} 
}
{ [ \psi] \  [t] := t' \  
[\textit{ok}:\textsf{false}] \
[\textit{er}: \psi ]}
\end{mathpar}}
Here, $\theta$ denotes $[x := x']$.
\caption{Rules for the other Memory Manipulation}
    \label{fig: rules for Memory Manipulation}
\end{figure}

%% file: 3_wpo.tex
\section{Weakest Postcondition}\label{sec: body-wpo}
In this section, we describe the definition of the weakest postcondition $\text{\normalfont{WPO}}[\![P,\mathbb{C},\epsilon]\!] $
and introduce a function $\textsf{wpo}(P,\mathbb{C},\epsilon)$ that computes weakest
postconditions.
We begin by formally defining the weakest postcondition as follows.
\begin{definition}[Weakest Postcondition]\normalfont

The \emph{weakest postcondition} of a precondition $P$, a program $\mathbb{C}$,
 and an exit condition $\epsilon$ is defined as follows.
$$
\text{\normalfont{WPO}}[\![P,\mathbb{C},\epsilon]\!] = \{ \sigma'\in\textsc{State}^+ \mid \exists \sigma. 
\sigma \models P \land
(\sigma,\sigma') \in [\![\mathbb{C}]\!]_\epsilon \}
$$
\end{definition}

\subsection{Weakest Postcondition Generating Function \(\normalfont\textsf{wpo}\)}

We define the weakest postcondition generating function $\textsf{wpo}$ as follows.

 \begin{definition}[$\normalfont\textsf{wpo}$]
\label{def: def of wpo}\normalfont 
We simultaneously define $\textsf{wpo}$ and $\textsf{wpo}_{\rm sh}$ as
  follows.
  \begin{itemize}
   \item For an assertion $P$ such that $\textsf{cano} (P,\mathbb{C}) = \bigvee_{i \in I} \exists
	 \overrightarrow{x_i}. \psi_i$, define
	  $$
	 \textsf{wpo}( P , \mathbb{C}, \epsilon) =  
	 \bigvee_{i \in I} \normalfont \exists \overrightarrow{x_i}.
	 \textsf{wpo}_\text{sh} ( \psi_i,\mathbb{C}, \epsilon).
	 $$
   \item For $\psi\in\textsf{CF}_{\rm sh}(\textsf{termS}(\psi) \cup
	 \textsf{termC} (\mathbb{C}))$,
	 $\textsf{wpo}_\text{sh}(\psi,\mathbb{C},\epsilon)$ is defined
	 in Figures \ref{fig: wpo for clauses} and \ref{fig: wpo for
	 clauses heap manipulate}.
  \end{itemize}
 \end{definition}
Note that infinite disjunctions are required only when calculating \(\textsf{wpo}\) for \(\mathbb{C}^\star\), following the setting of \cite{lee2024relative}.  
Additionally, the cases for commands that do not manipulate heaps and blocks are the same as those in \cite{lee2024relative}.

\begin{figure}
\normalfont
\scriptsize
\begin{align*}
\textsf{wpo}_\text{sh}  (\psi ,\texttt{skip},\epsilon) &= 
\begin{cases}
\psi & \epsilon= \textit{ok} \\
\textsf{false} &\epsilon= \textit{er}
\end{cases}   \\
\textsf{wpo}_\text{sh}  ( \psi ,\texttt{error},\epsilon) &= 
\begin{cases}
\textsf{false}  & \textit{ok} \\
\psi & \textit{er}
\end{cases} \\
\textsf{wpo}_\text{sh} ( \psi , \texttt{local $x $ in $\mathbb{C}$} ,\epsilon) &=
\bigvee_{j \in J} \exists   x'' , \overrightarrow{x_j}  . \varphi_j  \\
& \text{$\Big( \textsf{wpo} (  \psi [x:=x']  ,  \mathbb{C}  ,\epsilon)
[x:=x''] [x' := x] = \bigvee_{j \in J} \exists \overrightarrow{x_j} . \varphi_j \Big)$ }\\
\textsf{wpo}_\text{sh} (\psi,\texttt{assume($B$)},\epsilon) &=
\begin{cases}
\psi * B  & \epsilon= \textit{ok} \\
\textsf{false} & \epsilon= \textit{er}
\end{cases} \\
\textsf{wpo}_\text{sh} (\psi,\mathbb{C}_1;\mathbb{C}_2,\textit{ok}) &= 
\textsf{wpo}
(\textsf{wpo}_\text{sh} (\psi ,\mathbb{C}_1, \textit{ok}), \mathbb{C}_2,\textit{ok}) 
\\
\textsf{wpo}_\text{sh} (\psi,\mathbb{C}_1;\mathbb{C}_2,\textit{er}) &= 
 \textsf{wpo}_\text{sh}(\psi,\mathbb{C}_1, \textit{er})  \lor
\textsf{wpo}
(\textsf{wpo}_\text{sh} (\psi,\mathbb{C}_1, \textit{ok}), \mathbb{C}_2,\textit{er})
\\
\textsf{wpo}_\text{sh} (\psi,\mathbb{C}^\star,\textit{ok}) &= 
\bigvee_{n \in \mathbb{N}} \Upsilon (n) 
\quad\text{($\Upsilon (0) =\psi$  and 
$\Upsilon (n+1) = \textsf{wpo} (\Upsilon (n),\mathbb{C},\textit{ok})$)}\\
\textsf{wpo}_\text{sh} (\psi,\mathbb{C}^\star,\textit{er}) &= 
\bigvee_{n \in \mathbb{N}} \textsf{wpo} (\Upsilon (n), \mathbb{C}, \textit{er}) \\
\textsf{wpo}_\text{sh} (\psi ,\mathbb{C}_1 + \mathbb{C}_2,\epsilon) &= 
\textsf{wpo}_\text{sh}(\psi ,\mathbb{C}_1,\epsilon) \lor \textsf{wpo}_\text{sh}(\psi,\mathbb{C}_2,\epsilon) \\
\textsf{wpo}_\text{sh} (  \psi , x := t, \epsilon) &= 
\begin{cases}
\exists  x' .  \psi \theta  * x \approx  t  \theta   & \epsilon=\textit{ok} \\
\textsf{false} & \epsilon= \textit{er}
\end{cases} \\
\textsf{wpo}_\text{sh} (\psi ,x:=\texttt{*},\epsilon) &=
\begin{cases}
\exists x'. \psi \theta & \epsilon=\textit{ok} \\
\textsf{false} & \epsilon= \textit{er}
\end{cases} 
\end{align*}
\caption{$\normalfont\textsf{wpo}_\text{sh}$ for commands that do not manipulate heaps}
\label{fig: wpo for clauses}
\end{figure}

\newgeometry{left=1cm, right=1cm, top=3cm, bottom=3cm} 

\begin{figure*}
\normalfont

{
\scriptsize
\begin{align*}
&\textsf{wpo}_\text{sh} (\psi ,x:=\texttt{alloc($t$)},\textit{ok}) = \\ 
& \exists x'.\bigvee_{0\le\alpha\le\beta\le N} \textsf{rep}_\alpha^\beta \bigg(  (\psi\theta * \arr(x,x+t\theta) * \varphi) \lor\\
& \bigvee_{1\le j\le k\le m} \bigg( (\psi'\theta
 * \bigast_{i=1}^{j-1} \bararr(t_i\theta,t'_i\theta)
 * \bararr(t_j\theta,x)
 * \arr(x,x+t\theta)  * \bararr(x+t\theta,t'_k\theta)
 * \bigast_{i=k+1}^{m} \bararr(t_i\theta,t'_i\theta) 
 * t_j\theta<x<t'_j\theta * t_k\theta<x+t\theta< t'_k\theta * \varphi) \lor \\
 & \qquad\qquad (\psi'\theta
 * \bigast_{i=1}^{j-1} \bararr(t_i\theta,t'_i\theta)
 * \arr(x,x+t\theta)
 * \bararr(x+t\theta,t'_k\theta)
 * \bigast_{i=k+1}^{m} \bararr(t_i\theta,t'_i\theta) 
 * t'_{j-1}\theta\le x\le t_j\theta * t_k\theta<x+t\theta< t'_k\theta * \varphi) \lor \\
 & \qquad\qquad (\psi'\theta
 * \bigast_{i=1}^{j-1} \bararr(t_i\theta,t'_i\theta)
 * \bararr(t_j\theta,x)
 * \arr(x,x+t\theta)
 * \bigast_{i=k+1}^{m} \bararr(t_i\theta,t'_i\theta) 
 * t_j\theta<x<t'_j\theta * t'_k\theta\le x+t\theta\le t_{k+1}\theta * \varphi) \lor \\
 & \qquad\qquad (\psi'\theta
 * \bigast_{i=1}^{j-1} \bararr(t_i\theta,t'_i\theta)
 * \arr(x,x+t\theta)
 * \bigast_{i=k+1}^{m} \bararr(t_i\theta,t'_i\theta) 
 * t'_{j-1}\theta\le x\le t_{j}\theta * t'_k\theta\le x+t\theta\le t_{k+1}\theta * \varphi) \bigg)
 \bigg)\\
& \text{// $\psi = \bigast_{i=1}^{m} \bararr(t_i, {t'}_i) * \psi'$ with
 $\psi'$ containing no $\bararr$, $\psi_{\text{pure}} \models 
 \bigast_{i=1}^{m-1} t_i < t_{i+1}$, and
 $\psi_{\rm pure}\models {\tt null}R_1 u_1 R_2\cdots R_N u_N$ for $\textsf{termS}^-(\psi)=\{u_1,\ldots,u_N\}$.
 }  \\
 & \text{// $\varphi=b(x) \approx x * e(x) \approx
 x+t\theta * u_\alpha\theta< x \le u_{\alpha+1}\theta *
u_\beta\theta<x+t\theta\le u_{\beta+1}\theta$}  \\
&\textsf{wpo}_\text{sh} (\psi ,x:=\texttt{alloc($t$)},\textit{er}) = \textsf{false} \\
&\textsf{wpo}_\text{sh} ( \psi , \texttt{free($t$)},\textit{ok}) = \\
&\begin{cases}
\exists y . 
(\arr (t_1 , t)   * \bararr ( t , y ) * 
\arr(y, {t'}_k )  * \psi'  )
\tau 
& 
\text{if $\psi = \arr (t_1, {t'}_1 ) *  \bigast_{i=2}^{k-1} \Arr (t_i, {t'}_i)_{\alpha_i}
 *\arr (t_k, {t'}_k ) *  \psi'$ } \\
& 
\ \ \ \ \& \   \psi_{\text{pure}} \models 
t_1 < b(t) < {t'}_1 * t_k < e(t) < {t'}_k 
*
b(t) \approx t *  
\bigast_{i=1}^{k-1} {t'}_i \approx t_{i+1} * \bigast_{i=1}^{k-1} {t}_i <
  t_{i+1} \text{ for some  \( k \) and $\psi'$} \\ 
\exists y . 
( \bararr ( t , y ) * \arr(y, {t'}_k )  * \psi' ) 
\tau 
& 
\text{if $\psi = \bigast_{i=1}^{k-1} \Arr (t_i, {t'}_i)_{\alpha_i} * \arr(t_k , {t'}_k) * \psi'$} \\
& 
\ \ \ \ \& \   \psi_{\text{pure}} \models 
t_1 \approx b(t) * t_k < e(t) < {t'}_k 
*
b(t) \approx t * 
\bigast_{i=1}^{k-1} {t'}_i \approx t_{i+1}  * \bigast_{i=1}^{k-1} {t}_i
  < t_{i+1} \text{ for some  \( k \) and $\psi'$}   \\ 
\exists y .  ( \bararr ( t , y ) *   \psi' ) \tau 
& 
\text{if $\psi = \bigast_{i=1}^{k} \Arr (t_i, {t'}_i)_{\alpha_i} * \psi'$} \\
& 
\ \ \ \ \& \   \psi_{\text{pure}} \models 
t_1 \approx  b(t) *   e(t) \approx {t'}_k 
*
b(t) \approx t * 
\bigast_{i=1}^{k-1} {t'}_i \approx t_{i+1}  * \bigast_{i=1}^{k-1} {t}_i
  < t_{i+1} \text{ for some  \( k \) and $\psi'$} \\
\exists y . (  \arr (t_1 , t)   * \bararr (t , y )  *   \psi' ) 
\tau 
& 
\text{if $\psi = \arr(t_1, {t'}_1 ) *   \bigast_{i=2}^{k} \Arr (t_i, {t'}_i)_{\alpha_i} * \psi' $} \\
& 
\ \ \ \ \& \   \psi_{\text{pure}} \models 
t_1 < b(t) < {t'}_1 *   e(t) \approx  {t'}_k 
*
b(t) \approx t * 
\bigast_{i=1}^{k-1} {t'}_i \approx t_{i+1}  * \bigast_{i=1}^{k-1} {t}_i
  < t_{i+1} \text{ for some  \( k \) and $\psi'$} \\
\textsf{false}
&
\text{otherwise}
\end{cases}\\
&\textsf{wpo}_\text{sh} (\psi , \texttt{free($t$)},\textit{er}) =
\begin{cases}
\psi 
& \text{if $\psi_{\text{pure}} \models  b(t) \not\approx t \lor b(t) \approx \texttt{null}  $ } \\
\textsf{false}  &   \text{otherwise}
\end{cases} \\
&\textsf{wpo}_\text{sh} ( \psi , x:=[t] ,\textit{ok}) =  
\begin{cases}
\exists x'. 
\psi \theta * x \approx u \theta  
& 
\text{if $  \psi = t_a \mapsto u * \psi'  $
}  \ \& \   \psi_{\text{pure}} \models   t_a \approx t \text{ for some
  $t_a$,  \( u \) and $\psi'$}\\ 
\exists x' . \psi \theta 
& 
\text{if $\psi =\arr (t_a , {t'}_a)  * \psi'$}  
  \ \& \   \psi_{\text{pure}} \models   t_a \leq t < {t'}_a
  \text{ for some $t_a$, ${t'}_a$ and $\psi'$}\\ 
\textsf{false}
&
\text{otherwise}
\end{cases}\\
&\textsf{wpo}_\text{sh} ( \psi , x:=[t]  ,\textit{er}) =
\begin{cases}
\psi 
 & 
\text{if $\psi_{\text{pure}} \models   b(t) \approx \texttt{null} $ } \\
\textsf{false} & \text{otherwise}
\end{cases} \\
&\textsf{wpo}_\text{sh} ( \psi , [t]:=t'  ,\textit{ok}) =  
\begin{cases}
  t \mapsto t'  * \psi'  
& 
  \text{if $\psi = t_a \mapsto -  * \psi'   $}
  \ \& \   \psi_{\text{pure}} \models    t_a \approx t
  \text{ for some $t_a$  and $\psi'$}\\ 
  t \mapsto t' * \arr (t+1, {t'}_a)   * \psi'
& 
\text{if $\psi = \arr (t_a , {t'}_a)  * \psi'    $} 
  \ \& \ \psi_{\text{pure}} \models   t_a \approx t
  \text{ for some $t_a$, ${t'}_a$ and $\psi'$}\\ 
\arr (t_a, t) * t \mapsto t' * \arr (t+1, {t'}_a)  * \psi'   
& 
\text{if $\psi = \arr (t_a , {t'}_a)  * \psi'$} 
  \ \& \ \psi_{\text{pure}} \models    t_a < t * t  + 1 < {t'}_a
  \text{ for some $t_a$, ${t'}_a$ and $\psi'$}\\ 
\arr (t_a, t) * t \mapsto t'    * \psi'   
& 
\text{if $\psi =  \arr (t_a , {t'}_a)  * \psi'$} 
  \ \& \ \psi_{\text{pure}} \models  t + 1 \approx {t'}_a
  \text{ for some $t_a$, ${t'}_a$ and $\psi'$}\\ 
\textsf{false} 
&
\text{otherwise}
\end{cases}\\
&\textsf{wpo}_\text{sh} ( \psi ,[t]:=t'  ,\textit{er}) =
\begin{cases}
\psi 
 & 
\text{if $\psi_{\text{pure}} \models   b(t) \approx \texttt{null} $ } \\
\textsf{false} & \text{otherwise}
\end{cases} \\
\end{align*}
}
 Here, $\theta$ denotes $[x := x']$.
 $\textsf{rep}_\alpha^\beta(\chi)$ is the formula obtained from $\chi$ by
 replacing $b(u_i\theta)$ and $e(u_i\theta)$ with ${\tt null}$ for
 $\alpha<i\le\beta$.
 $\tau$ denotes $[T_b(\psi, t) := t][T_e(\psi, t) := y]$.

\caption{$\normalfont\textsf{wpo}_\text{sh}$ for commands manipulating heaps}
\label{fig: wpo for clauses heap manipulate}
\end{figure*}

\restoregeometry 

%% file: 4_completeness.tex
\section{Soundness and Relative Completeness}\label{sec: body-completeness}

In this section, we prove the soundness and relative completeness of ISL.  

For soundness, we use the proof of existing work \cite{raad2020local,lee2024relative}
for commands that do not
manipulate heaps.
For commands that manipulate
heaps, we utilize
the expressiveness of the weakest postconditions
(Proposition \ref{prop: expressiveness of WPO calculus}), which states
that the set of exact states satisfying $\textsf{wpo}(P, \mathbb{C},
\epsilon)$ is exactly $\text{WPO}[\![P, \mathbb{C}, \epsilon]\!]$.

The relative completeness is proven by the expressiveness and the property
that the triple $[P] \ \mathbb{C} \ [\epsilon :
\textsf{wpo}(P, \mathbb{C}, \epsilon)] $ is always derivable in our
proof system (Proposition \ref{prop: for all p, c, epsilon, we have
wpo}).

To prove expressiveness, we present the following lemmas.
Lemmas \ref{lma: Substitution for assignment}, \ref{lma: quantifier is free in and out of WPO}, \ref{lma: from bigvee to bigcup}, and \ref{lma: sensei's one shot lemma}  
are the same as those used in the expressiveness proof in \cite{lee2024relative}, except that our versions use models with blocks.
Since the results remain consistent with our current setting, we will directly apply these lemmas to prove the expressiveness of our system.

\begin{lemma}[Substitution with models with blocks]\label{lma: Substitution for assignment} 
\begin{align*}
     (s,h,B) \models  \psi [x:=t] 
   \iff (s[x \mapsto [\![ t  ]\!]_{s,B}  ], h, B) \models  \psi
\end{align*}

\end{lemma}

This lemma is a traditional substitution lemma, also used to prove expressiveness in \cite{lee2024relative}.  
Since substitution does not alter the heap or block, it can naturally be extended to the expanded models with blocks.

\begin{lemma}\label{lma: quantifier is free in and out of WPO}
If $x \notin \normalfont\textsf{termC}(\mathbb{C})$, then the following holds for any 
$(s', h',B')$, $\psi$, $\mathbb{C}$ and $\epsilon$.
{\normalfont\begin{align*}
&(s', h',B') \in \text{\normalfont{WPO}}[\![\exists x. \psi,\mathbb{C},\epsilon]\!] \iff \\
&\exists v \in \textsc{Val} . 
(s' [x \mapsto v] , h' , B') \in \text{\normalfont{WPO}}[\![\psi,\mathbb{C},\epsilon]\!] 
\end{align*}}
\end{lemma}

\begin{lemma}\label{lma: from bigvee to bigcup}
$$
\text{\normalfont{WPO}}[\![ \bigvee_{i \in I}  P ,\mathbb{C},\epsilon ]\!] \iff 
\bigcup_{i \in I} \text{\normalfont{WPO}}[\![ P,\mathbb{C},\epsilon]\!]
$$
\end{lemma}

\begin{lemma}\label{lma: sensei's one shot lemma}
Suppose $\normalfont\textsf{cano} (P, \mathbb{C}) = \bigvee_{i \in I} \exists \overrightarrow{x_i} . \psi_i$.
If for all $i \in I$, $(s',h',B')$ and $\epsilon$, 
$\normalfont
(s', h',B') \models \textsf{wpo}_\text{sh} (\psi_i, \mathbb{C}, \epsilon) \iff 
(s', h',B') \in \text{WPO} [\![ \psi_i ,  \mathbb{C}, \epsilon  ]\!]$,
then for any $(s',h')$ and $\epsilon$, 
$\normalfont
\text{\textnormal{(}}s', h',B'\text{\textnormal{)}} \models \textsf{wpo} (P, \mathbb{C}, \epsilon) \iff 
\text{\textnormal{(}}s', h',B'\text{\textnormal{)}} \in \text{WPO} [\![ P ,  \mathbb{C}, \epsilon  ]\!] $ holds.
\end{lemma}

These lemmas are extended versions of the lemmas in \cite{lee2024relative},  
with the only difference being that the models in this version include block information \(B'\).  
However, this difference does not affect the validity of the lemmas.

\begin{proposition}[Expressiveness]\label{prop: expressiveness of WPO calculus}
\begin{center}
$\normalfont\forall \sigma'\in\textsc{State}^+. \sigma' \in \text{\normalfont{WPO}}[\![P,\mathbb{C},\epsilon]\!] \iff
\normalfont \sigma' \models 
\textsf{wpo} (
P ,\mathbb{C},\epsilon)$
\end{center}
\end{proposition}
\begin{proof}
With Lemma \ref{lma: sensei's one shot lemma}, to prove this proposition, it is sufficient to show that
\[
\forall \sigma'. \sigma' \in \text{\normalfont{WPO}}[\![\psi, \mathbb{C}, \epsilon]\!] 
\iff 
\sigma' \models \textsf{wpo}_\text{sh}(\psi, \mathbb{C}, \epsilon)
\]
holds for canonical symbolic heaps $\psi$.
For the cases where $\textsf{wpo}_\text{sh}$ recursively generates
$\textsf{wpo}$, we assume as an induction hypothesis that the above
holds.

The cases for atomic commands that do not manipulate heaps share the
 same proofs as in \cite{lee2024relative}.
Therefore, we provide proofs only for heap-manipulating commands in the Appendix.
\end{proof}

By applying expressiveness (Proposition \ref{prop: expressiveness of WPO calculus}),  
we can demonstrate the soundness of our proof system.

\begin{theorem}[Soundness of ISL]\label{lma: local soundness}
\begin{center}
For any $P,\mathbb{C},\epsilon,Q$, if 
$\vdash  [P] \ \mathbb{C} \ [\epsilon: Q]$, then
$\models [P] \ \mathbb{C} \ [\epsilon: Q]$.
\end{center}
\end{theorem}

\begin{proof}
 By induction on $\vdash  [P] \ \mathbb{C} \ [\epsilon: Q]$.
 We follow 
 the previous work \cite{lee2024relative} for commands that do not
 manipulate heaps.

 For the cases of heap-manipulating rules, we use the expressiveness
 (Proposition \ref{prop: expressiveness of WPO calculus}).
 Assume $\vdash [\psi] \ \mathbb{C} \ [\epsilon: Q]$ and $\sigma'
 \models Q$.  By definition of the inference rules and
 $\textsf{wpo}_{\rm sh}$,
 it is easy to see $Q\models
 \textsf{wpo}(\psi,\mathbb{C},\epsilon)$. Hence, $\sigma'\models
 \textsf{wpo}(\psi,\mathbb{C},\epsilon)$.  
By expressiveness, \(\sigma' \in \text{WPO} \llbracket \psi, \mathbb{C}, \epsilon \rrbracket\) holds,  
which states that there exists \(\sigma \models \psi\) such that \((\sigma, \sigma') \in \llbracket \mathbb{C} \rrbracket_\epsilon\).

 Thus, we conclude $\models [P] \ \mathbb{C} \ [\epsilon: Q]$.
\end{proof}

Next, we show the following lemma to prove Proposition \ref{prop: for
all p, c, epsilon, we have wpo}.

\begin{lemma}\label{lma: for all derivation if and then}
For any $P$, $\mathbb{C}$, $\epsilon$, let us say $
\normalfont\textsf{cano}(P, \mathbb{C}) = \bigvee_{i \in I} \exists
\overrightarrow{x_i}. \psi_i$.  Then, the following holds.
$$\normalfont \forall i \in I. \vdash [\psi_i] \ \mathbb{C} \ [ \epsilon
: \textsf{wpo}_\text{sh} (\psi_i,\mathbb{C},\epsilon) ] \
\Longrightarrow \ \vdash [P] \ \mathbb{C} \ [ \epsilon :
\textsf{wpo}(P,\mathbb{C},\epsilon) ]$$
\end{lemma}
\begin{proof}
The lemma proven here is identical to the lemma in \cite{lee2024relative},  
and the proof is also identical since the $\textsc{Exist}$ and $\textsc{Disj}$ rules are the same  
in both \cite{lee2024relative} and this paper.
\end{proof}

\begin{proposition}\label{prop: for all p, c, epsilon, we have wpo}
For any $P,\mathbb{C},\epsilon$,
$\normalfont\vdash [P] \ \mathbb{C} \ 
[ \epsilon :  
\textsf{wpo}(P,\mathbb{C},\epsilon) ]$.
\end{proposition}
\begin{proof}
With Lemma \ref{lma: for all derivation if and then}, it is sufficient
 to show that for any $\psi, \mathbb{C}, \epsilon$, we have $\vdash
 [\psi] \ \mathbb{C} \ [ \epsilon : \textsf{wpo}_\text{sh}(\psi,
 \mathbb{C}, \epsilon)]$.  This statement is proven by induction on
 $\mathbb{C}$ and follows the same proof for rules that do not
 manipulate heaps in this paper and \cite{lee2024relative}.
 For rules that manipulate heaps, we can directly prove this lemma
from the definition of $\textsf{wpo}_{\rm sh}$ using
 the $\textsc{Disj}$ rule.
\end{proof}

By Proposition \ref{prop: expressiveness of WPO calculus}
and Proposition \ref{prop: for all p, c, epsilon, we have wpo},
we prove the relative completeness of ISL.

\begin{theorem}[Relative Completeness]
$$\text{For any $P,\mathbb{C},\epsilon,Q$, if 
$\models [P] \ \mathbb{C} \ [\epsilon: Q]$, then
$\vdash  [P] \ \mathbb{C} \ [\epsilon: Q]$.}$$
\end{theorem}

\begin{proof}
The proof of this theorem is identical to that in \cite{lee2024relative}, as shown below.
$$
\inferrule*[right=\textsc{Cons}]
{
{\inferrule*[]
{\normalfont\text{Applying Proposition \ref{prop: for all p, c, epsilon, we have wpo}} }
{\normalfont\vdash [P] \ \mathbb{C} \ 
[ \epsilon :  
\textsf{wpo}(P,\mathbb{C},\epsilon) ]}} 
\ \
{
\inferrule*[right=\normalfont\text{Proposition \ref{prop: expressiveness of WPO calculus}}]
{ \models [P] \ \mathbb{C} \ [\epsilon: Q] }
{\normalfont Q \models 
\textsf{wpo}(P,\mathbb{C},\epsilon)}
}
}
{\vdash  [P] \ \mathbb{C} \ [\epsilon: Q]}
$$
\end{proof}

%% file: 5_related_work.tex
\section{Related Work}
\subsection{Separation Logic with array predicates}
Separation logic with array predicates has been extensively studied in the context of decidability in entailments and biabduction for symbolic heaps with arrays, with notable commonalities in the works of \cite{brotherston2017biabduction,kimura2021decidability}, and \cite{holik2022low}. 
The approaches in \cite{brotherston2017biabduction,kimura2021decidability} share a foundational reliance on a variable-length representation for arrays, using the $\arr(x, y)$ style, and incorporate arithmetic.
Despite differences in the specific conditions of their decision
procedures, their aims and settings are closely aligned. 
Based on these foundations, our work extends this line of research by exploring relative completeness in Incorrectness Separation Logic with array predicates.

We are also inspired by the work of \cite{holik2022low}, particularly their use of models incorporating detailed block information for memory manipulation, especially for allocation and deallocation. While \cite{holik2022low} define array predicates based on start points and sizes, our approach adopts the variable-length $\arr(x, y)$ style, consistent with the settings of \cite{brotherston2017biabduction} and \cite{kimura2021decidability}, to handle arrays flexibly.

\subsection{Relative completeness of under-approximation logic}
Relative completeness results have been established for various
under-approximation logics, including Incorrectness Logic (IL)
\cite{o2019incorrectness} and Reverse Hoare Logic (RHL)
\cite{de2011reverse}. In IL, the proof of relative completeness relies
on using {semantic predicates}, which assume that expressiveness
holds. In contrast, RHL \cite{de2011reverse}, using syntactic
predicates, achieves relative completeness by generating the weakest postconditions via the function $\textsf{wpo}$ and demonstrating its expressiveness.

Building on these foundational methods, our prior work
\cite{lee2024relative} established the relative completeness of
Incorrectness Separation Logic (ISL) in a system limited to points-to
predicates and negative heap predicates. Both \cite{lee2024relative} and
the current work follow the approach of \cite{de2011reverse} for proving
relative completeness.
However, this work takes a different approach, unlike
\cite{lee2024relative}, which restricts its canonical form to equalities
and inequalities due to the absence of arithmetic. It extends the scope by incorporating arithmetic into the logical framework, similar to \cite{brotherston2017biabduction} and \cite{kimura2021decidability}. This enhancement allows the canonical form to account for all possible sequences of strict inequalities and equalities.

%% file: 6_conclusion.tex
\section{Conclusions and Future work}\label{sec: Conclusions and Future work}

We introduced an extended Incorrectness Separation Logic (ISL) with
support for variable-length array predicates and pointer
arithmetic. Additionally, we established the relative completeness of
this enhanced ISL by constructing functions $\textsf{wpo}$ generating
the weakest postcondition and demonstrating their expressiveness.

For future work, while infinitary syntax aligns well with our goals of proving relative completeness, it is not ideally suited for the developing a practical automatic theorem prover, a limitation shared by our earlier work \cite{lee2024relative}. 
Additionally, we aim to simplify ISL with arrays by introducing more straightforward proof rules. 
The current heap-manipulation rules are numerous, and their postconditions, which correspond to the output of the \textsf{wpo} function, are notably complex. 
We acknowledge that this complexity is a trade-off made to achieve simpler proofs for relative completeness. 
Therefore, our goal is to develop a version of ISL with arrays that retains the expressiveness and relative completeness properties of this work while employing simpler proof rules.

%% file: 7_appendix_expressiveness_arxiv.tex
\newpage

\section{Proof of Proposition \ref{prop: expressiveness of WPO calculus} for $x:={\tt alloc}(t)$}
\label{sec: expressiveness proof of atomic command in appendix (alloc)}

We prove that $(s', h', B') \models \textsf{wpo}( \psi , x :=
\texttt{alloc($t$)} , \textit{ok} ) \iff (s', h', B') \in \text{WPO}
[\![ \psi ,x := \texttt{alloc($t$)} , \textit{ok} ]\!]  $ holds.  By
definition, $(s', h', B') \in \text{WPO} [\![ \psi , x :=
\texttt{alloc($t$)} , \textit{ok} ]\!]  $ is equivalent to $\exists
(s,h,B) \models \psi , l \in \textsc{Val}.  (\bigcup_{[u,v) \in B}
[u,v)) \cap [ l , l + [\![ t ]\!]_{s,B} ) = \emptyset \land (s', h', B')
= (s [x \mapsto l] , h[l + i \mapsto - ]_{i=0}^{[\![ t ]\!]_{s,B} - 1 }
, B \cup \{ [l , l + [\![ t ]\!]_{s,B} ) \} ) $.  
In this case, we suppose the following.
\begin{itemize}
    \item $\psi = \left( \bigast_{i=1}^{m} \bararr(t_i, {t'}_i) \right) * \psi'$, where  
          $\psi'$ contains no $\bararr$.
    \item $\psi_{\text{pure}} \models t_i < t_{i+1}$ for $1 \leq i < m$.
    \item $\psi_{\text{pure}} \models {\tt null}\ R_1\ u_1\ R_1\ u_1\ R_2\ \cdots\ R_N u_N$,  
          where $R_i \in \{<,\approx\}$ and $\termS^-(\psi) = \{u_i \mid 1 \leq i \leq N\}$.
\end{itemize}

Firstly, we prove 
$(s', h', B') \models \textsf{wpo}(  \psi , x := \texttt{alloc($t$)}  , \textit{ok} ) \Rightarrow
(s', h', B') \in \text{WPO} [\![  \psi , x := \texttt{alloc($t$)}  , \textit{ok} ]\!]  $.

\subsection{$(s', h', B') \models \textsf{wpo}(  \psi , x := \texttt{alloc($t$)}  , \textit{ok} ) \Rightarrow
(s', h', B') \in \text{WPO} [\![  \psi , x := \texttt{alloc($t$)}  , \textit{ok} ]\!]  $}

By $(s',h',B')\models\textsf{wpo}( \psi , x := \texttt{alloc($t$)},
\textit{ok} )$, which is of the form $\exists x'.\chi$, there exists
$v'\in{\normalfont\textsc{Val}}$ such that $(s'[x'\mapsto
v'],h',B')\models\chi$, and hence it satisfies one of disjuncts in
$\chi$. Let $s''$ be $s'[x'\mapsto v']$.  
For brevity, we use the notation $\theta$ for $[x := x']$ and  
$\tilde{t}$ for $\textsf{rep}_\alpha^\beta(t\theta)$ for some $t$, respectively.

\subsubsection{Case 1}

We demonstrate the case of
{
\begin{align*}
 (s'',h',B')\models &
  \tilde{\psi}'
 * \bigast_{i=1}^{j-1} \bararr(\tilde{t_i},\tilde{t}'_i)
 * \bararr(\tilde{t}_j,x)
 * \arr(x,x+\tilde{t}) \\
 & * \bararr(x+\tilde{t},\tilde{t}'_k)
 * \bigast_{i=k+1}^{m} \bararr(\tilde{t}_i,\tilde{t}'_i)\\
 & * \tilde{t}_j<x<\tilde{t}'_j * \tilde{t}_k<x+\tilde{t}< \tilde{t}'_k \\
 & * b(x) \approx x * e(x) \approx x+\tilde{t} \\
 & * \tilde{u}_\alpha < x \le \tilde{u}_{\alpha+1} *
 \tilde{u}_\beta<x+\tilde{t}\le \tilde{u}_{\beta+1},
\end{align*}
}
for some $\alpha,\beta,j,k$.
$[\inter{x}_{s'',B},\inter{x+\tilde{t}}_{s'',B})\in B'$ holds, since
$(s'',B')\models b(x)\approx x * e(x)\approx x+\tilde{t}$.
Define $s=s'[x\mapsto v']$ and $B= B' \setminus \{ [ [\![ x
]\!]_{s'',B'} , [\![ x + \tilde{t} ]\!]_{s'',B'} ) \}$.

We can say that $\inter{\tilde{t}}_{s'',B'} = \inter{t}_{s,B}$  
since the following derivation holds.  

Since $t$ is $b,e$-free, the following is true:
\[
\inter{\tilde{t}}_{s'',B'} = \inter{t \theta}_{s'',B'} = \inter{t \theta}_{s'',B}.
\]
Additionally, by using Lemma \ref{lma: Substitution for assignment} and the fact that $x'$ is fresh, the following holds:
\[
\inter{t \theta}_{s'',B} = \inter{t [x := x']}_{s'[x' \mapsto v'],B}
= \inter{t }_{s'[x' \mapsto v'][x \mapsto v'],B}
= \inter{t }_{s' [x \mapsto v'],B}
= \inter{t }_{s,B}.
\]
For any $u_i\in\termS^-(\psi)$, since $u_i$ contains neither $b$ nor $e$,
$\tilde{u_i}=u_i \theta$ and its semantics does not depend on $B'$.
Hence,
$\inter{\tilde{u_i}}_{s'',B'}=\inter{u_i}_{s,B}$ by $x'\not\in \textsf{fv}(u_i)$
similarly to the case of $\inter{\tilde{t}}_{s'',B'} = \inter{t}_{s,B}$.
Furthermore, we have $\inter{x}_{s'',B'}=s'(x)$ since $x$ and $x'$ are
distinct. Hence, $B=B'\setminus\{[s'(x),s'(x)+\inter{t}_{s,B})\}$ since
$t\in\termS^-(\psi)$.
Then, for any $u_i\in\termS^-(\psi)$, we have the following:
\begin{itemize}
 \item If $\alpha<i\le \beta$, since $(s'',B')$ satisfies $x \le
       \tilde{u_i} < x+\tilde{t}$, 
       the value
       $\inter{\tilde{u_i}}_{s'',B'}=
       \inter{u_i\theta}_{s'',B'}=\inter{u_i}_{s,B}$ 
       is in the block
       $[s'(x),s'(x)+\inter{t}_{s,B})$, and hence
       $\inter{b(u_i)}_{s,B}=\inter{e(u_i)}_{s,B}=0$ 
       by $B=B'\setminus\{[s'(x),s'(x)+\inter{t}_{s,B})\}$.
 \item Otherwise, 
       $\inter{b(u_i\theta)}_{s'',B'}=\inter{b(u_i)}_{s,B}$ and
       $\inter{e(u_i\theta)}_{s'',B'}=\inter{e(u_i)}_{s,B}$.
\end{itemize}
Therefore, by the definition of $\textsf{rep}_\alpha^\beta$,
$\inter{\tilde{r}}_{s'',B'} = \inter{r}_{s,B}$ holds for any $r\in{\sf termS}(\psi)$.
Hence, for any pure formula $\pi$ in $\psi$, $(s'', B')\models \tilde{\pi}$ iff
$(s,B)\models \pi$.

Define $h$ as{
\begin{align*}
h&=h'[l \mapsto \bot]_
 { l = [\![ x ]\!]_{s'', B'} }
 ^{ [\![ \tilde{{t'}}_j ]\!]_{s'', B'} - 1 }
  [l \mapsto \ \uparrow ]_
 { l = [\![ \tilde{{t'}}_j ]\!]_{s'', B'}  }
 ^{ [\![ \tilde{t}_{j+1} ]\!]_{s'', B'} - 1 }\\
 &
 \Bigl(
 [l \mapsto \bot]_
 { l = [\![ \tilde{t}_i ]\!]_{s'', B'}   }
 ^{ [\![ \tilde{{t'}}_i ]\!]_{s'', B'} - 1   }
  [l \mapsto \ \uparrow ]_
 { l = [\![ \tilde{{t'}}_i ]\!]_{s'', B'}   }
 ^{ [\![ \tilde{t}_{i+1} ]\!]_{s'', B'}  - 1  }
 \Bigr)_{i=j+1}^{k-1} \\
 &   [l \mapsto \bot]_
 { l =  [\![ \tilde{t}_k ]\!]_{s'', B'}    }
 ^{ [\![  x + \tilde{t} ]\!]_{s'', B'} - 1   },
\end{align*}
}
which can be rewritten as
{
\begin{align*}
h&=h'[l \mapsto \bot]_
 { l = s'(x) }
 ^{ [\![ {t'}_j ]\!]_{s, B} - 1 }
  [l \mapsto \ \uparrow ]_
 { l = [\![ {{t'}}_j ]\!]_{s, B}  }
 ^{ [\![ {t}_{j+1} ]\!]_{s, B} - 1 } \\
  &\Bigl(
 [l \mapsto \bot]_
 { l = [\![ {t}_i ]\!]_{s, B}   }
 ^{ [\![ {{t'}}_i ]\!]_{s, B} - 1   }
  [l \mapsto \ \uparrow ]_
 { l = [\![ {{t'}}_i ]\!]_{s, B}   }
 ^{ [\![ {t}_{i+1} ]\!]_{s, B}  - 1  }
 \Bigr)_{i=j+1}^{k-1}\\
  & [l \mapsto \bot]_
 { l =  [\![ {t}_k ]\!]_{s, B}    }
 ^{ s'(x)+[\![ {t} ]\!]_{s, B} - 1   }.
\end{align*}
}
Note that this is well-defined since $(s'',B')\models x<\tilde{t}'_j
* \tilde{t}_k <x +\tilde{t}$ and $(s'',B')\models \tilde{t}_i <
\tilde{t}_{i+1}$ for $j\le i<k$ by ${\psi}_{\text{pure}}\models
t_i<t_{i+1}$. 
Moreover, 
$(h,B)$ is exact since
$\textsf{dom}_+(h)=\textsf{dom}_+(h')\setminus[s'(x),s'(x)+\inter{t}_{s,B})$.

Then, we have
{
\begin{align*}
 (s'',h,B) \models
 & \tilde{\psi}'
 * \bigast_{i=1}^{j-1} \bararr(\tilde{t_i},\tilde{t}'_i)
 * \bararr(\tilde{t}_j,x)\\
 &
 * \bararr(x,\tilde{{t'}}_{j})
 * \bigast_{i=j+1}^{k-1} \bararr(\tilde{t}_i,\tilde{{t'}}_i)
 * \bararr(\tilde{{t'}}_{k},x+\tilde{t})\\
 & * \bararr(x+\tilde{t},\tilde{t}'_k)
 * \bigast_{i=k+1}^{m} \bararr(\tilde{t}_i,\tilde{t}'_i),
\end{align*}
}
and then, appending $\bararr(\tilde{t}_j,x) * \bararr(x,\tilde{{t'}}_{j})$
and $\bararr(\tilde{{t'}}_{k},x+\tilde{t}) *
\bararr(x+\tilde{t},\tilde{t}'_k)$, we have
$ (s'',h,B) \models
  \tilde{\psi}'
 * \bigast_{i=1}^{m} \bararr(\tilde{t_i},\tilde{t}'_i)$,
and hence,
$(s,h,B) \models
  \psi'
 * \bigast_{i=1}^{m} \bararr(t_i,{t'}_i)$.
That means $(s,h,B)\models \psi$.

Since $s'=s[x\mapsto s'(x)]$, we have
$((s,h,B),(s',h',B'))\in\inter{x:={\tt alloc}(t)}_{\it ok}$.
Hence, we have $(s',h',B')\in\text{WPO}[\![\psi, x := \texttt{alloc}(t), \textit{ok}]\!]$.

\subsubsection{Case 2}

We demonstrate the case of

{
\begin{align*}
 (s'',h',B')\models &
  \tilde{\psi}'
 * \bigast_{i=1}^{j-1} \bararr(\tilde{t_i},\tilde{t}'_i)
 * \arr(x,x+\tilde{t}) \\
 & * \bararr(x+\tilde{t},\tilde{t}'_k)
 * \bigast_{i=k+1}^{m} \bararr(\tilde{t}_i,\tilde{t}'_i)\\
 & * \tilde{t'}_{j-1} \leq  x \leq \tilde{t}_j * \tilde{t}_k<x+\tilde{t}< \tilde{t}'_k \\
 & * b(x) \approx x * e(x) \approx x+\tilde{t} \\
 & * \tilde{u}_\alpha < x \le \tilde{u}_{\alpha+1} *
 \tilde{u}_\beta<x+\tilde{t}\le \tilde{u}_{\beta+1},
\end{align*}
}

for some $\alpha,\beta,j,k$.
$[\inter{x}_{s'',B},\inter{x+\tilde{t}}_{s'',B})\in B'$ holds, since
$(s'',B')\models b(x)\approx x * e(x)\approx x+\tilde{t}$.
Define $s=s'[x\mapsto v']$ and $B= B' \setminus \{ [ [\![ x
]\!]_{s'',B'} , [\![ x + \tilde{t} ]\!]_{s'',B'} ) \}$.

We can say that $\inter{\tilde{t}}_{s'',B'} = \inter{t}_{s,B}$  
since the following derivation holds.  

Since $t$ is $b,e$-free, the following is true:
\[
\inter{\tilde{t}}_{s'',B'} = \inter{t \theta}_{s'',B'} = \inter{t \theta}_{s'',B}.
\]
Additionally, by using Lemma \ref{lma: Substitution for assignment} and the fact that $x'$ is fresh, the following holds:
\[
\inter{t \theta}_{s'',B} = \inter{t [x := x']}_{s'[x' \mapsto v'],B}
= \inter{t }_{s'[x' \mapsto v'][x \mapsto v'],B}
= \inter{t }_{s' [x \mapsto v'],B}
= \inter{t }_{s,B}.
\]
For any $u_i\in\termS^-(\psi)$, since $u_i$ contains neither $b$ nor $e$,
$\tilde{u_i}=u_i \theta$ and its semantics does not depend on $B'$.
Hence,
$\inter{\tilde{u_i}}_{s'',B'}=\inter{u_i}_{s,B}$ by $x'\not\in \textsf{fv}(u_i)$
similarly to the case of $\inter{\tilde{t}}_{s'',B'} = \inter{t}_{s,B}$.
Furthermore, we have $\inter{x}_{s'',B'}=s'(x)$ since $x$ and $x'$ are
distinct. Hence, $B=B'\setminus\{[s'(x),s'(x)+\inter{t}_{s,B})\}$ since
$t\in\termS^-(\psi)$.
Then, for any $u_i\in\termS^-(\psi)$, we have the following:
\begin{itemize}
 \item If $\alpha<i\le \beta$, since $(s'',B')$ satisfies $x \le
       \tilde{u_i} < x+\tilde{t}$, 
       the value
       $\inter{\tilde{u_i}}_{s'',B'}=
       \inter{u_i\theta}_{s'',B'}=\inter{u_i}_{s,B}$ 
       is in the block 
       $[s'(x),s'(x)+\inter{t}_{s,B})$, and hence
       $\inter{b(u_i)}_{s,B}=\inter{e(u_i)}_{s,B}=0$ 
       by $B=B'\setminus\{[s'(x),s'(x)+\inter{t}_{s,B})\}$.
 \item Otherwise, 
       $\inter{b(u_i\theta)}_{s'',B'}=\inter{b(u_i)}_{s,B}$ and
       $\inter{e(u_i\theta)}_{s'',B'}=\inter{e(u_i)}_{s,B}$.
\end{itemize}
Therefore, by the definition of $\textsf{rep}_\alpha^\beta$,
$\inter{\tilde{r}}_{s'',B'} = \inter{r}_{s,B}$ holds for any $r\in{\sf termS}(\psi)$.
Hence, for any pure formula $\pi$ in $\psi$, $(s'', B')\models \tilde{\pi}$ iff
$(s,B)\models \pi$.

Define $h$ as
{
\begin{align*}
h&=h'[l \mapsto \uparrow]_
 { l = [\![ x ]\!]_{s'', B'} }
 ^{ [\![ \tilde{{t}}_j ]\!]_{s'', B'} - 1 }
 \\
 &
 \Bigl(
 [l \mapsto \bot]_
 { l = [\![ \tilde{t}_i ]\!]_{s'', B'}   }
 ^{ [\![ \tilde{{t'}}_i ]\!]_{s'', B'} - 1   }
  [l \mapsto \ \uparrow ]_
 { l = [\![ \tilde{{t'}}_i ]\!]_{s'', B'}   }
 ^{ [\![ \tilde{t}_{i+1} ]\!]_{s'', B'}  - 1  }
 \Bigr)_{i=j}^{k-1} \\
 &   [l \mapsto \bot]_
 { l =  [\![ \tilde{t}_k ]\!]_{s'', B'}    }
 ^{ [\![  x + \tilde{t} ]\!]_{s'', B'} - 1   },
\end{align*}
}

which can be rewritten as
{
\begin{align*}
h&=h'[l \mapsto \uparrow]_
 { l = s'(x) }
 ^{ [\![ {t}_j ]\!]_{s, B} - 1 }
 \\
 &
 \Bigl(
 [l \mapsto \bot]_
 { l = [\![ {t}_i ]\!]_{s, B}   }
 ^{ [\![ {{t'}}_i ]\!]_{s, B} - 1   }
  [l \mapsto \ \uparrow ]_
 { l = [\![ {{t'}}_i ]\!]_{s, B}   }
 ^{ [\![ {t}_{i+1} ]\!]_{s, B}  - 1  }
 \Bigr)_{i=j}^{k-1} \\
 &   [l \mapsto \bot]_
 { l =  [\![ {t}_k ]\!]_{s, B}    }
 ^{ s'(x)+[\![ {t} ]\!]_{s, B} - 1   },
\end{align*}
}

Note that this is well-defined since $(s'',B')\models x<\tilde{t}'_j
* \tilde{t}_k <x +\tilde{t}$ and $(s'',B')\models \tilde{t}_i <
\tilde{t}_{i+1}$ for $j\le i<k$ by ${\psi}_{\text{pure}}\models
t_i<t_{i+1}$. 
Moreover, 
$(h,B)$ is exact since
$\textsf{dom}_+(h)=\textsf{dom}_+(h')\setminus[s'(x),s'(x)+\inter{t}_{s,B})$.

Then, we have
{
\begin{align*}
 (s'',h,B)\models &
  \tilde{\psi}'
 * \bigast_{i=1}^{j-1} \bararr(\tilde{t_i},\tilde{t}'_i)
 * \bigast_{i=j}^{k-1} \bararr(\tilde{t}_i,\tilde{{t'}}_i)
 \\
 &  * \bararr( \tilde{t}_k , x+\tilde{t} ) * \bararr(x+\tilde{t},\tilde{t}'_k)
 * \bigast_{i=k+1}^{m} \bararr(\tilde{t}_i,\tilde{t}'_i),
\end{align*}
}

and then, appending 
$\bararr(\tilde{{t'}}_{k},x+\tilde{t}) *
\bararr(x+\tilde{t},\tilde{t}'_k)$, we have
$ (s'',h,B) \models
  \tilde{\psi}'
 * \bigast_{i=1}^{m} \bararr(\tilde{t_i},\tilde{t}'_i)$,
and hence,
$(s,h,B) \models
  \psi'
 * \bigast_{i=1}^{m} \bararr(t_i,{t'}_i)$.
That means $(s,h,B)\models \psi$.

Since $s'=s[x\mapsto s'(x)]$, we have
$((s,h,B),(s',h',B'))\in\inter{x:={\tt alloc}(t)}_{\it ok}$.
Hence, we have $(s',h',B')\in\text{WPO}[\![\psi, x := \texttt{alloc}(t), \textit{ok}]\!]$.

\subsubsection{Case 3}

We demonstrate the case of

{
\begin{align*}
 (s'',h',B')\models &
  \tilde{\psi}'
 * \bigast_{i=1}^{j-1} \bararr(\tilde{t_i},\tilde{t}'_i)
 * \bararr(\tilde{t}_j,x)
 * \arr(x,x+\tilde{t}) \\
 & 
 * \bigast_{i=k+1}^{m} \bararr(\tilde{t}_i,\tilde{t}'_i)\\
 & * \tilde{t}_j<x<\tilde{t}'_j 
 * \tilde{t'}_k \leq x+\tilde{t} \leq  \tilde{t}_{k+1} 
 \\
 & * b(x) \approx x * e(x) \approx x+\tilde{t} \\
 & * \tilde{u}_\alpha < x \le \tilde{u}_{\alpha+1} *
 \tilde{u}_\beta<x+\tilde{t}\le \tilde{u}_{\beta+1},
\end{align*}
}

for some $\alpha,\beta,j,k$.
$[\inter{x}_{s'',B},\inter{x+\tilde{t}}_{s'',B})\in B'$ holds, since
$(s'',B')\models b(x)\approx x * e(x)\approx x+\tilde{t}$.
Define $s=s'[x\mapsto v']$ and $B= B' \setminus \{ [ [\![ x
]\!]_{s'',B'} , [\![ x + \tilde{t} ]\!]_{s'',B'} ) \}$.

We can say that $\inter{\tilde{t}}_{s'',B'} = \inter{t}_{s,B}$  
since the following derivation holds.  

Since $t$ is $b,e$-free, the following is true:
\[
\inter{\tilde{t}}_{s'',B'} = \inter{t \theta}_{s'',B'} = \inter{t \theta}_{s'',B}.
\]
Additionally, by using Lemma \ref{lma: Substitution for assignment} and the fact that $x'$ is fresh, the following holds:
\[
\inter{t \theta}_{s'',B} = \inter{t [x := x']}_{s'[x' \mapsto v'],B}
= \inter{t }_{s'[x' \mapsto v'][x \mapsto v'],B}
= \inter{t }_{s' [x \mapsto v'],B}
= \inter{t }_{s,B}.
\]
For any $u_i\in\termS^-(\psi)$, since $u_i$ contains neither $b$ nor $e$,
$\tilde{u_i}=u_i \theta$ and its semantics does not depend on $B'$.
Hence,
$\inter{\tilde{u_i}}_{s'',B'}=\inter{u_i}_{s,B}$ by $x'\not\in \textsf{fv}(u_i)$
similarly to the case of $\inter{\tilde{t}}_{s'',B'} = \inter{t}_{s,B}$.
Furthermore, we have $\inter{x}_{s'',B'}=s'(x)$ since $x$ and $x'$ are
distinct. Hence, $B=B'\setminus\{[s'(x),s'(x)+\inter{t}_{s,B})\}$ since
$t\in\termS^-(\psi)$.
Then, for any $u_i\in\termS^-(\psi)$, we have the following:
\begin{itemize}
 \item If $\alpha<i\le \beta$, since $(s'',B')$ satisfies $x \le
       \tilde{u_i} < x+\tilde{t}$, 
       the value
       $\inter{\tilde{u_i}}_{s'',B'}=
       \inter{u_i\theta}_{s'',B'}=\inter{u_i}_{s,B}$ 
       is in the block
       $[s'(x),s'(x)+\inter{t}_{s,B})$, and hence
       $\inter{b(u_i)}_{s,B}=\inter{e(u_i)}_{s,B}=0$ 
       by $B=B'\setminus\{[s'(x),s'(x)+\inter{t}_{s,B})\}$.
 \item Otherwise, 
       $\inter{b(u_i\theta)}_{s'',B'}=\inter{b(u_i)}_{s,B}$ and
       $\inter{e(u_i\theta)}_{s'',B'}=\inter{e(u_i)}_{s,B}$.
\end{itemize}
Therefore, by the definition of $\textsf{rep}_\alpha^\beta$,
$\inter{\tilde{r}}_{s'',B'} = \inter{r}_{s,B}$ holds for any $r\in{\sf termS}(\psi)$.
Hence, for any pure formula $\pi$ in $\psi$, $(s'', B')\models \tilde{\pi}$ iff
$(s,B)\models \pi$.

Define $h$ as
{
\begin{align*}
h&=h'[l \mapsto \bot]_
 { l = [\![ x ]\!]_{s'', B'} }
 ^{ [\![ \tilde{{t'}}_j ]\!]_{s'', B'} - 1 }
 [l \mapsto \ \uparrow ]_
 { l = [\![ \tilde{{t'}}_j ]\!]_{s'', B'}  }
 ^{ [\![ \tilde{t}_{j+1} ]\!]_{s'', B'} - 1 }
 \\
 &
 \Bigl(
 [l \mapsto \bot]_
 { l = [\![ \tilde{t}_i ]\!]_{s'', B'}   }
 ^{ [\![ \tilde{{t'}}_i ]\!]_{s'', B'} - 1   }
  [l \mapsto \ \uparrow ]_
 { l = [\![ \tilde{{t'}}_i ]\!]_{s'', B'}   }
 ^{ [\![ \tilde{t}_{i+1} ]\!]_{s'', B'}  - 1  }
 \Bigr)_{i=j+1}^{k-1} \\
 &
 [l \mapsto \bot]_
 { l = [\![ \tilde{t}_k ]\!]_{s'', B'}   }
 ^{ [\![ \tilde{{t'}}_k ]\!]_{s'', B'} - 1   }
 [l \mapsto \uparrow]_
 { l =  [\![ \tilde{t'}_k ]\!]_{s'', B'}    }
 ^{ [\![  x + \tilde{t} ]\!]_{s'', B'} - 1   },
\end{align*}
}
which can be rewritten as
{
\begin{align*}
h&=h'[l \mapsto \bot]_
 { l = s'(x) }
 ^{ [\![ {{t'}}_j ]\!]_{s, B} - 1 }
 [l \mapsto \ \uparrow ]_
 { l = [\![ {{t'}}_j ]\!]_{s, B}  }
 ^{ [\![ {t}_{j+1} ]\!]_{s, B} - 1 }
 \\
 &
 \Bigl(
 [l \mapsto \bot]_
 { l = [\![ {t}_i ]\!]_{s, B}   }
 ^{ [\![ {{t'}}_i ]\!]_{s, B} - 1   }
  [l \mapsto \ \uparrow ]_
 { l = [\![ {{t'}}_i ]\!]_{s, B}   }
 ^{ [\![ {t}_{i+1} ]\!]_{s, B}  - 1  }
 \Bigr)_{i=j+1}^{k-1} \\
 &
 [l \mapsto \bot]_
 { l = [\![ {t}_k ]\!]_{s, B}   }
 ^{ [\![ {{t'}}_k ]\!]_{s, B} - 1   }
 [l \mapsto \uparrow]_
 { l =  [\![ {t'}_k ]\!]_{s, B}    }
 ^{ s'(x)+[\![ {t} ]\!]_{s, B} - 1   },
\end{align*}
}
Note that this is well-defined since $(s'',B')\models x<\tilde{t}'_j
* \tilde{t}_k <x +\tilde{t}$ and $(s'',B')\models \tilde{t}_i <
\tilde{t}_{i+1}$ for $j\le i<k$ by ${\psi}_{\text{pure}}\models
t_i<t_{i+1}$. 
Moreover, 
$(h,B)$ is exact since
$\textsf{dom}_+(h)=\textsf{dom}_+(h')\setminus[s'(x),s'(x)+\inter{t}_{s,B})$.

Then, we have
{
\begin{align*}
 (s'',h,B)\models &
  \tilde{\psi}'
 * \bigast_{i=1}^{j-1} \bararr(\tilde{t_i},\tilde{t}'_i)
 * \bararr(\tilde{t}_j,x)
 * \bararr (x , \tilde{t'}_j ) \\
 & 
 * \bigast_{i=j+1}^{m} \bararr(\tilde{t}_i,\tilde{t}'_i)\\
 & * \tilde{t}_j<x<\tilde{t}'_j 
 * \tilde{t'}_k \leq x+\tilde{t} \leq  \tilde{t}_{k+1} 
 \\
 & * b(x) \approx x * e(x) \approx x+\tilde{t} \\
 & * \tilde{u}_\alpha < x \le \tilde{u}_{\alpha+1} *
 \tilde{u}_\beta<x+\tilde{t}\le \tilde{u}_{\beta+1},
\end{align*}
}

and then, appending $\bararr(\tilde{t}_j,x) * \bararr(x,\tilde{{t'}}_{j})$, 
we have
$ (s'',h,B) \models
  \tilde{\psi}'
 * \bigast_{i=1}^{m} \bararr(\tilde{t_i},\tilde{t}'_i)$,
and hence,
$(s,h,B) \models
  \psi'
 * \bigast_{i=1}^{m} \bararr(t_i,{t'}_i)$.
That means $(s,h,B)\models \psi$.

Since $s'=s[x\mapsto s'(x)]$, we have
$((s,h,B),(s',h',B'))\in\inter{x:={\tt alloc}(t)}_{\it ok}$.
Hence, we have $(s',h',B')\in\text{WPO}[\![\psi, x := \texttt{alloc}(t), \textit{ok}]\!]$.

\subsubsection{Case 4}

We demonstrate the case of

{
\begin{align*}
 (s'',h',B')\models &
  \tilde{\psi}'
 * \bigast_{i=1}^{j-1} \bararr(\tilde{t_i},\tilde{t}'_i)
 * \arr(x,x+\tilde{t}) \\
 & 
 * \bigast_{i=k+1}^{m} \bararr(\tilde{t}_i,\tilde{t}'_i)\\
 & * \tilde{t'}_{j-1} \leq  x \leq \tilde{t}_j  
 * \tilde{t'}_k \leq x+\tilde{t} \leq  \tilde{t}_{k+1} 
 \\
 & * b(x) \approx x * e(x) \approx x+\tilde{t} \\
 & * \tilde{u}_\alpha < x \le \tilde{u}_{\alpha+1} *
 \tilde{u}_\beta<x+\tilde{t}\le \tilde{u}_{\beta+1},
\end{align*}
}

for some $\alpha,\beta,j,k$.
$[\inter{x}_{s'',B},\inter{x+\tilde{t}}_{s'',B})\in B'$ holds, since
$(s'',B')\models b(x)\approx x * e(x)\approx x+\tilde{t}$.
Define $s=s'[x\mapsto v']$ and $B= B' \setminus \{ [ [\![ x
]\!]_{s'',B'} , [\![ x + \tilde{t} ]\!]_{s'',B'} ) \}$.

We can say that $\inter{\tilde{t}}_{s'',B'} = \inter{t}_{s,B}$  
since the following derivation holds.  

Since $t$ is $b,e$-free, the following is true:
\[
\inter{\tilde{t}}_{s'',B'} = \inter{t \theta}_{s'',B'} = \inter{t \theta}_{s'',B}.
\]
Additionally, by using Lemma \ref{lma: Substitution for assignment} and the fact that $x'$ is fresh, the following holds:
\[
\inter{t \theta}_{s'',B} = \inter{t [x := x']}_{s'[x' \mapsto v'],B}
= \inter{t }_{s'[x' \mapsto v'][x \mapsto v'],B}
= \inter{t }_{s' [x \mapsto v'],B}
= \inter{t }_{s,B}.
\]
For any $u_i\in\termS^-(\psi)$, since $u_i$ contains neither $b$ nor $e$,
$\tilde{u_i}=u_i \theta$ and its semantics does not depend on $B'$.
Hence,
$\inter{\tilde{u_i}}_{s'',B'}=\inter{u_i}_{s,B}$ by $x'\not\in \textsf{fv}(u_i)$
similarly to the case of $\inter{\tilde{t}}_{s'',B'} = \inter{t}_{s,B}$.
Furthermore, we have $\inter{x}_{s'',B'}=s'(x)$ since $x$ and $x'$ are
distinct. Hence, $B=B'\setminus\{[s'(x),s'(x)+\inter{t}_{s,B})\}$ since
$t\in\termS^-(\psi)$.
Then, for any $u_i\in\termS^-(\psi)$, we have the following:
\begin{itemize}
 \item If $\alpha<i\le \beta$, since $(s'',B')$ satisfies $x \le
       \tilde{u_i} < x+\tilde{t}$, 
       the value
       $\inter{\tilde{u_i}}_{s'',B'}=
       \inter{u_i\theta}_{s'',B'}=\inter{u_i}_{s,B}$ 
       is in the block
       $[s'(x),s'(x)+\inter{t}_{s,B})$, and hence
       $\inter{b(u_i)}_{s,B}=\inter{e(u_i)}_{s,B}=0$ 
       by $B=B'\setminus\{[s'(x),s'(x)+\inter{t}_{s,B})\}$.
 \item Otherwise, 
       $\inter{b(u_i\theta)}_{s'',B'}=\inter{b(u_i)}_{s,B}$ and
       $\inter{e(u_i\theta)}_{s'',B'}=\inter{e(u_i)}_{s,B}$.
\end{itemize}
Therefore, by the definition of $\textsf{rep}_\alpha^\beta$,
$\inter{\tilde{r}}_{s'',B'} = \inter{r}_{s,B}$ holds for any $r\in{\sf termS}(\psi)$.
Hence, for any pure formula $\pi$ in $\psi$, $(s'', B')\models \tilde{\pi}$ iff
$(s,B)\models \pi$.

Define $h$ as
{
\begin{align*}
h&=
h'[l \mapsto \uparrow]_
 { l = [\![ x ]\!]_{s'', B'} }
 ^{ [\![ \tilde{{t}}_j ]\!]_{s'', B'} - 1 }
 \\
 &
 \Bigl(
 [l \mapsto \bot]_
 { l = [\![ \tilde{t}_i ]\!]_{s'', B'}   }
 ^{ [\![ \tilde{{t'}}_i ]\!]_{s'', B'} - 1   }
  [l \mapsto \ \uparrow ]_
 { l = [\![ \tilde{{t'}}_i ]\!]_{s'', B'}   }
 ^{ [\![ \tilde{t}_{i+1} ]\!]_{s'', B'}  - 1  }
 \Bigr)_{i=j}^{k-1} \\
 &
 [l \mapsto \bot]_
 { l = [\![ \tilde{t}_k ]\!]_{s'', B'}   }
 ^{ [\![ \tilde{{t'}}_k ]\!]_{s'', B'} - 1   }
 [l \mapsto \uparrow]_
 { l =  [\![ \tilde{t'}_k ]\!]_{s'', B'}    }
 ^{ [\![  x + \tilde{t} ]\!]_{s'', B'} - 1   },
\end{align*}
}
which can be rewritten as

{
\begin{align*}
h&=
h'[l \mapsto \uparrow]_
 { l = s'(x) }
 ^{ [\![ {{t}}_j ]\!]_{s, B} - 1 }
 \\
 &
 \Bigl(
 [l \mapsto \bot]_
 { l = [\![ {t}_i ]\!]_{s, B}   }
 ^{ [\![ {{t'}}_i ]\!]_{s, B} - 1   }
  [l \mapsto \ \uparrow ]_
 { l = [\![ {{t'}}_i ]\!]_{s, B}   }
 ^{ [\![ {t}_{i+1} ]\!]_{s, B}  - 1  }
 \Bigr)_{i=j}^{k-1} \\
 &
 [l \mapsto \bot]_
 { l = [\![ {t}_k ]\!]_{s, B}   }
 ^{ [\![ {{t'}}_k ]\!]_{s, B} - 1   }
 [l \mapsto \uparrow]_
 { l =  [\![ {t'}_k ]\!]_{s, B}    }
 ^{ s'(x)+[\![ {t} ]\!]_{s, B} - 1   },
\end{align*}
}
Note that this is well-defined since $(s'',B')\models x<\tilde{t}'_j
* \tilde{t}_k <x +\tilde{t}$ and $(s'',B')\models \tilde{t}_i <
\tilde{t}_{i+1}$ for $j\le i<k$ by ${\psi}_{\text{pure}}\models
t_i<t_{i+1}$. 
Moreover, 
$(h,B)$ is exact since
$\textsf{dom}_+(h)=\textsf{dom}_+(h')\setminus[s'(x),s'(x)+\inter{t}_{s,B})$.

Then, we have
{
\begin{align*}
 (s'',h,B)\models &
  \tilde{\psi}'
 * \bigast_{i=1}^{j-1} \bararr(\tilde{t_i},\tilde{t}'_i)
 * \bigast_{i=j}^{k} \bararr(\tilde{t}_i,\tilde{{t'}}_i)
 * \bigast_{i=k+1}^{m} \bararr(\tilde{t}_i,\tilde{t}'_i),
\end{align*}
}

and then, 
we have
$ (s'',h,B) \models
  \tilde{\psi}'
 * \bigast_{i=1}^{m} \bararr(\tilde{t_i},\tilde{t}'_i)$,
and hence,
$(s,h,B) \models
  \psi'
 * \bigast_{i=1}^{m} \bararr(t_i,{t'}_i)$.
That means $(s,h,B)\models \psi$.

Since $s'=s[x\mapsto s'(x)]$, we have
$((s,h,B),(s',h',B'))\in\inter{x:={\tt alloc}(t)}_{\it ok}$.
Hence, we have $(s',h',B')\in\text{WPO}[\![\psi, x := \texttt{alloc}(t), \textit{ok}]\!]$.

\subsubsection{Case 5}

We demonstrate the case of

{
\begin{align*}
 (s'',h',B')\models &
  \tilde{\psi} * \arr(x,x+\tilde{t}) \\
 & * b(x) \approx x * e(x) \approx x+\tilde{t} \\
 & * \tilde{u}_\alpha < x \le \tilde{u}_{\alpha+1} *
 \tilde{u}_\beta<x+\tilde{t}\le \tilde{u}_{\beta+1},
\end{align*}
}

for some $\alpha,\beta$.
$[\inter{x}_{s'',B},\inter{x+\tilde{t}}_{s'',B})\in B'$ holds, since
$(s'',B')\models b(x)\approx x * e(x)\approx x+\tilde{t}$.
Define $s=s'[x\mapsto v']$ and $B= B' \setminus \{ [ [\![ x
]\!]_{s'',B'} , [\![ x + \tilde{t} ]\!]_{s'',B'} ) \}$.

We can say that $\inter{\tilde{t}}_{s'',B'} = \inter{t}_{s,B}$  
since the following derivation holds.  

Since $t$ is $b,e$-free, the following is true:
\[
\inter{\tilde{t}}_{s'',B'} = \inter{t \theta}_{s'',B'} = \inter{t \theta}_{s'',B}.
\]
Additionally, by using Lemma \ref{lma: Substitution for assignment} and the fact that $x'$ is fresh, the following holds:
\[
\inter{t \theta}_{s'',B} = \inter{t [x := x']}_{s'[x' \mapsto v'],B}
= \inter{t }_{s'[x' \mapsto v'][x \mapsto v'],B}
= \inter{t }_{s' [x \mapsto v'],B}
= \inter{t }_{s,B}.
\]
For any $u_i\in\termS^-(\psi)$, since $u_i$ contains neither $b$ nor $e$,
$\tilde{u_i}=u_i \theta$ and its semantics does not depend on $B'$.
Hence,
$\inter{\tilde{u_i}}_{s'',B'}=\inter{u_i}_{s,B}$ by $x'\not\in \textsf{fv}(u_i)$
similarly to the case of $\inter{\tilde{t}}_{s'',B'} = \inter{t}_{s,B}$.
Furthermore, we have $\inter{x}_{s'',B'}=s'(x)$ since $x$ and $x'$ are
distinct. Hence, $B=B'\setminus\{[s'(x),s'(x)+\inter{t}_{s,B})\}$ since
$t\in\termS^-(\psi)$.
Then, for any $u_i\in\termS^-(\psi)$, we have the following:
\begin{itemize}
 \item If $\alpha<i\le \beta$, since $(s'',B')$ satisfies $x \le
       \tilde{u_i} < x+\tilde{t}$, 
       the value
       $\inter{\tilde{u_i}}_{s'',B'}=
       \inter{u_i\theta}_{s'',B'}=\inter{u_i}_{s,B}$ 
       is in the block
       $[s'(x),s'(x)+\inter{t}_{s,B})$, and hence
       $\inter{b(u_i)}_{s,B}=\inter{e(u_i)}_{s,B}=0$ 
       by $B=B'\setminus\{[s'(x),s'(x)+\inter{t}_{s,B})\}$.
 \item Otherwise, 
       $\inter{b(u_i\theta)}_{s'',B'}=\inter{b(u_i)}_{s,B}$ and
       $\inter{e(u_i\theta)}_{s'',B'}=\inter{e(u_i)}_{s,B}$.
\end{itemize}
Therefore, by the definition of $\textsf{rep}_\alpha^\beta$,
$\inter{\tilde{r}}_{s'',B'} = \inter{r}_{s,B}$ holds for any $r\in{\sf termS}(\psi)$.
Hence, for any pure formula $\pi$ in $\psi$, $(s'', B')\models \tilde{\pi}$ iff
$(s,B)\models \pi$.

Define $h$ as

{
\begin{align*}
h&=
h'[l \mapsto \uparrow]_
 { l = [\![ x ]\!]_{s'', B'} }
^{ [\![  x + \tilde{t} ]\!]_{s'', B'} - 1   }
\end{align*}
}
which can be rewritten as

{
\begin{align*}
h&=
h'[l \mapsto \uparrow]_
 { l = s'(x) }
^{ s'(x)+[\![ {t} ]\!]_{s, B} - 1   }
\end{align*}
}
Moreover, 
$(h,B)$ is exact since
$\textsf{dom}_+(h)=\textsf{dom}_+(h')\setminus[s'(x),s'(x)+\inter{t}_{s,B})$.

Then, we have
{
\begin{align*}
 (s'',h,B)\models 
  \tilde{\psi}
\end{align*}
}

and then,  we have $(s,h,B) \models  \psi$.

Since $s'=s[x\mapsto s'(x)]$, we have
$((s,h,B),(s',h',B'))\in\inter{x:={\tt alloc}(t)}_{\it ok}$.
Hence, we have $(s',h',B')\in\text{WPO}[\![\psi, x := \texttt{alloc}(t), \textit{ok}]\!]$.

\subsection{$(s', h', B') \in \text{WPO} [\![\psi, x := \texttt{alloc($t$)},
\textit{ok}]\!] \Rightarrow
(s', h', B') \models \textsf{wpo}(\psi, x := \texttt{alloc($t$)}, \textit{ok})$}

Next, we prove the converse:  
$(s', h', B') \in \text{WPO} [\![\psi, x := \texttt{alloc($t$)},
\textit{ok}]\!] \Rightarrow
(s', h', B') \models \textsf{wpo}(\psi, x := \texttt{alloc($t$)}, \textit{ok})$.
We frequently use the notation $h[l \mapsto -]$ to denote that $l$ maps
to some value.

Assume $(s', h', B') \in \text{WPO} [\![\psi, x := \texttt{alloc($t$)},
\textit{ok}]\!]$, and then there exists $(s,h,B)$ and $l\in\textsc{Loc}$
such that $(s,h,B)\models \psi$, $(\bigcup B) \cap [ l
, l + [\![ t ]\!]_{s,B} ) = \emptyset$, and $(s', h', B') = (s [x
\mapsto l] , h[l + i \mapsto - ]_{i=0}^{[\![ t ]\!]_{s,B} - 1 } , B \cup
\{ [l , l + [\![ t ]\!]_{s,B} ) \} )$.  Then, $l=s'(x)$ holds. Since
$x'$ is fresh, these do not depend on the value of $s(x')$, so we can
assume $s(x')=s'(x)$. Then, $s'=s[x'\mapsto s(x)]$ holds.  

Since $(s,h,B) \models \psi$ holds, $h$ is of the form  
$$h = h_1 \circ \cdots \circ h_m \circ h'',$$
where $(s,h_i, B) \models \bararr(t_i, {t'}_i)$ for $1 \leq i \leq m$ and $(s,h'',B) \models \psi'$.
That is, $ \textsf{dom}(h_i) = \textsf{dom}_{-}(h_i) = 
[ [\![t_i]\!]_{s,B}, \inter{{t'}_i}_{s,B})$ for each $i$.

Moreover, we have $(s,B)\models t_i<{t'}_i\le t_{i+1}$ for each $i$
and $(s,B)\models{\tt null} R_1 u_1 \cdots R_N u_N$ for
$u_i\in\termS^-(\psi)$.  Without
loss of generality, we assume
$\inter{u_\alpha}_{s,B}<s'(x)\le\inter{u_{\alpha+1}}_{s,B}$ and
$\inter{u_\beta}_{s,B}<s'(x)+\inter{t}_{s,B}\le
\inter{u_{\beta+1}}_{s,B}$ for some $\alpha,\beta$.  If $\beta=N$,
ignore $\inter{u_{\beta+1}}_{s,B}$.
Then, we can show $\inter{r}_{s,B}=\inter{\tilde{r}}_{s',B'}$ as in the
previous case where $\tilde{r}=\textsf{rep}_\alpha^\beta(r\theta)$.
Hence, by $(s,h,B)\models\psi$, we have $(s',B')$ satisfies
$\tilde{\psi}_{\text{pure}}$,
$b(x) \approx x * e(x) \approx x+\tilde{t}$,
$\tilde{u}_\alpha < x \le \tilde{u}_{\alpha+1}$, and
$\tilde{u}_\beta<x+\tilde{t}\le \tilde{u}_{\beta+1}$.

In the following, we divide the cases depending on where the value of $s'(x)$
and $s'(x)+[\![t]\!]_{s,B}$ are in the sequence of $\inter{t_i}_{s,B}$
and $\inter{{t'}_i}_{s,B}$ for $1\le i\le m$.

\subsubsection{Case 1}
If $ [\![t_j]\!]_{s,B}<s'(x)<[\![{t'}_j]\!]_{s,B}$ and
$ [\![t_k]\!]_{s,B}<s'(x)+[\![t]\!]_{s,B}<[\![{t'}_k]\!]_{s,B}$ holds for
some $1\le j\le k\le m$.  In this case, in the heap $h'=h[s'(x) + i
\mapsto - ]_{i=0}^{[\![ t ]\!]_{s,B} - 1 }$, $h_j$ is divided into two
parts $h'^1_j$ and $h'^2_j$ such that $h'^1_j$ is $h_j$ whose domain is
shorten to $[[\![t_j]\!]_{s,B},s'(x))$ and
$\textsf{dom}(h'^2_j)=\textsf{dom}_{+}(h'^2_j)=[s'(x),[\![{t'}_j]\!]_{s,B})$.
Similarly, $h_k$ is divided into $h'^1_k$ and $h'^2_k$, where
$\textsf{dom}(h'^1_k)=\textsf{dom}_{+}(h'^1_k) = 
[[\![t_k]\!]_{s,B},s'(x)+[\![t]\!]_{s,B})$
and
$h'^2_k$ is $h_k$ whose domain is shorten to
$[s'(x)+[\![t]\!]_{s,B},[\![{t'}_k]\!]_{s,B})$.
Then, $h'$ can be written
as $h'= h_1\circ\cdots \circ h_{j-1} \circ h'^1_{j}\circ h''' \circ h'^2_k\circ
h_{k+1}\cdots \circ  h_{m}$, 
where 
$\textsf{dom}(h''')=\textsf{dom}_{+}(h''')=[s'(x),s'(x)+[\![t]\!]_{s,B})$.

From the above and $\inter{r}_{s,B}=\inter{\tilde{r}}_{s',B'}$, we have
{
\begin{align*}
 (s',h',B')\models &
  \tilde{\psi}'
 * \bigast_{i=1}^{j-1} \bararr(\tilde{t_i},\tilde{t}'_i)
 * \bararr(\tilde{t}_j,x)
 * \arr(x,x+\tilde{t}) \\
 & * \bararr(x+\tilde{t},\tilde{t}'_k)
 * \bigast_{i=k+1}^{m} \bararr(\tilde{t}_i,\tilde{t}'_i)\\
 & * \tilde{t}_j<x<\tilde{t}'_j * \tilde{t}_k<x+\tilde{t}< \tilde{t}'_k \\
 & * b(x) \approx x * e(x) \approx x+\tilde{t} \\
 & * \tilde{u}_\alpha < x \le \tilde{u}_{\alpha+1} *
 \tilde{u}_\beta<x+\tilde{t}\le \tilde{u}_{\beta+1},
\end{align*}
}
which is one of the disjuncts in  $\chi$ such that ${\sf
wpo}(\psi,x:={\tt alloc}(t),{\it ok})=\exists x'.\chi$.

\subsubsection{Case 2}

If $ [\![{t'}_{j-1}]\!]_{s,B} \leq s'(x) \leq [\![{t}_j]\!]_{s,B}$ and
$ [\![t_k]\!]_{s,B} < s'(x)+[\![t]\!]_{s,B} < [\![{t'}_k]\!]_{s,B}$ holds for
some $1\le j\le k\le m$.  
In this case, in the heap $h'=h[s'(x) + i
\mapsto - ]_{i=0}^{[\![ t ]\!]_{s,B} - 1 }$, 
$h_k$ is divided into $h'^1_k$ and $h'^2_k$, where
$\textsf{dom}(h'^1_k)=\textsf{dom}_{+}(h'^1_k) = 
[[\![t_k]\!]_{s,B},s'(x)+[\![t]\!]_{s,B})$
and
$h'^2_k$ is $h_k$ whose domain is shorten to
$[s'(x)+[\![t]\!]_{s,B},[\![{t'}_k]\!]_{s,B})$.
Then, $h'$ can be written
as $h'= h_1\circ\cdots \circ h_{j-1} \circ   h''' \circ h'^2_k\circ
h_{k+1}\cdots \circ  h_{m}$, 
where 
$\textsf{dom}(h''')=\textsf{dom}_{+}(h''')=[s'(x),s'(x)+[\![t]\!]_{s,B})$.

From the above and $\inter{r}_{s,B}=\inter{\tilde{r}}_{s',B'}$, we have
{
\begin{align*}
 (s',h',B')\models &
  \tilde{\psi}'
 * \bigast_{i=1}^{j-1} \bararr(\tilde{t_i},\tilde{t}'_i)
 * \arr(x,x+\tilde{t}) \\
 & * \bararr(x+\tilde{t},\tilde{t}'_k)
 * \bigast_{i=k+1}^{m} \bararr(\tilde{t}_i,\tilde{t}'_i)\\
 & * \tilde{t'}_{j-1} \leq  x \leq \tilde{t}_j * \tilde{t}_k<x+\tilde{t}< \tilde{t}'_k \\
 & * b(x) \approx x * e(x) \approx x+\tilde{t} \\
 & * \tilde{u}_\alpha < x \le \tilde{u}_{\alpha+1} *
 \tilde{u}_\beta<x+\tilde{t}\le \tilde{u}_{\beta+1},
\end{align*}
}
which is one of the disjuncts in  $\chi$ such that ${\sf
wpo}(\psi,x:={\tt alloc}(t),{\it ok})=\exists x'.\chi$.


\subsubsection{Case 3}

If $ [\![t_j]\!]_{s,B}<s'(x)<[\![{t'}_j]\!]_{s,B}$ and
$ [\![{t'}_k]\!]_{s,B} \leq s'(x)+[\![t]\!]_{s,B} \leq [\![{t'}_{k+1}]\!]_{s,B}$ holds for
some $1\le j\le k\le m$.  
In this case, in the heap $h'=h[s'(x) + i
\mapsto - ]_{i=0}^{[\![ t ]\!]_{s,B} - 1 }$, $h_j$ is divided into two
parts $h'^1_j$ and $h'^2_j$ such that $h'^1_j$ is $h_j$ whose domain is
shorten to $[[\![t_j]\!]_{s,B},s'(x))$ and
$\textsf{dom}(h'^2_j)=\textsf{dom}_{+}(h'^2_j)=[s'(x),[\![{t'}_j]\!]_{s,B})$.
Then, $h'$ can be written
as $h'= h_1\circ\cdots \circ h_{j-1} \circ h'^1_{j}\circ h''' \circ h'_{k+1} \circ
\cdots \circ  h_{m}$, 
where 
$\textsf{dom}(h''')=\textsf{dom}_{+}(h''')=[s'(x),s'(x)+[\![t]\!]_{s,B})$.

From the above and $\inter{r}_{s,B}=\inter{\tilde{r}}_{s',B'}$, we have
{
\begin{align*}
 (s',h',B')\models &
  \tilde{\psi}'
 * \bigast_{i=1}^{j-1} \bararr(\tilde{t_i},\tilde{t}'_i)
 * \bararr(\tilde{t}_j,x)
 * \arr(x,x+\tilde{t}) \\
 & 
 * \bigast_{i=k+1}^{m} \bararr(\tilde{t}_i,\tilde{t}'_i)\\
 & * \tilde{t}_j<x<\tilde{t}'_j * \tilde{t}_k<x+\tilde{t}< \tilde{t}'_k \\
 & * b(x) \approx x * e(x) \approx x+\tilde{t} \\
 & * \tilde{u}_\alpha < x \le \tilde{u}_{\alpha+1} *
 \tilde{u}_\beta<x+\tilde{t}\le \tilde{u}_{\beta+1},
\end{align*}
}
which is one of the disjuncts in  $\chi$ such that ${\sf
wpo}(\psi,x:={\tt alloc}(t),{\it ok})=\exists x'.\chi$.


\subsubsection{Case 4}
If $ [\![{t'}_{j-1}]\!]_{s,B} \leq s'(x) \leq [\![{t}_j]\!]_{s,B}$ and
$ [\![{t'}_k]\!]_{s,B} \leq s'(x)+[\![t]\!]_{s,B} \leq [\![{t}_{k+1}]\!]_{s,B}$ holds for
some $1\le j\le k\le m$.  
Then, $h'$ can be written
as $h'= h_1\circ\cdots \circ h_{j-1} \circ   h''' \circ  
h_{k+1}\cdots \circ  h_{m}$, 
where 
$\textsf{dom}(h''')=\textsf{dom}_{+}(h''')=[s'(x),s'(x)+[\![t]\!]_{s,B})$.

From the above and $\inter{r}_{s,B}=\inter{\tilde{r}}_{s',B'}$, we have
{
\begin{align*}
 (s',h',B')\models &
  \tilde{\psi}'
 * \bigast_{i=1}^{j-1} \bararr(\tilde{t_i},\tilde{t}'_i)
 * \arr(x,x+\tilde{t}) 
 * \bigast_{i=k+1}^{m} \bararr(\tilde{t}_i,\tilde{t}'_i)\\
 & * \tilde{t}_j<x<\tilde{t}'_j * \tilde{t}_k<x+\tilde{t}< \tilde{t}'_k \\
 & * b(x) \approx x * e(x) \approx x+\tilde{t} \\
 & * \tilde{u}_\alpha < x \le \tilde{u}_{\alpha+1} *
 \tilde{u}_\beta<x+\tilde{t}\le \tilde{u}_{\beta+1},
\end{align*}
}
which is one of the disjuncts in  $\chi$ such that ${\sf
wpo}(\psi,x:={\tt alloc}(t),{\it ok})=\exists x'.\chi$.

\subsubsection{Case 5}

{
Otherwise, one of the following conditions holds:
\begin{itemize}
    \item $[\![{t'}_{i}]\!]_{s,B} \leq s'(x) < s'(x) + [\![{t}]\!]_{s,B} \leq [\![{t}_{i+1}]\!]_{s,B}$ holds 
    for some $i$.
    \item $ s'(x) + [\![{t}]\!]_{s,B} < [\![{t}_{1}]\!]_{s,B}$ or
    $s'(x) \geq [\![{t'}_m]\!]_{s,B} $ hold.
\end{itemize}
In either case, $h'$ can be written as  
$ h' = h_1 \circ \cdots \circ h_{m} \circ h'''$ ,
where  
$ \textsf{dom}(h''') = \textsf{dom}_{+}(h''') = [s'(x), s'(x) + [\![t]\!]_{s,B})$.
}

From the above and $\inter{r}_{s,B}=\inter{\tilde{r}}_{s',B'}$, we have
{
\begin{align*}
 (s',h',B')\models &
  \tilde{\psi} * \arr(x,x+\tilde{t}) \\
 & * b(x) \approx x * e(x) \approx x+\tilde{t} \\
 & * \tilde{u}_\alpha < x \le \tilde{u}_{\alpha+1} *
 \tilde{u}_\beta<x+\tilde{t}\le \tilde{u}_{\beta+1},
\end{align*}
}
which is one of the disjuncts in  $\chi$ such that ${\sf
wpo}(\psi,x:={\tt alloc}(t),{\it ok})=\exists x'.\chi$.

%

%
\section{Proof of Proposition \ref{prop: expressiveness of WPO calculus} for $\normalfont\texttt{free($t$)}$}
\label{sec: expressiveness proof of atomic command in appendix (free)}

We prove
$(s', h', B') \models \textsf{wpo}(  \psi , \texttt{free($t$)}  , \textit{ok} ) \iff 
(s', h', B') \in \text{WPO} [\![  \psi , \texttt{free($t$)}  , \textit{ok} ]\!]  $ holds.
By definition, $(s', h', B') \in \text{WPO} [\![  \psi , \texttt{free($t$)}  , \textit{ok} ]\!]  $
is equivalent to
$
\exists (s,h,B) \models \psi .
[\![ t ]\!]_{s,B} = b_B ([\![ t ]\!]_{s,B}) \land b_B ([\![ t ]\!]_{s,B}) > 0 \land
(s', h', B')  = 
(s, h [i \mapsto \bot]_{i= b_B ([\![ t ]\!]_{s,B})}^{e_B ([\![ t ]\!]_{s,B}) - 1} , 
B \setminus \{  [ b_B ([\![ t ]\!]_{s,B})  ,  e_B ([\![ t ]\!]_{s,B})   ) \}  )  $.

\subsection{$(s', h', B') \models \textsf{wpo}( \psi , \texttt{free($t$)} ,
\textit{ok} ) \Rightarrow (s', h', B') \in \text{WPO} [\![ \psi ,
\texttt{free($t$)} , \textit{ok} ]\!]  $}
Firstly, we prove $(s', h', B') \models \textsf{wpo}( \psi , \texttt{free($t$)} ,
\textit{ok} ) \Rightarrow (s', h', B') \in \text{WPO} [\![ \psi ,
\texttt{free($t$)} , \textit{ok} ]\!]  $.


\subsubsection{Case 1}

{
We  prove the case for 
$\psi = \arr (t_1, {{t'}}_1 ) *  \bigast_{i=2}^{k-1} \Arr (t_i, {{t'}}_i)_{\alpha_i}
 *\arr (t_k, {{t'}}_k ) *  \psi'$
 and  $\psi_{\text{pure}} \models 
t_1 < b(t) < {{t'}}_1 * t_k < e(t) < {{t'}}_k 
*
b(t) \approx t *  
\bigast_{i=1}^{k-1} {{t'}}_i \approx t_{i+1} * \bigast_{i=1}^{k-1} {t}_i <
t_{i+1} $  for some  $k, \psi'$.
}
We use $\tau$ as an abbreviation for 
$[T_b(\pi, t) := t] [T_e(\pi, t) := y] $  for brevity.

By the assumption, we have
$
(s', h', B') \models \exists y .  (\arr (t_1 , t) * \bararr ( t , y ) *
\arr(y, {{t'}}_k ) * \psi' ) \tau$, and then 
$(s'[y\mapsto v], h', B') \models (\arr (t_1 , t) * \bararr ( t , y ) *
\arr(y, {{t'}}_k ) * \psi' ) \tau$ for some $v\in\textsc{Loc}$. Let
$s''=s'[y\mapsto v]$.

Define
 \begin{align*}
  s= & s',\\
  B= & B'\cup\{[\inter{t}_{s'',B'},v)\}=B'\cup\{[\inter{t}_{s,B},v)\},\\
  h= &  h'
  [l \mapsto -  ]_{l= [\![ t ]\!]_{s,B}  }^{\inter{{t'}_1}_{s,B} - 1}
  \left([l \mapsto v_l  ]_{l= \inter{t_i}_{s,B}  }^{\inter{{t'}_i}_{s,B} - 1}
  \right)_{i=2}^{k-1}
  [l \mapsto -  ]_{l= \inter{t_k}_{s,B}  }^{v-1},
 \end{align*}
 where $v_{\inter{t_i}_{s,B}}=\inter{\alpha_i}_{s,B}$ 
 if  $\Arr(t_i,{t'}_i)_{\alpha_i}=t_i\mapsto\alpha_i$.
Otherwise, we do not care about the value of $v_l$, i.e., we simply consider $v_l$ as $-$.

 Since $t$ contains neither
 $y,b$, nor $e$, and hence we have $t\tau=t$ and
 $\inter{t}_{s',B'}=\inter{t}_{s,B}$. Since $(s'',h',B')$ contains a
 portion satisfying $\bararr(t,y)$, and hence $(\bigcup
 B')\cap[\inter{t}_{s'',B'},v)=\emptyset$.

 For any term $r\in\termS^-(\psi)$, if $\psi_{\text{pure}}\models
 b(r)\approx b(t)$, we have
 $\inter{b(r)}_{s,B}=\inter{b(t)}_{s,B}=\inter{t}_{s,B}=\inter{t}_{s',B'}$,
 and otherwise,
 we have $\inter{b(r)}_{s,B}=\inter{b(r)}_{s',B'}$.  If
 $\psi_{\text{pure}}\models e(r)\approx e(t)$, we have
 $\inter{e(r)}_{s,B}=v=\inter{y}_{s',B'}$, and otherwise, we have
 $\inter{e(r)}_{s,B}=\inter{e(r)}_{s',B'}$. Hence, we have
 $\inter{r}_{s,B}=\inter{r\tau}_{s',B'}$ for any term
 $r\in\termS(\psi)$.

 Since $\psi_{\rm pure}\models t<{t'}_1 * \bigast_{i=1}^{k-1}({t'}_i\approx
 t_{i+1} * t_i<t_{i+1}) * t_k\le e(t)<{t'}_k$, we have
 $\inter{t}_{s,B}<\inter{{t'}_1}_{s,B}=\inter{t_2}_{s,B}<\inter{{t'}_2}_{s,B}=\cdots\inter{t_k}_{s,B}<v<\inter{{t'}_k}_{s,B}$. Hence,
 $h$ is well-defined, and we have $(s,h,B)\models \psi$.
 Furthermore, by definition, we have $((s,h,B),(s',h',B'))\in\inter{{\tt free}(t)}_{\it ok}$.


\subsubsection{Case 2}

{
We  prove the case for 
$\psi =   \bigast_{i=1}^{k-1} \Arr (t_i, {{t'}}_i)_{\alpha_i}
 *\arr (t_k, {{t'}}_k ) *  \psi'$
 and  $\psi_{\text{pure}} \models 
t_1 \approx b(t)
* t_k < e(t) < {{t'}}_k 
*
b(t) \approx t *  
\bigast_{i=1}^{k-1} {{t'}}_i \approx t_{i+1} * \bigast_{i=1}^{k-1} {t}_i <
t_{i+1} $  for some  $k, \psi'$.
}
We use $\tau$ as an abbreviation for 
$[T_b(\pi, t) := t] [T_e(\pi, t) := y] $  for brevity.

We prove $(s', h', B') \models \textsf{wpo}( \psi , \texttt{free($t$)} ,
\textit{ok} ) \Rightarrow (s', h', B') \in \text{WPO} [\![ \psi ,
\texttt{free($t$)} , \textit{ok} ]\!]  $.

By the assumption, we have
$
(s', h', B') \models
 \exists y . 
 ( \bararr ( t , y ) *
\arr(y, {{t'}}_k ) * \psi' ) \tau$, and then 
$(s'[y\mapsto v], h', B') \models
( \bararr ( t , y ) *
\arr(y, {{t'}}_k ) * \psi' ) \tau$ for some $v\in\textsc{Loc}$. Let
$s''=s'[y\mapsto v]$.

Define
 \begin{align*}
  s= & s',\\
  B= & B'\cup\{[\inter{t}_{s'',B'},v)\}=B'\cup\{[\inter{t}_{s,B},v)\},\\
  h= &  h'
  \left([l \mapsto v_l  ]_{l= \inter{t_i}_{s,B}  }^{\inter{{t'}_i}_{s,B} - 1}
  \right)_{i=1}^{k-1}
  [l \mapsto -  ]_{l= \inter{t_k}_{s,B}  }^{v-1},
 \end{align*}
 where $v_{\inter{t_i}_{s,B}}=\inter{\alpha_i}_{s,B}$ 
 if  $\Arr(t_i,{t'}_i)_{\alpha_i}=t_i\mapsto\alpha_i$.
Otherwise, we do not care about the value of $v_l$, i.e., we simply consider $v_l$ as $-$.

 Since $t$ contains neither
 $y,b$, nor $e$, and hence we have $t\tau=t$ and
 $\inter{t}_{s',B'}=\inter{t}_{s,B}$. Since $(s'',h',B')$ contains a
 portion satisfying $\bararr(t,y)$, and hence $(\bigcup
 B')\cap[\inter{t}_{s'',B'},v)=\emptyset$.

 For any term $r\in\termS^-(\psi)$, if $\psi_{\text{pure}}\models
 b(r)\approx b(t)$, we have
 $\inter{b(r)}_{s,B}=\inter{b(t)}_{s,B}=\inter{t}_{s,B}=\inter{t}_{s',B'}$,
 and otherwise,
 we have $\inter{b(r)}_{s,B}=\inter{b(r)}_{s',B'}$.  If
 $\psi_{\text{pure}}\models e(r)\approx e(t)$, we have
 $\inter{e(r)}_{s,B}=v=\inter{y}_{s',B'}$, and otherwise, we have
 $\inter{e(r)}_{s,B}=\inter{e(r)}_{s',B'}$. Hence, we have
 $\inter{r}_{s,B}=\inter{r\tau}_{s',B'}$ for any term
 $r\in\termS(\psi)$.

{ 
Since $\psi_{\rm pure}\models t<{t'}_1 * \bigast_{i=1}^{k-1}({t'}_i\approx
 t_{i+1} * t_i<t_{i+1}) * t_k\le e(t)<{t'}_k$, we have
 $\inter{t}_{s,B}<\inter{{t'}_1}_{s,B}=\inter{t_2}_{s,B}<\inter{{t'}_2}_{s,B}=\cdots\inter{t_k}_{s,B}<v<\inter{{t'}_k}_{s,B}$.}
 Hence,
 $h$ is well-defined, and we have $(s,h,B)\models \psi$.
 Furthermore, by definition, we have $((s,h,B),(s',h',B'))\in\inter{{\tt free}(t)}_{\it ok}$.

\subsubsection{Case 3}

{
We  prove the case for 
$\psi =   \bigast_{i=1}^{k} \Arr (t_i, {{t'}}_i)_{\alpha_i} *  \psi'$
 and  $\psi_{\text{pure}} \models 
t_1 \approx b(t)
*  e(t) \approx {{t'}}_k 
*
b(t) \approx t *  
\bigast_{i=1}^{k-1} {{t'}}_i \approx t_{i+1} * \bigast_{i=1}^{k-1} {t}_i <
t_{i+1} $  for some  $k, \psi'$.
}
We use $\tau$ as an abbreviation for 
$[T_b(\pi, t) := t] [T_e(\pi, t) := y] $  for brevity.

We prove $(s', h', B') \models \textsf{wpo}( \psi , \texttt{free($t$)} ,
\textit{ok} ) \Rightarrow (s', h', B') \in \text{WPO} [\![ \psi ,
\texttt{free($t$)} , \textit{ok} ]\!]  $.

By the assumption, we have
$
(s', h', B') \models \exists y . 
( \bararr ( t , y ) * \psi' ) \tau$, and then 
$(s'[y\mapsto v], h', B') \models
( \bararr ( t , y ) * \psi' ) \tau$ for some $v\in\textsc{Loc}$. Let
$s''=s'[y\mapsto v]$.

Define
 \begin{align*}
  s= & s',\\
  B= & B'\cup\{[\inter{t}_{s'',B'},v)\}=B'\cup\{[\inter{t}_{s,B},v)\},\\
  h= &  h'
  \left([l \mapsto v_l  ]_{l= \inter{t_i}_{s,B}  }^{\inter{{t'}_i}_{s,B} - 1}
  \right)_{i=1}^{k},
 \end{align*}
 where $v_{\inter{t_i}_{s,B}}=\inter{\alpha_i}_{s,B}$ 
 if  $\Arr(t_i,{t'}_i)_{\alpha_i}=t_i\mapsto\alpha_i$.
Otherwise, we do not care about the value of $v_l$, i.e., we simply consider $v_l$ as $-$.

 Since $t$ contains neither
 $y,b$, nor $e$, and hence we have $t\tau=t$ and
 $\inter{t}_{s',B'}=\inter{t}_{s,B}$. Since $(s'',h',B')$ contains a
 portion satisfying $\bararr(t,y)$, and hence $(\bigcup
 B')\cap[\inter{t}_{s'',B'},v)=\emptyset$.

 For any term $r\in\termS^-(\psi)$, if $\psi_{\text{pure}}\models
 b(r)\approx b(t)$, we have
 $\inter{b(r)}_{s,B}=\inter{b(t)}_{s,B}=\inter{t}_{s,B}=\inter{t}_{s',B'}$,
 and otherwise,
 we have $\inter{b(r)}_{s,B}=\inter{b(r)}_{s',B'}$.  If
 $\psi_{\text{pure}}\models e(r)\approx e(t)$, we have
 $\inter{e(r)}_{s,B}=v=\inter{y}_{s',B'}$, and otherwise, we have
 $\inter{e(r)}_{s,B}=\inter{e(r)}_{s',B'}$. Hence, we have
 $\inter{r}_{s,B}=\inter{r\tau}_{s',B'}$ for any term
 $r\in\termS(\psi)$.

{ Since $\psi_{\rm pure}\models t<{t'}_1 * \bigast_{i=1}^{k-1}({t'}_i\approx
 t_{i+1} * t_i<t_{i+1}) * t_k < e(t) \approx {t'}_k$, we have
 $\inter{t}_{s,B}<\inter{{t'}_1}_{s,B}=\inter{t_2}_{s,B}<\inter{{t'}_2}_{s,B}=\cdots\inter{t_k}_{s,B}< v \approx \inter{{t'}_k}_{s,B}$.}
 Hence,
 $h$ is well-defined, and we have $(s,h,B)\models \psi$.
 Furthermore, by definition, we have $((s,h,B),(s',h',B'))\in\inter{{\tt free}(t)}_{\it ok}$.

\subsubsection{Case 4}

{
We  prove the case for 
$\psi = \arr (t_1, {{t'}}_1 ) *  \bigast_{i=2}^{k} \Arr (t_i, {{t'}}_i)_{\alpha_i}  *  \psi'$
 and  $\psi_{\text{pure}} \models 
t_1 < b(t) < {{t'}}_1 *  e(t) \approx {{t'}}_k 
*
b(t) \approx t *  
\bigast_{i=1}^{k-1} {{t'}}_i \approx t_{i+1} * \bigast_{i=1}^{k-1} {t}_i <
t_{i+1} $  for some  $k, \psi'$.
}
We use $\tau$ as an abbreviation for 
$[T_b(\pi, t) := t] [T_e(\pi, t) := y] $  for brevity.

We prove $(s', h', B') \models \textsf{wpo}( \psi , \texttt{free($t$)} ,
\textit{ok} ) \Rightarrow (s', h', B') \in \text{WPO} [\![ \psi ,
\texttt{free($t$)} , \textit{ok} ]\!]  $.

By the assumption, we have
$
(s', h', B') \models \exists y .  (\arr (t_1 , t) * \bararr ( t , y )  * \psi' ) \tau$, and then 
$(s'[y\mapsto v], h', B') \models (\arr (t_1 , t) * \bararr ( t , y )  * \psi' ) \tau$ for some $v\in\textsc{Loc}$. Let
$s''=s'[y\mapsto v]$.

Define
 \begin{align*}
  s= & s',\\
  B= & B'\cup\{[\inter{t}_{s'',B'},v)\}=B'\cup\{[\inter{t}_{s,B},v)\},\\
  h= &  h'
  [l \mapsto -  ]_{l= [\![ t ]\!]_{s,B}  }^{\inter{{t'}_1}_{s,B} - 1}
  \left([l \mapsto v_l  ]_{l= \inter{t_i}_{s,B}  }^{\inter{{t'}_i}_{s,B} - 1}
  \right)_{i=2}^{k},
 \end{align*}
 where $v_{\inter{t_i}_{s,B}}=\inter{\alpha_i}_{s,B}$ 
 if  $\Arr(t_i,{t'}_i)_{\alpha_i}=t_i\mapsto\alpha_i$.
Otherwise, we do not care about the value of $v_l$, i.e., we simply consider $v_l$ as $-$.

 Since $t$ contains neither
 $y,b$, nor $e$, and hence we have $t\tau=t$ and
 $\inter{t}_{s',B'}=\inter{t}_{s,B}$. Since $(s'',h',B')$ contains a
 portion satisfying $\bararr(t,y)$, and hence $(\bigcup
 B')\cap[\inter{t}_{s'',B'},v)=\emptyset$.

 For any term $r\in\termS^-(\psi)$, if $\psi_{\text{pure}}\models
 b(r)\approx b(t)$, we have
 $\inter{b(r)}_{s,B}=\inter{b(t)}_{s,B}=\inter{t}_{s,B}=\inter{t}_{s',B'}$,
 and otherwise,
 we have $\inter{b(r)}_{s,B}=\inter{b(r)}_{s',B'}$.  If
 $\psi_{\text{pure}}\models e(r)\approx e(t)$, we have
 $\inter{e(r)}_{s,B}=v=\inter{y}_{s',B'}$, and otherwise, we have
 $\inter{e(r)}_{s,B}=\inter{e(r)}_{s',B'}$. Hence, we have
 $\inter{r}_{s,B}=\inter{r\tau}_{s',B'}$ for any term
 $r\in\termS(\psi)$.

{ Since $\psi_{\rm pure}\models t<{t'}_1 * \bigast_{i=1}^{k-1}({t'}_i\approx
 t_{i+1} * t_i<t_{i+1}) * t_k\le e(t)<{t'}_k$, we have
 $\inter{t}_{s,B}<\inter{{t'}_1}_{s,B}=\inter{t_2}_{s,B}<\inter{{t'}_2}_{s,B}=\cdots\inter{t_k}_{s,B}<v<\inter{{t'}_k}_{s,B}$.} Hence,
 $h$ is well-defined, and we have $(s,h,B)\models \psi$.
 Furthermore, by definition, we have $((s,h,B),(s',h',B'))\in\inter{{\tt free}(t)}_{\it ok}$.

\subsubsection{Case 5}
Otherwise, 
$(s', h', B') \models \textsf{wpo}( \psi , \texttt{free($t$)} ,
\textit{ok} ) \Rightarrow (s', h', B') \in \text{WPO} [\![ \psi ,
\texttt{free($t$)} , \textit{ok} ]\!]  $ holds trivially 
since $\textsf{wpo}( \psi , \texttt{free($t$)} ,
\textit{ok} )  = \textsf{false}$.


\subsection{$(s', h', B') \in
\text{WPO} [\![ \psi , \texttt{free($t$)} , \textit{ok} ]\!] \Rightarrow
(s', h', B') \models \textsf{wpo}( \psi ,
\texttt{free($t$)} , \textit{ok} )$}

Next, we prove the converse: $(s', h', B') \in
\text{WPO} [\![ \psi , \texttt{free($t$)} , \textit{ok} ]\!] \Rightarrow
(s', h', B') \models \textsf{wpo}( \psi ,
\texttt{free($t$)} , \textit{ok} )$.
By the assumption, there exists $(s,h,B)$ such that
\begin{itemize}
 \item $(s,h,B)\models\psi$,
 \item $(s',h',B')=(s,h[l\mapsto
       \bot]_{l=\inter{t}_{s,B}}^{e_B(\inter{t}_{s,B})-1},
       B \setminus \{[\inter{t}_{s,B},e_B(\inter{t_{s,B}}))\})$,
 \item $\inter{t}_{s,B}=b_B(\inter{t}_{s,B})>0$.
\end{itemize}

\subsubsection{Case 1}
{
We  prove the case for 
$\psi = \arr (t_1, {{t'}}_1 ) *  \bigast_{i=2}^{k-1} \Arr (t_i, {{t'}}_i)_{\alpha_i}
 *\arr (t_k, {{t'}}_k ) *  \psi'$
 and  $\psi_{\text{pure}} \models 
t_1 < b(t) < {{t'}}_1 * t_k < e(t) < {{t'}}_k 
*
b(t) \approx t *  
\bigast_{i=1}^{k-1} {{t'}}_i \approx t_{i+1} * \bigast_{i=1}^{k-1} {t}_i <
t_{i+1} $  for some  $k, \psi'$.
}
We use $\tau$ as an abbreviation for 
$[T_b(\pi, t) := t] [T_e(\pi, t) := y] $  for brevity.

Since $ \psi = \arr (t_1, {{t'}}_1 ) * \bigast_{i=2}^{k-1} \Arr (t_i,
{{t'}}_i)_{\alpha_i} *\arr (t_k, {{t'}}_k ) * \psi'$,
$h$ is of the form $h_1\circ h_2\circ \cdots\circ h_k\circ h''$ such that
\begin{itemize}
 \item $ (s,h_1,B)\models \arr(t_1,{{t'}}_1)$,
 \item $ (s,h_i,B)\models \Arr(t_i,{{t'}}_i)_{\alpha_i}$ for $2\le i\le k-1$,
 \item $ (s,h_k,B)\models \arr(t_k,{{t'}}_k)$,
 \item $(s,h'',B)\models \psi'$.
\end{itemize}
{
Since $(s,B)$ satisfies $t_1<t\approx b(t)<{t'}_1$ and $t_k<e(t)<{t'}_k$, 
$h_1$ is divided into two portions $h^1_1$ and $h^2_1$,
where $h^1_1$ is $h_1$ whose domain is shorten to
$[\inter{t_1}_{s,B},\inter{t}_{s,B})$ and $h^2_1(l)=\bot$ for
$l\in[\inter{t}_{s,B},\inter{{t'}_1}_{s,B})$.  
Similarly, $h_k$ is divided
into two portions $h^1_k$ and $h^2_k$,  
where $h^2_k$ is $h_k$ whose domain is shorten to
$[ e_B(\inter{t}_{s,B}) , \inter{{t'}_k}_{s,B}  )$ and 
$h^1_k(l)=\bot$ for
$l\in[ \inter{t_k}_{s,B} , e_B(\inter{t}_{s,B}) )$.  
}
Furthermore, we have
$(s,B)\models {t'}_i\approx t_{i+1}$ for $1\le i<k$. Then, $h'$ can be
written as $  h^1_1\circ h_\bot \circ h^2_k\circ h''$, where
$\textsf{dom}(h_\bot)=\textsf{dom}_{-}(h_\bot)=[\inter{t}_{s,B}, e_B(\inter{t}_{s,B}) )$.

As in the previous case, $\inter{r}_{s,B}=\inter{r\tau}_{s',B'}$
holds. 
Therefore, we have $ (s'[y\mapsto e_B(\inter{t}_{s,B})],h',B')\models
(\arr (t_1 , t) * \bararr ( t , y ) *
\arr(y, {{t'}}_k ) * \psi' ) \tau$, and then 
$ (s', h', B') \models\exists y.(\arr (t_1 , t) * \bararr ( t , y ) *
\arr(y, {{t'}}_k ) * \psi' ) \tau$. That means $(s',h',B')$ satisfies
${\sf wpo}(\psi,{\tt free}(t),{\it ok})$.


\subsubsection{Case 2}
{
We  prove the case for  
$\psi =   \bigast_{i=1}^{k-1} \Arr (t_i, {{t'}}_i)_{\alpha_i}
 *\arr (t_k, {{t'}}_k ) *  \psi'$
 and  $\psi_{\text{pure}} \models 
t_1 \approx b(t)
* t_k < e(t) < {{t'}}_k 
*
b(t) \approx t *  
\bigast_{i=1}^{k-1} {{t'}}_i \approx t_{i+1} * \bigast_{i=1}^{k-1} {t}_i <
t_{i+1} $  for some  $k, \psi'$.
}
We use $\tau$ as an abbreviation for 
$[T_b(\pi, t) := t] [T_e(\pi, t) := y] $  for brevity.

Since $ \psi =   \bigast_{i=1}^{k-1} \Arr (t_i,
{{t'}}_i)_{\alpha_i} *\arr (t_k, {{t'}}_k ) * \psi'$,
$h$ is of the form $h_1\circ h_2\circ \cdots\circ h_k\circ h''$ such that
\begin{itemize}
 \item $ (s,h_i,B)\models \Arr(t_i,{{t'}}_i)_{\alpha_i}$ for $1\le i\le k-1$,
 \item $ (s,h_k,B)\models \arr(t_k,{{t'}}_k)$,
 \item $(s,h'',B)\models \psi'$.
\end{itemize}
{
Since $(s,B)$ satisfies $t_1 \approx t\approx b(t)<{t'}_1$ and $t_k<e(t)<{t'}_k$, 
$h_k$ is divided
into two portions $h^1_k$ and $h^2_k$,  
where $h^2_k$ is $h_k$ whose domain is shorten to
$[ e_B(\inter{t}_{s,B}) , \inter{{t'}_k}_{s,B}  )$ and 
$h^1_k(l)=\bot$ for
$l\in[ \inter{t_k}_{s,B} , e_B(\inter{t}_{s,B}) )$.  
}
Furthermore, we have
$(s,B)\models {t'}_i\approx t_{i+1}$ for $1\le i<k$. Then, $h'$ can be
written as $   h_\bot \circ h^2_k\circ h''$, where
$\textsf{dom}(h_\bot)=\textsf{dom}_{-}(h_\bot)=[\inter{t}_{s,B}, e_B(\inter{t}_{s,B}) )$.

As in the previous case, $\inter{r}_{s,B}=\inter{r\tau}_{s',B'}$
holds. 
Therefore, we have $ (s'[y\mapsto e_B(\inter{t}_{s,B})],h',B')\models
(  \bararr ( t , y ) *
\arr(y, {{t'}}_k ) * \psi' ) \tau$, and then 
$ (s', h', B') \models\exists y.(  \bararr ( t , y ) *
\arr(y, {{t'}}_k ) * \psi' ) \tau$. That means $(s',h',B')$ satisfies
${\sf wpo}(\psi,{\tt free}(t),{\it ok})$.


\subsubsection{Case 3}
{
We  prove the case for 
$\psi =   \bigast_{i=1}^{k} \Arr (t_i, {{t'}}_i)_{\alpha_i} *  \psi'$
 and  $\psi_{\text{pure}} \models 
t_1 \approx b(t)
*  e(t) \approx {{t'}}_k 
*
b(t) \approx t *  
\bigast_{i=1}^{k-1} {{t'}}_i \approx t_{i+1} * \bigast_{i=1}^{k-1} {t}_i <
t_{i+1} $  for some  $k, \psi'$.
}
We use $\tau$ as an abbreviation for 
$[T_b(\pi, t) := t] [T_e(\pi, t) := y] $  for brevity.

Since $ \psi =   \bigast_{i=1}^{k} \Arr (t_i,
{{t'}}_i)_{\alpha_i} * \psi'$,
$h$ is of the form $h_1\circ h_2\circ \cdots\circ h_k\circ h''$ such that
\begin{itemize}
 \item $ (s,h_i,B)\models \Arr(t_i,{{t'}}_i)_{\alpha_i}$ for $1\le i\le k$,
 \item $(s,h'',B)\models \psi'$.
\end{itemize}
{
$(s,B)$ satisfies $t_1 \approx t\approx b(t)<{t'}_1$ and 
$e(t) \approx {t'}_k$,
}
and we have
$(s,B)\models {t'}_i\approx t_{i+1}$ for $1\le i<k$. 
Then, $h'$ can be
written as $   h_\bot \circ  h''$, where
$\textsf{dom}(h_\bot)=\textsf{dom}_{-}(h_\bot)=[\inter{t}_{s,B}, e_B(\inter{t}_{s,B}) )$.

As in the previous case, $\inter{r}_{s,B}=\inter{r\tau}_{s',B'}$
holds. 
Therefore, we have $ (s'[y\mapsto e_B(\inter{t}_{s,B})],h',B')\models
(  \bararr ( t , y ) * \psi' ) \tau$, and then 
$ (s', h', B') \models\exists y.(  \bararr ( t , y ) * \psi' ) \tau$. That means $(s',h',B')$ satisfies
${\sf wpo}(\psi,{\tt free}(t),{\it ok})$.

\subsubsection{Case 4}

{
We  prove the case for 
$\psi = \arr (t_1, {{t'}}_1 ) *  \bigast_{i=2}^{k} \Arr (t_i, {{t'}}_i)_{\alpha_i}  *  \psi'$
 and  $\psi_{\text{pure}} \models 
t_1 < b(t) < {{t'}}_1 *  e(t) \approx {{t'}}_k 
*
b(t) \approx t *  
\bigast_{i=1}^{k-1} {{t'}}_i \approx t_{i+1} * \bigast_{i=1}^{k-1} {t}_i <
t_{i+1} $  for some  $k, \psi'$.
}
We use $\tau$ as an abbreviation for 
$[T_b(\pi, t) := t] [T_e(\pi, t) := y] $  for brevity.

Since $ \psi = \arr (t_1, {{t'}}_1 ) * \bigast_{i=2}^{k} \Arr (t_i,
{{t'}}_i)_{\alpha_i} * \psi'$,
$h$ is of the form $h_1\circ h_2\circ \cdots\circ h_k\circ h''$ such that
\begin{itemize}
 \item $ (s,h_1,B)\models \arr(t_1,{{t'}}_1)$,
 \item $ (s,h_i,B)\models \Arr(t_i,{{t'}}_i)_{\alpha_i}$ for $2\le i\le k$,
 \item $(s,h'',B)\models \psi'$.
\end{itemize}
{
Since $(s,B)$ satisfies $t_1<t\approx b(t)<{t'}_1$ and $e(t) \approx {t'}_k$, 
$h_k$ is divided
into two portions $h^1_k$ and $h^2_k$,  
where $h^2_k$ is $h_k$ whose domain is shorten to
$[ e_B(\inter{t}_{s,B}) , \inter{{t'}_k}_{s,B}  )$ and 
$h^1_k(l)=\bot$ for
$l\in[ \inter{t_k}_{s,B} , e_B(\inter{t}_{s,B}) )$.  
}
Furthermore, we have
$(s,B)\models {t'}_i\approx t_{i+1}$ for $1\le i<k$. Then, $h'$ can be
written as $   h_\bot \circ h^2_k\circ h''$, where
$\textsf{dom}(h_\bot)=\textsf{dom}_{-}(h_\bot)=[\inter{t}_{s,B}, e_B(\inter{t}_{s,B}) )$.

As in the previous case, $\inter{r}_{s,B}=\inter{r\tau}_{s',B'}$
holds. 
Therefore, we have $ (s'[y\mapsto e_B(\inter{t}_{s,B})],h',B')\models
(\arr (t_1 , t) * \bararr ( t , y ) * \psi' ) \tau$, and then 
$ (s', h', B') \models\exists y.(\arr (t_1 , t) * \bararr ( t , y )  * \psi' ) \tau$. That means $(s',h',B')$ satisfies
${\sf wpo}(\psi,{\tt free}(t),{\it ok})$.

\subsubsection{Case 5}

The following is the union of all cases between Case 1 and Case 4. 

\begin{align*}
&\psi = \bigast_{i=1}^{k} \Arr (t_i, {{t'}}_i)_{\alpha_i} * \psi' \text{ and } \\
&\psi_{\text{pure}} \models 
t_1 \leq b(t) < {{t'}}_1 * t_k < e(t) \leq {{t'}}_k * \\
& b(t) \approx t *  
\bigast_{i=1}^{k-1} {{t'}}_i \approx t_{i+1} * \bigast_{i=1}^{k-1} {t}_i < t_{i+1} \\
&\text{for some $k, \psi'$}
\end{align*}

The following is equivalent to the above 
since we have that $t_i < {t'}_i$ and ${{t'}}_i \approx t_{i+1}$.

\begin{align*}
&\psi = \bigast_{i=1}^{k} \Arr (t_i, {{t'}}_i)_{\alpha_i} * \psi' \text{ and } \\
&\psi_{\text{pure}} \models 
t_1 \leq b(t) < {{t'}}_1 * t_k < e(t) \leq {{t'}}_k * \\
& b(t) \approx t *  
\bigast_{i=1}^{k-1} {{t'}}_i \approx t_{i+1} 
\text{ for some $k, \psi'$}
\end{align*}

The negation of the above formula is as follows:

\begin{align*}
& \text{If } \psi = \bigast_{i=1}^{k} \Arr (t_i, {{t'}}_i)_{\alpha_i} * \psi' \text{ and} &  \\
& \psi_{\text{pure}} \models 
t_1 \leq b(t) < {{t'}}_1 * t_k < e(t) * b(t) \approx t *  
\bigast_{i=1}^{k-1} {{t'}}_i \approx t_{i+1} 
 &  \\
&\text{for some $k,\psi'$, then } \psi_{\text{pure}} \models e(t) > {t'}_k \text{ holds.} 
& \cdots (\dagger) 
\end{align*}

In this case, we prove \(\text{WPO} \llbracket \psi , \texttt{free}(t) , \textit{ok} \rrbracket = \emptyset\)  
by contradiction.
That is, we assume that \(\text{WPO} \llbracket \psi , \texttt{free}(t) , \textit{ok} \rrbracket \neq \emptyset\)  
and show that this leads to a contradiction.

Suppose $(s, h, B) \models \psi$ and 
$( (s, h, B),  \sigma ) \in \llbracket \texttt{free($t$)} \rrbracket_{\text{ok}}$ for some $\sigma$.
By the definition of $\llbracket \texttt{free($t$)} \rrbracket_{\text{ok}}$, we have:
\[ \llbracket t \rrbracket_{s, B} = b_B(\llbracket t \rrbracket_{s, B}) > 0. \]
Since we assume that $(s, h, B)$ is exact, $\llbracket t \rrbracket_{s, B} \in \textsf{dom}_{+}(h)$ holds.
$\psi$ contains $\Arr (t_1, {t'}_1)$ for some $t_1, {t'}_1$ such that
$\psi_{\text{pure}} \models t_1 \leq t < {t'}_1 * t \approx b(t)$.

Since \((s, h, B)\) is exact,  
\[
[\llbracket b(t) \rrbracket_{s, B}, \llbracket e(t) \rrbracket_{s, B})  
\subseteq \textsf{dom}_{+}(h)
\]
holds, and \(\psi\) contains \(\Arr (t_1, {t'}_1) * \dots * \Arr (t_k, {t'}_k)\)  
such that  
$\psi_{\text{pure}} \models   e(t) > t_k  * \bigast_{i=1}^{k-1} {t'}_{i} \approx t_{i+1}$
for some $k$.

Fix \( k \) as the maximum among such values of \( k \).
Since \( k \) is maximal, \(\psi_{\text{pure}} \models {t'}_k \not\approx u\)  
holds for any \(\Arr (u, u')\) in \(\psi\).
Because \( k \) is chosen to be maximal, the chain of arrays  
\(\Arr (t_1, t'_1), \dots, \Arr (t_k, {t'}_k)\)  
cannot be extended by adding another array at the end.

If there were an array \(\Arr (u, u')\) in \(\psi\) such that  
\({t'}_k \approx u\), then this array could be attached to the end of the chain,  
creating a longer sequence of arrays.

However, since \( k \) is the maximum possible length, such an extension is not possible.  
Therefore, for every array \(\Arr (u, u')\) in \(\psi\), we must have  
\({t'}_k \not\approx u\).  
This means that \({t'}_k\) does not match the starting point \( u \)  
of any array in  \(\psi\).
Therefore, the following holds.
$$\llbracket {t'}_k \rrbracket_{s, B} \notin \textsf{dom}_{+}(h)$$

By \((\dagger)\), we have \(\psi_{\text{pure}} \models e(t) > {t'}_k\),  
which means \(\llbracket e(t) \rrbracket_{s, B} > \llbracket {t'}_k \rrbracket_{s, B}\).  
Additionally, by the condition, we have  
\[
\llbracket b(t) \rrbracket_{s, B} < \llbracket {t'}_k \rrbracket_{s, B} < \llbracket e(t) \rrbracket_{s, B}.
\]  
However, we also have \(\llbracket {t'}_k \rrbracket_{s, B} \notin \textsf{dom}_{+}(h)\),  
which contradicts the property of exactness:  
$[ \llbracket b(t) \rrbracket_{s,B} , \llbracket e(t) \rrbracket_{s,B} ) \subseteq \textsf{dom}_+(h)$.

Consequently, we conclude that \((s,h,B)\) is not exact,  
and thus  
\[
\text{WPO} \llbracket \psi , \texttt{free}(t) , \textit{ok} \rrbracket = \emptyset
\]  
holds.  
Please note that in the current setting, for all \(\sigma\) and \(\sigma'\) where  
\( (\sigma, \sigma') \in  \llbracket  \texttt{free}(t)  \rrbracket_\textit{ok} \),  
\(\sigma\) and $\sigma'$ are exact.

%
%
%


\subsection{Case for \textit{er}}
We consider the case of {\it er} and $\psi_\text{pure} \models
b(t) \not\approx t \lor b(t) \approx \texttt{null}$.  In this case, we
have $\textsf{wpo}(\psi, \texttt{free($t$)}, \textit{er}) = \psi$.  By
the definition of $\inter{{\tt free}(t)}_{\it er}$, we have $(s,h,B)\in
{\rm WPO}\inter{\psi,{{\tt free}(t)},{\it er}}$ iff $(s,h,B)\models\psi$
and either $b_B(\inter{t}_{s,B})\ne \inter{t}_{s,B}$ or
$b_B(\inter{t}_{s,B})=0$, which is equivalent to $(s,h,B)\models \psi$
by the assumption $\psi_\text{pure} \models b(t) \not\approx t \lor b(t)
\approx \texttt{null}$.

Otherwise, both $\psi_{\rm pure}\models b(t)\approx t$ and
$\psi_{\rm pure}\models b(t)\not\approx {\tt null}$ hold. Hence, ${\rm
WPO}\inter{\psi,{\tt free}(t),{\it er}}=\emptyset$.


\section{Proof of Proposition \ref{prop: expressiveness of WPO calculus} for $x := [t]$}
\label{sec: expressiveness proof of atomic command in appendix (load)}

We prove
$(s', h', B') \models \textsf{wpo}(  \psi , x := [t]  , \textit{ok} ) \iff 
(s', h', B') \in \text{WPO} [\![  \psi , x := [t]  , \textit{ok} ]\!]  $ holds.
By definition, $(s', h', B') \in \text{WPO} [\![  \psi , x := [t]   , \textit{ok} ]\!]  $
is equivalent to
$ \exists (s,h,B) \models \psi .  b_B ( [\![ t ]\!]_{s,B} ) > 0 \land
(s', h', B')  = 
(s [x \mapsto h ( [\![ t ]\!]_{s,B}  ) ] , h  , B   )  $.

\subsection{$(s', h', B') \models \textsf{wpo}( \psi , x := [t] ,
\textit{ok} ) \Rightarrow (s', h', B') \in \text{WPO} [\![ \psi , x :=
[t] , \textit{ok} ]\!]  $}

\subsubsection{Case 1}
Here, we  prove the case for $ \psi=t_a\mapsto u*\psi'$ and
$\psi_{\rm pure}\models t_a\approx t$ for some $t_a,u$, and $\psi'$.

We have the following.

{
\begin{align*}
& 
(s', h', B')  \models   \exists x' .
\psi [x := x'] * x \approx u [x := x']  \\
\Leftrightarrow \ & \exists v'. 
(s' [x' \mapsto v'] , h', B')    \models 
\psi [x := x'] * x \approx u [x := x'] \\
\Rightarrow \ & \exists v'. 
(s' [x' \mapsto v'] , h', B')    \models 
\psi [x := x']  \text{ and }
s' [x' \mapsto v'] (x) = [\![ u [x := x'] ]\!]_{s' [x' \mapsto v']  , B'} \\
\Leftrightarrow \ & \exists v'. 
(s' [x' \mapsto v'] [x \mapsto  v'] , h', B')    \models \psi  \text{ and }
s' [x' \mapsto v'] (x) = [\![ u  ]\!]_{s' [x' \mapsto v'] [x \mapsto v']  , B'} \\
\text{// }  & \text{By Lemma \ref{lma: Substitution for assignment}} \\ 
\Leftrightarrow \ & \exists v'. 
(s'  [x \mapsto  v'] , h', B')    \models \psi  \text{ and }
s'  (x) = [\![ u  ]\!]_{s' [x \mapsto v']  , B'} \\
\text{// }  & \text{Since $x'$ is fresh} \\ 
\Leftrightarrow \ & \exists v'. 
(s'  [x \mapsto  v'] , h', B')    \models \psi  \text{ and }
s'  (x) = h' ( [\![ t  ]\!]_{s' [x \mapsto v']  , B'} ) \\
\text{// }  & \text{By premise, $\psi \iff t_a \approx t * t_a \mapsto u * \psi'$ holds. } \\
\text{// }  & \text{Consequently, $ [\![ u  ]\!]_{s' [x \mapsto v']  , B'} = 
h' ( [\![ t  ]\!]_{s' [x \mapsto v']  , B'} ) $ also holds.} 
\end{align*}
}
Then, fix $v'$ and let $s$ be $s'[x\mapsto v']$, and then we have
$(s,h',B')\models \psi$. Since $ \psi=t_a\mapsto u*\psi'$ and 
$ \psi_{\rm
pure}\models t_a\approx t$ hold and $(h',B')$ is exact,
$b_{B'}(\inter{t}_{s,B'})=b_{B'}(\inter{t_a}_{s,B'})>0$ holds.
Furthermore, we have $s'(x)=h'(\inter{t}_{s'[x\mapsto v'],B'})
=h'(\inter{t}_{s,B'})$, and hence $s'=s[x\mapsto
h(\inter{t}_{s,B})]$.  Therefore, we have
$((s,h',B'),(s',h',B'))\in\inter{x:=[t]}_{\it ok}$.

\subsubsection{Case 2}

Here, we  prove the case for $ \psi= \arr (t_a , {t'}_a) *\psi'$ and
$\psi_{\rm pure}\models t_a \leq  t < {t'}_a $ 
for some $t_a,{t'}_a$, and $\psi'$.

Firstly, we prove $(s', h', B') \models \textsf{wpo}( \psi , x := [t] ,
\textit{ok} ) \Rightarrow (s', h', B') \in \text{WPO} [\![ \psi , x :=
[t] , \textit{ok} ]\!]  $. We have the following.

{
\begin{align*}
& 
(s', h', B')  \models   \exists x' .
\psi [x := x']   \\
\Leftrightarrow \ & \exists v'. 
(s' [x' \mapsto v'] , h', B')    \models 
\psi [x := x']  \\
\Leftrightarrow \ & \exists v'. 
(s' [x' \mapsto v'] [x \mapsto  v'] , h', B')    \models \psi   \\
\text{// }  & \text{By Lemma \ref{lma: Substitution for assignment}} \\ 
\Leftrightarrow \ & \exists v'. 
(s'  [x \mapsto  v'] , h', B')    \models \psi   
\end{align*}
}
Then, fix $v'$ and let $s$ be $s'[x\mapsto v']$, and then we have
$(s,h',B')\models \psi$. Since $ \psi= \arr (t_a , {t'}_a) *\psi' $ and 
$\psi_{\rm pure}\models t_a \leq  t < {t'}_a $  hold and $(h',B')$ is exact,
$b_{B'}(\inter{t}_{s,B'})=b_{B'}(\inter{t_a}_{s,B'})>0$ holds.
Furthermore, we have $s'(x)=h'(\inter{t}_{s'[x\mapsto v'],B'})
=h'(\inter{t}_{s,B'})$, and hence $s'=s[x\mapsto
h(\inter{t}_{s,B})]$.  Therefore, we have
$((s,h',B'),(s',h',B'))\in\inter{x:=[t]}_{\it ok}$.

\subsubsection{Case 3}

Otherwise, 
$(s', h', B') \models \textsf{wpo}( \psi , x := [t] ,
\textit{ok} ) \Rightarrow (s', h', B') \in \text{WPO} [\![ \psi ,
x := [t], \textit{ok} ]\!]  $ holds trivially 
since $\textsf{wpo}( \psi , x := [t] ,
\textit{ok} )  = \textsf{false}$.

\subsection{$(s', h', B') \in \text{WPO} [\![  \psi , x := [t] , \textit{ok} ]\!]
\Rightarrow
(s', h', B') \models \textsf{wpo}(  \psi , x := [t] , \textit{ok} )$}
Next, we prove the converse.

\subsubsection{Case 1}
Here, we  prove the case for $ \psi=t_a\mapsto u*\psi'$ and
$\psi_{\rm pure}\models t_a\approx t$ for some $t_a,u$, and $\psi'$.

By the assumption, there exists $s$ such that
\begin{itemize}
 \item $(s,h',B')\models \psi$,
 \item $b_{B'}(\inter{t}_{s,B'})>0$, and
 \item $s'=s[x\mapsto h'(\inter{t}_{s,B'})]$.
\end{itemize}

Since $x'$ is fresh, they do not depend on the value of
$s(x')$, and hence we can assume $s(x')=s(x)$, and then we have
$ s = s'[x' \mapsto s(x)]  [x \mapsto s(x) ] $.
Since we have $(s,h',B') \models \psi$,  
we can derive  
{\begin{align*}
&(s,h',B') \models \psi \\ 
\Rightarrow \ &(s'[x' \mapsto s(x)] [x \mapsto s(x)], h', B') \models \psi \\
\Rightarrow \ & (s'[x' \mapsto s(x)], h', B') \models \psi[x := x']
\end{align*}}
by Lemma \ref{lma: Substitution for assignment} and the fact that  
$s'[x' \mapsto s(x)](x') = s(x)$.
Let $s''$ be $s'[x' \mapsto s(x)]$.
$s''(x) = s' [x' \mapsto s(x)] (x) = s'(x) = h' (\inter{t}_{s,B'})$ by the definition.
The equality  
$ h' (\inter{t}_{s,B'}) = h' (\inter{t_a}_{s,B'})$
holds by $(s,h',B) \models \psi$ and $ \psi_\text{pure} \models t \approx t_a$.  
$ h' (\inter{t_a}_{s,B'}) = \inter{u}_{s,B'}$ since $\psi$ contains 
$ t_a \mapsto u$.
Additionally, the following holds.
\begin{align*}
    \inter{u}_{s,B'} &=  \inter{u}_{s'[x' \mapsto s(x) ] [x \mapsto s(x)],B'}  
    & \text{// Definition of $s$} \\
    &= \inter{u[x:=x'] }_{s'[x' \mapsto s(x) ] ,B'}  & 
    \text{// By Lemma \ref{lma: Substitution for assignment}} \\
    &\   & 
    \text{// and $s'[x' \mapsto s(x)](x') = s(x)$} \\
        &= \inter{u [x:=x']  }_{s'' ,B'} 
\end{align*}

Hence, $ (s'',B') \models x \approx u[x := x']$, and  
$ (s'',h',B') \models \psi[x := x'] * x \approx u[x := x']$,  
since $ s''(x) = \inter{u [x := x']}_{s'',B'}$.
Then, we have $ (s',h',B')\models \exists x'.\psi[x:=x']*x\approx
u[x:=x']$,
which is ${\sf wpo}(\psi,x:=[t],{\it ok})$.

\subsubsection{Case 2}
Here, we  prove the case for $ \psi= \arr (t_a , {t'}_a) *\psi'$ and
$\psi_{\rm pure}\models t_a \leq  t < {t'}_a $ 
for some $t_a,{t'}_a$, and $\psi'$.

Since $x'$ is fresh, they do not depend on the value of
$s(x')$, and hence we can assume $s(x')=s(x)$, and then we have
$ s = s'[x' \mapsto s(x)]  [x \mapsto s(x) ] $.
Since we have $(s,h',B') \models \psi$,  
we can derive  
{\begin{align*}
&(s,h',B') \models \psi \\ 
\Rightarrow \ &(s'[x' \mapsto s(x)] [x \mapsto s(x)], h', B') \models \psi \\
\Rightarrow \ & (s'[x' \mapsto s(x)], h', B') \models \psi[x := x'] \\
\Rightarrow \ & \exists v' . (s'[x' \mapsto v'], h', B') \models \psi[x := x'] \\
\Leftrightarrow \ &  (s', h', B') \models \exists x' .  \psi[x := x'] 
\end{align*}}
by Lemma \ref{lma: Substitution for assignment} and the fact that  
$s'[x' \mapsto s(x)](x') = s(x)$.
Then, we have $ (s',h',B')\models \exists x'.\psi[x:=x']$,
which is ${\sf wpo}(\psi,x:=[t],{\it ok})$.

\subsubsection{Case 3}

The following is the union of all cases between Case 1 and Case 2. 

\begin{align*}
&\text{$\psi$ contains some $\Arr (t_a, {{t'}}_a)_{\alpha}$
such that} \\
&\text{$\psi_{\text{pure}} \models 
t_a \leq t < {t'}_a$}
\end{align*}

The negation of the above formula is as follows:

\begin{align*}
&\text{For any $\Arr (t_a, {{t'}}_a)_{\alpha}$ in $\psi$, } \\
&\text{$\psi_{\text{pure}} \not\models 
t_a \leq t < {t'}_a$}
\end{align*}

Suppose that \((s,h,B) \models \psi\).  
By the given condition, we have  
\[
\llbracket t \rrbracket_{s,B} \notin \textsf{dom}_+ (h).
\]  
Hence, there is no \(\sigma\) such that  
\[
((s,h,B) , \sigma ) \in \llbracket x := [t] \rrbracket_\textit{ok}.
\]  
Therefore, \(\text{WPO} \llbracket \psi , x := [t] , \textit{ok} \rrbracket = \emptyset\) holds.

\subsection{Case for \textit{er}}

Next, we consider the case of {\it er} and $\psi_\text{pure} \models
b(t) \approx \texttt{null}$.  In this case, we
have $\textsf{wpo}(\psi, x:=[t], \textit{er}) = \psi$.  By
the definition of $\inter{x:=[t]}_{\it er}$, we have $(s,h,B)\in
{\rm WPO}\inter{\psi,x:=[t],{\it er}}$ iff $(s,h,B)\models\psi$
and $b_B(\inter{t}_{s,B})=0$, which is equivalent to $(s,h,B)\models \psi$
by the assumption $\psi_\text{pure} \models b(t)
\approx \texttt{null}$.

Otherwise, $\psi_{\rm pure}\models
 b(t)\not\approx {\tt null}$ hold. Hence, ${\rm
WPO}\inter{\psi,x:=t,{\it er}}=\emptyset$.

\section{Proof of Proposition \ref{prop: expressiveness of WPO calculus} for $[t] := t'$}
\label{sec: expressiveness proof of atomic command in appendix (store)}

We prove
$(s', h', B') \models \textsf{wpo}(  \psi , [t] := {t'}  , \textit{ok} ) \iff 
(s', h', B') \in \text{WPO} [\![  \psi ,  [t] := {t'}  , \textit{ok} ]\!]  $ holds.
By definition, $(s', h', B') \in \text{WPO} [\![  \psi ,  [t] := {t'}   , \textit{ok} ]\!]  $
is equivalent to
$ \exists (s,h,B) \models \psi .  b_B ( [\![ t ]\!]_{s,B} ) > 0 \land
(s', h', B')  = 
(s  , h [ [\![ t ]\!]_{s,B} \mapsto [\![ {t'} ]\!]_{s,B}  ]  , B   )  $.

\subsection{$(s', h', B') \models \textsf{wpo}(  \psi , [t] := {t'}  , \textit{ok} ) \Rightarrow
(s', h', B') \in \text{WPO} [\![  \psi , [t] := {t'}  , \textit{ok} ]\!]
$}
Firstly, we prove $(s', h', B') \models \textsf{wpo}(  \psi , [t] := {t'}  , \textit{ok} ) \Rightarrow
(s', h', B') \in \text{WPO} [\![  \psi , [t] := {t'}  , \textit{ok} ]\!]
$.

\subsubsection{Case 1}
We prove the case for $ \psi= t_a \mapsto - * \psi'$ and 
$ \psi_{\rm pure}\models t_a \approx t $ for some $ t_a$, and $ \psi'$.

By the assumption $(s',h',B')$ contains a portion satisfying $t\mapsto
{t'}$, and hence $h'(\inter{t}_{s',B'})=\inter{{t'}}_{s',B'}$.

We have the following.
{ 
\begin{align*}
& 
(s', h', B')  \models    t \mapsto {t'} * \psi' \\
\Rightarrow \ & 
\exists v' . (s', h' [ [\![ t ]\!]_{s',B'} \mapsto v'  ]  , B')  \models
  t \mapsto - * \psi' \\
\Leftrightarrow \ & 
\exists v' . ( s', h' [ [\![ t ]\!]_{s',B'} \mapsto v'  ]  , B' )  \models 
  t_a \mapsto - * \psi' \\
\text{// } & 
 \text{Since $\psi_{\rm pure}\models t_a \approx t $ holds}
 \end{align*}
}
Fix $v'$ and let $h$ be $h' [ [\![ t ]\!]_{s',B'} \mapsto v'  ]$,
and then we have $(s',h,B')\models\psi$. Since
$h'=h[\inter{t}_{s',B'}\mapsto \inter{{t'}}_{s',B'}]$, we have
$((s',h,B'),(s',h',B'))\in\inter{[t]:={t'}}_{\it ok}$.

\subsubsection{Case 2}

We prove the case for $ \psi=\arr(t_a,{t'}_a)*\psi'$ and 
$ \psi_{\rm pure}\models t_a \approx t $
for some $ t_a,{t'}_a$, and $ \psi'$.

Firstly, we prove $(s', h', B') \models \textsf{wpo}(  \psi , [t] := {t'}  , \textit{ok} ) \Rightarrow
(s', h', B') \in \text{WPO} [\![  \psi , [t] := {t'}  , \textit{ok} ]\!]
$.
By the assumption $(s',h',B')$ contains a portion satisfying $t\mapsto
{t'}$, and hence $h'(\inter{t}_{s',B'})=\inter{{t'}}_{s',B'}$.

We have the following.
{ 
\begin{align*}
& 
(s', h', B')  \models   t \mapsto {t'} * \arr (t+1 , {{t'}}_a) * \psi' \\
\Rightarrow \ & 
\exists v' . (s', h' [ [\![ t ]\!]_{s',B'} \mapsto v'  ]  , B')  \models
 \arr (t, t+1)  * \arr (t+1 , {{t'}}_a) * \psi' \\
\Rightarrow \ & 
\exists v' . ( s', h' [ [\![ t ]\!]_{s',B'} \mapsto v'  ]  , B' )  \models   \arr (t_a,  {{t'}}_a) * \psi' \\
\text{// } & 
 \text{Appending arrays}
 \end{align*}
}
Fix $v'$ and let $h$ be $h' [ [\![ t ]\!]_{s',B'} \mapsto v'  ]$,
and then we have $(s',h,B')\models\psi$. Since
$h'=h[\inter{t}_{s',B'}\mapsto \inter{{t'}}_{s',B'}]$, we have
$((s',h,B'),(s',h',B'))\in\inter{[t]:={t'}}_{\it ok}$.

\subsubsection{Case 3}

We prove the case for $ \psi=\arr(t_a,{t'}_a)*\psi'$ and 
$ \psi_{\rm pure}\models t_a<t*t+1<{t'}_a$ for some $ t_a,{t'}_a$, and $ \psi'$.

Firstly, we prove $(s', h', B') \models \textsf{wpo}(  \psi , [t] := {t'}  , \textit{ok} ) \Rightarrow
(s', h', B') \in \text{WPO} [\![  \psi , [t] := {t'}  , \textit{ok} ]\!]
$.
By the assumption $(s',h',B')$ contains a portion satisfying $t\mapsto
{t'}$, and hence $h'(\inter{t}_{s',B'})=\inter{{t'}}_{s',B'}$.

We have the following.
{ 
\begin{align*}
& 
(s', h', B')  \models   \arr (t_a, t) * t \mapsto {t'} * \arr (t+1 , {{t'}}_a) * \psi' \\
\Rightarrow \ & 
\exists v' . (s', h' [ [\![ t ]\!]_{s',B'} \mapsto v'  ]  , B')  \models
\arr (t_a, t) *  \arr (t, t+1)  * \arr (t+1 , {{t'}}_a) * \psi' \\
\Rightarrow \ & 
\exists v' . ( s', h' [ [\![ t ]\!]_{s',B'} \mapsto v'  ]  , B' )  \models   \arr (t_a,  {{t'}}_a) * \psi' \\
\text{// } & 
 \text{Appending arrays}
 \end{align*}
}
Fix $v'$ and let $h$ be $h' [ [\![ t ]\!]_{s',B'} \mapsto v'  ]$,
and then we have $(s',h,B')\models\psi$. Since
$h'=h[\inter{t}_{s',B'}\mapsto \inter{{t'}}_{s',B'}]$, we have
$((s',h,B'),(s',h',B'))\in\inter{[t]:={t'}}_{\it ok}$.

\subsubsection{Case 4}

We prove the case for $ \psi=\arr(t_a,{t'}_a)*\psi'$ and 
$ \psi_{\rm pure}\models t+1 \approx {t'}_a$ for some $ t_a,{t'}_a$, and $ \psi'$.

Firstly, we prove $(s', h', B') \models \textsf{wpo}(  \psi , [t] := {t'}  , \textit{ok} ) \Rightarrow
(s', h', B') \in \text{WPO} [\![  \psi , [t] := {t'}  , \textit{ok} ]\!]
$.
By the assumption $(s',h',B')$ contains a portion satisfying $t\mapsto
{t'}$, and hence $h'(\inter{t}_{s',B'})=\inter{{t'}}_{s',B'}$.

We have the following.
{ 
\begin{align*}
& 
(s', h', B')  \models   \arr (t_a, t) * t \mapsto {t'}  * \psi' \\
\Rightarrow \ & 
\exists v' . (s', h' [ [\![ t ]\!]_{s',B'} \mapsto v'  ]  , B')  \models
\arr (t_a, t) *  \arr (t, t+1)  * \psi' \\
\Rightarrow \ & 
\exists v' . ( s', h' [ [\![ t ]\!]_{s',B'} \mapsto v'  ]  , B' )  \models   \arr (t_a,  {{t'}}_a) * \psi' \\
\text{// } & 
 \text{Appending arrays}
 \end{align*}
}
Fix $v'$ and let $h$ be $h' [ [\![ t ]\!]_{s',B'} \mapsto v'  ]$,
and then we have $(s',h,B')\models\psi$. Since
$h'=h[\inter{t}_{s',B'}\mapsto \inter{{t'}}_{s',B'}]$, we have
$((s',h,B'),(s',h',B'))\in\inter{[t]:={t'}}_{\it ok}$.

\subsubsection{Case 5}

Otherwise, 
$(s', h', B') \models \textsf{wpo}( \psi ,  [t]:={t'} ,
\textit{ok} ) \Rightarrow (s', h', B') \in \text{WPO} [\![ \psi ,
 [t]:={t'}, \textit{ok} ]\!]  $ holds trivially 
since $\textsf{wpo}( \psi ,  [t]:={t'} ,
\textit{ok} )  = \textsf{false}$.

\subsection{$(s', h', B') \in \text{WPO} [\![ \psi ,
[t] := {t'} , \textit{ok} ]\!]  \Rightarrow (s', h', B') \models
\textsf{wpo}( \psi , [t] := {t'} , \textit{ok} ) $} 
Next, we prove the converse: $(s', h', B') \in \text{WPO} [\![ \psi ,
[t] := {t'} , \textit{ok} ]\!]  \Rightarrow (s', h', B') \models
\textsf{wpo}( \psi , [t] := {t'} , \textit{ok} ) $.

\subsubsection{Case 1}
We prove the case for $ \psi= t_a \mapsto - * \psi'$ and 
$ \psi_{\rm pure}\models t_a \approx t $ for some $ t_a$, and $ \psi'$.
By the assumption,
there exists $h$ such that $(s',h,B')\models \psi$,
$b_{B'}(\inter{t}_{s',B'})>0$, and $h'=h[\inter{t}_{s',B'}\mapsto
\inter{{t'}}_{s',B'}]$.  
Since $  \psi=t_a \mapsto -  * \psi'$
and $ \psi_{\rm pure}\models t_a \approx t $, we have
$[\inter{t_a}_{s',B'},\inter{{t}_a}_{s',B'} + 1 )\subseteq\textsf{dom}_{+}(h)$ and
$  \inter{t_a}_{s',B'} = \inter{t}_{s',B'}<\inter{{t}_a}_{s',B'}+1$.
Then, since $  (s', h, B') \models t_a \mapsto -  * \psi'$, we have  
\[  
(s', h', B') \models  t \mapsto {t'} * \psi',
\]
which is ${\sf wpo}(\psi, [t] := {t'}, {\it ok})$.

\subsubsection{Case 2}

We prove the case for $ \psi=\arr(t_a,{t'}_a)*\psi'$ and 
$ \psi_{\rm pure}\models t_a \approx t $
for some $ t_a,{t'}_a$, and $ \psi'$.

 By the assumption,
there exists $h$ such that $(s',h,B')\models \psi$,
$b_{B'}(\inter{t}_{s',B'})>0$, and $h'=h[\inter{t}_{s',B'}\mapsto
\inter{{t'}}_{s',B'}]$.  
Since $  \psi=\arr(t_a,{t'}_a) * \psi'$
and $ \psi_{\rm pure}\models t_a \approx t $, we have
$[\inter{t_a}_{s',B'},\inter{{t'}_a}_{s',B'})\subseteq\textsf{dom}_{+}(h)$ and
$  \inter{t_a}_{s',B'} = \inter{t}_{s',B'}<\inter{{t'}_a}_{s',B'}$.
Then, since $  (s', h, B') \models \arr (t_a, {{t'}}_a) * \psi'$, we have  
\[  
(s', h', B') \models t \mapsto {t'} * \arr (t+1 , {{t'}}_a) * \psi',
\]
which is ${\sf wpo}(\psi, [t] := {t'}, {\it ok})$.

\subsubsection{Case 3}

We prove the case for $ \psi=\arr(t_a,{t'}_a)*\psi'$ and 
$ \psi_{\rm pure}\models t_a<t*t+1<{t'}_a$ for some $ t_a,{t'}_a$, and $ \psi'$.

By the assumption,
there exists $h$ such that $(s',h,B')\models \psi$,
$b_{B'}(\inter{t}_{s',B'})>0$, and $h'=h[\inter{t}_{s',B'}\mapsto
\inter{{t'}}_{s',B'}]$.  
Since $  \psi=\arr(t_a,{t'}_a) * \psi'$
and $ \psi_{\rm pure}\models t_a<t<{t'}_a$, we have
$[\inter{t_a}_{s',B'},\inter{{t'}_a}_{s',B'})\subseteq\textsf{dom}_{+}(h)$ and
$  \inter{t_a}_{s',B'}<\inter{t}_{s',B'}<\inter{{t'}_a}_{s',B'}$.
Then, since $  (s', h, B') \models \arr (t_a, {{t'}}_a) * \psi'$, we have  
\[  
(s', h', B') \models \arr (t_a, t) * t \mapsto {t'} * \arr (t+1 , {{t'}}_a) * \psi',
\]
which is ${\sf wpo}(\psi, [t] := {t'}, {\it ok})$.

\subsubsection{Case 4}

We prove the case for $ \psi=\arr(t_a,{t'}_a)*\psi'$ and 
$ \psi_{\rm pure}\models t+1 \approx {t'}_a$ for some $ t_a,{t'}_a$, and $ \psi'$.

By the assumption,
there exists $h$ such that $(s',h,B')\models \psi$,
$b_{B'}(\inter{t}_{s',B'})>0$, and $h'=h[\inter{t}_{s',B'}\mapsto
\inter{{t'}}_{s',B'}]$.  
Since $  \psi=\arr(t_a,{t'}_a) * \psi'$
and $ \psi_{\rm pure}\models t+1 \approx {t'}_a$, we have
$[\inter{t_a}_{s',B'},\inter{{t'}_a}_{s',B'})\subseteq\textsf{dom}_{+}(h)$ and
$  \inter{t_a}_{s',B'}<\inter{t}_{s',B'}<\inter{{t'}_a}_{s',B'}$.
Then, since $  (s', h, B') \models \arr (t_a, {{t'}}_a) * \psi'$, we have  
\[  
(s', h', B') \models \arr (t_a, t) * t \mapsto {t'}  * \psi',
\]
which is ${\sf wpo}(\psi, [t] := {t'}, {\it ok})$.

\subsubsection{Case 5}

The following is the union of all cases between Case 1 and Case 4.

\begin{align*}
&\text{$\psi$ contains some $\Arr (t_a, {{t'}}_a)_{\alpha}$
such that} \\
&\text{$\psi_{\text{pure}} \models 
t_a \leq t  * t+1 \leq {t'}_a$
}
\end{align*}

The negation of the above formula is as follows:

\begin{align*}
&\text{For any $\Arr (t_a, {{t'}}_a)_{\alpha}$ in $\psi$, } \\
&\text{$\psi_{\text{pure}} \not\models 
t_a \leq t  * t+1 \leq {t'}_a$}
\end{align*}

Suppose that \((s,h,B) \models \psi\).  
By the given condition, we have  
\[
\llbracket t \rrbracket_{s,B} \notin \textsf{dom}_+ (h).
\]  
Hence, there is no \(\sigma\) such that  
\[
((s,h,B) , \sigma ) \in \llbracket  [t] := {t'} \rrbracket_\textit{ok}.
\]  
Therefore, \(\text{WPO} \llbracket \psi ,  [t] := {t'} , \textit{ok} \rrbracket = \emptyset\) holds.

\subsection{Case for \textit{er}}

The case for {\it er} is proven in a similar way to the case of $x:=[t]$.

%% file: main.bbl
\begin{thebibliography}{10}
\providecommand{\url}[1]{\texttt{#1}}
\providecommand{\urlprefix}{URL }
\providecommand{\doi}[1]{https://doi.org/#1}

\bibitem{brotherston2017biabduction}
Brotherston, J., Gorogiannis, N., Kanovich, M.: Biabduction (and related
  problems) in array separation logic. In: Automated Deduction--CADE 26: 26th
  International Conference on Automated Deduction, Gothenburg, Sweden, August
  6--11, 2017, Proceedings. pp. 472--490. Springer (2017)

\bibitem{calcagno2011infer}
Calcagno, C., Distefano, D.: Infer: An automatic program verifier for memory
  safety of c programs. In: NASA Formal Methods Symposium. pp. 459--465.
  Springer (2011)

\bibitem{calcagno2015open}
Calcagno, C., Distefano, D.,  O'Hearn, P. Open-sourcing Facebook Infer: Identify bugs before you ship. Facebook Engineering Blog, June 2015.



\bibitem{cook1978soundness}
Cook, S.A.: Soundness and completeness of an axiom system for program
  verification. SIAM Journal on Computing  \textbf{7}(1),  70--90 (1978)

\bibitem{de2011reverse}
De~Vries, E., Koutavas, V.: Reverse {Hoare} {Logic}. In: International
  Conference on Software Engineering and Formal Methods. pp. 155--171. Springer
  (2011)

\bibitem{hoare1969axiomatic}
Hoare, C.A.R.: An axiomatic basis for computer programming. Communications of
  the ACM  \textbf{12}(10),  576--580 (1969)

\bibitem{holik2022low}
Hol{\'\i}k, L., Peringer, P., Rogalewicz, A., {\v{S}}okov{\'a}, V., Vojnar, T.,
  Zuleger, F.: Low-level bi-abduction. In: 36th European Conference on
  Object-Oriented Programming (ECOOP 2022). Schloss Dagstuhl-Leibniz-Zentrum
  f{\"u}r Informatik (2022)

\bibitem{kimura2021decidability}
Kimura, D., Tatsuta, M.: Decidability for entailments of symbolic heaps with
  arrays. Logical Methods in Computer Science  \textbf{17} (2021)

\bibitem{lee2024relative}
Lee, Y., Nakazawa, K.: Relative completeness of incorrectness separation logic.
  In: Asian Symposium on Programming Languages and Systems. pp. 264--282.
  Springer (2024)

\bibitem{o2019incorrectness}
O'Hearn, P.W.: Incorrectness logic. Proceedings of the ACM on Programming
  Languages  \textbf{4}(POPL),  1--32 (2019)

\bibitem{raad2020local}
Raad, A., Berdine, J., Dang, H.H., Dreyer, D., O'Hearn, P., Villard, J.: Local
  reasoning about the presence of bugs: Incorrectness separation logic. In:
  Computer Aided Verification: 32nd International Conference, CAV 2020, Los
  Angeles, CA, USA, July 21--24, 2020, Proceedings, Part II 32. pp. 225--252.
  Springer (2020)

\bibitem{reynolds2002separation}
Reynolds, J.C.: Separation logic: A logic for shared mutable data structures.
  In: Proceedings 17th Annual IEEE Symposium on Logic in Computer Science. pp.
  55--74. IEEE (2002)

\end{thebibliography}
